\documentclass[a4paper,11pt]{article}
\usepackage{import}
\usepackage{./Preamble2}
\usepackage{fancyvrb}

\preprint{BRX-TH-6728}

\title{
Quantum Bit Threads and the Entropohedron
}

\author[a,b]{Matthew Headrick,}
\author[c]{Sreeman Reddy Kasireddy,}
\author[d]{and Andrew Rolph}

\affiliation[a]{Martin Fisher School of Physics, Brandeis University, Waltham MA 02139, USA}
\affiliation[b]{Institut des Hautes Etudes Scientifiques, 91440 Bures-sur-Yvette, France}
\affiliation[c]{Department of Particle Physics and Astrophysics, Weizmann Institute of Science, Rehovot 7610001, Israel}
\affiliation[d]{Vrije Universiteit Brussel (VUB) and The International Solvay Institutes, Pleinlaan 2, B-1050 Brussels, Belgium}
\emailAdd{headrick@brandeis.edu}
\emailAdd{sreeman@weizmann.ac.il}
\emailAdd{andrew.d.rolph@gmail.com}

\abstract{
We derive several new quantum bit thread prescriptions for holographic entanglement entropy, equivalent for static states to the quantum extremal surface formula. Our new prescriptions come in many varieties: vector field-based or based on measures over bulk curves, dependent or independent of the bulk UV regulator, loose and strict versions of constraints, and more. We also explore how bit threads behave in the presence of entanglement islands and baby universes. Finally, our prescriptions inspire new measures of entanglement that we call entanglement distribution functions, which can be packaged into a convex polytope that we call the entropohedron.
\newline
\newline
A video abstract is available at \url{https://youtu.be/xSyRAXkPpdw}.
}

\makeatletter
\gdef\@fpheader{}
\makeatother
\begin{document}
\nolinenumbers
\maketitle
\flushbottom

\section{Introduction}

Holographic entanglement entropy, a cornerstone of our modern understanding of quantum gravity, reveals how bulk geometry encodes boundary entropy, and vice versa. There are two equivalent sets of prescriptions for holographic entanglement entropy, the older being surface-based, starting with the Ryu-Takayanagi (RT) formula \cite{Ryu:2006bv}.

In the RT formula, given a static (or time reflection-symmetric) slice $\Sigma$ of a classical holographic spacetime, the entanglement entropy of any boundary region $A$ equals a quarter of the area, in Planck units, of the minimal bulk surface homologous to $A$:
\be\label{RT}
S(A) = \min_{r\in\RR_A}\frac{|\eth r|}{4\GN}\,.
\ee
$\RR_A$ is the set of bulk regions $r\subseteq \Sigma$ coincident with $A$ on the boundary, $r\cap\partial\Sigma=A$, and $\eth r$ is the bulk part of the boundary of $r$, $\eth r:=\partial r\setminus\partial\Sigma$. With these two definitions, $\eth r$ is a surface homologous to $A$. We write the formula in terms of regions rather than surfaces because this will be more amenable to generalizations later.

The more recent formulation of holographic entanglement entropy are the bit thread prescriptions.
A classical bit thread is a continuous bulk curve connecting $A$ to its complement $A^c$. In the bit thread reformulation of the RT formula \cite{Freedman:2016zud}, $S(A)$ equals the maximum number of threads connecting $A$ to $A^c$, subject to a bound on their density, $\rho\le1/4\GN$. Classical bit threads can be described mathematically as the field lines of a divergenceless bulk vector field $v$. With the density of threads identified with the norm of $v$, the density bound becomes a norm bound $|v|\le1/4\GN$ and the number of threads connecting $A$ to $A^c$ is identified with the flux of $v$ on $A$, giving the \textit{classical 
flow} formula:
\be\label{maxflow}
S(A) = \max_{v}\int_An\cdot v\quad\text{subject to:}\quad|v|\le\frac1{4\GN}\,,\quad\nabla\cdot v=0\,,
\ee
where $n$ is the inward-directed unit normal vector. The equivalence between \eqref{maxflow} and the original RT formula \eqref{RT} is guaranteed by the Riemannian max flow-min cut theorem \cite{federer1974real, strang1983maximal, sullivan1990crystalline, nozawa1990max, Headrick:2017ucz}. With bit threads, the connection between boundary entanglement and bulk geometry is even more manifest than in the RT formula.  Bit thread prescriptions have both conceptual advantages over surface-based prescriptions, such as making the physical meaning of entropy inequalities manifest, and technical advantages, as they are expressed as convex optimisation problems.

Quantum entanglement between bulk fields also contributes to boundary entanglement entropy. Within the set of surface-based prescriptions, this is accounted for by the \emph{quantum extremal surface} (QES) formula, which replaces the area in the RT formula by the generalized entropy:\footnote{Although ``quantum minimal surface'' would be a more appropriate name here, in order to agree with the general literature, we will continue to call \eqref{QES} the QES formula. The covariant QES formula involves minimizing among surfaces that extremize the generalized entropy. The fact that the covariant QES reduces in the time-symmetric case to a minimization on the time-symmetric slice follows from the quantum focusing conjecture by the same reasoning as the reduction from the HRT to the RT formula: the quantum minimal surface $\gamma_1$ on that slice is necessarily extremal by the time symmetry; but if the minimal QES $\gamma_2$ is off the slice, then projecting it by light rays onto the slice produces a surface $\gamma_3$ with $S_{\rm gen}(\gamma_3)\le S_{\rm gen}(\gamma_2)< S_{\rm gen}(\gamma_1)$, a contradiction.}
\be\label{QES}
S(A) = \min_{r\in\RR_A}S_{\rm gen}(r)\,,
\ee
where the generalized entropy is
\be\label{Sgendef}
S_{\rm gen}(r) = \frac{|\eth r|}{4\GN}+S_{\rm b}(r)\,.
\ee
The QES formula has led to important insights in quantum gravity, such as, in the context of black hole evaporation, entropy curves consistent with unitarity, and the presence of entanglement islands \cite{Almheiri:2019psf,Penington:2019npb,Almheiri_2021}.

There are a number of subtle and interesting issues related to both the definition of the generalized entropy (such as the contribution of the gravitational field itself \cite{Colin-Ellerin:2025dgq} and additional terms that depend on the details and dimensionality of the bulk theory), and the QES formula (such as conditions on the state of the bulk fields for its validity \cite{Akers:2020pmf}). Nonetheless, in this paper, we will focus on the QES formula in the form~\eqref{QES}, as it captures the essence of the bulk entropy contribution and is rich enough to be worth studying on its own.

A bit thread reformulation of the QES formula~\eqref{QES} was given in the paper \cite{Rolph:2021hgz}.%
\footnote{Other bit-thread formulas including quantum corrections to $S(A)$ have been proposed. For example, the paper \cite{Agon:2021tia} gave a flow proposal that captures corrections to first order in $\GN$ (whereas the QES formula and \cite{Rolph:2021hgz} account for correction to all order in $\GN$). And \cite{Du:2024xoz} gave a flow reformulation of the ``generalised entanglement wedge" proposal of \cite{Bousso:2022hlz}.} Specifically, the definition of a flow was relaxed to allow for a non-zero divergence controlled by the bulk entropy:  
\be\label{qmaxflow1}
S(A) = \max_v\int_An\cdot v\quad\text{subject to:}\quad|v|\le\frac1{4\GN}\,,\quad\forall r\in\RR_A\,,\,\,-\int_r\nabla\cdot v\le S_{\rm b}(r)\,.
\ee
For reasons that we will shortly make clear, we call this the \emph{loose quantum flow} prescription.
In this quantum prescription, unlike the classical case, bit threads can start and end at points in the bulk, with the number of threads that can end in a given bulk region bounded by the region's entropy. This allows for more flux to pass through $A$ than can pass through any surface homologous to $A$.
The paper \cite{Rolph:2021hgz} proved the equivalence of \eqref{qmaxflow1} and \eqref{QES}, and explored its properties in general as well as in several examples. Note that the data required to evaluate \eqref{qmaxflow1} is precisely the same as that required to evaluate \eqref{QES}, namely the metric and the entropy of every bulk region coincident to $A$. It's also worth noting that the proof of equivalence relied essentially on the strong subadditivity property of the bulk entropies.

\begin{figure}
    \centering
    \includegraphics[width=0.75\linewidth]{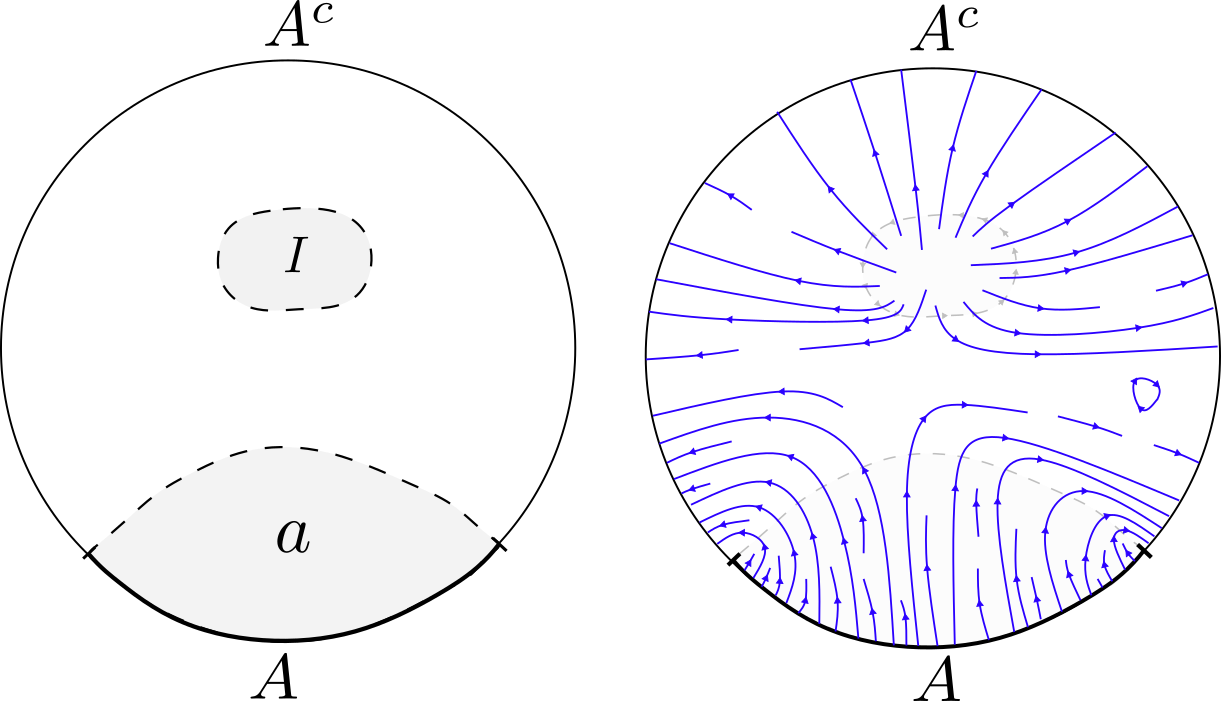}
    \caption{Contrasting the surface-based QES formula and a quantum bit thread prescription in an island setup. Both disks represent the same geometry, a time slice of an asymptotically AdS spacetime, and the same highly entangled bulk state. Left: the QES formula calculates $S(A)$ by minimising over bulk surfaces, $\displaystyle S(A) = \min_{\eth r \sim A } \Sgen  (r) = \Sgen  (a \cup I)$.
    Right: an optimal flow configuration for the quantum bit thread prescription~\eqref{qmaxflow1}. We maximise the boundary flux, $S(A) = \max_v \int_A n\cdot v$, with the $v$-constraints given in~\eqref{qmaxflow1}. The blue curves are the bit threads, which in this prescription are the integral curves of $v$, and the threads are maximally packed on $\eth a \cup \eth I$; the boundary of the island is a bit thread bottleneck.}
    \label{fig:QuantumThreadsIsland}
\end{figure}

\begin{figure}
    \centering
    \includegraphics[width=0.65\linewidth]{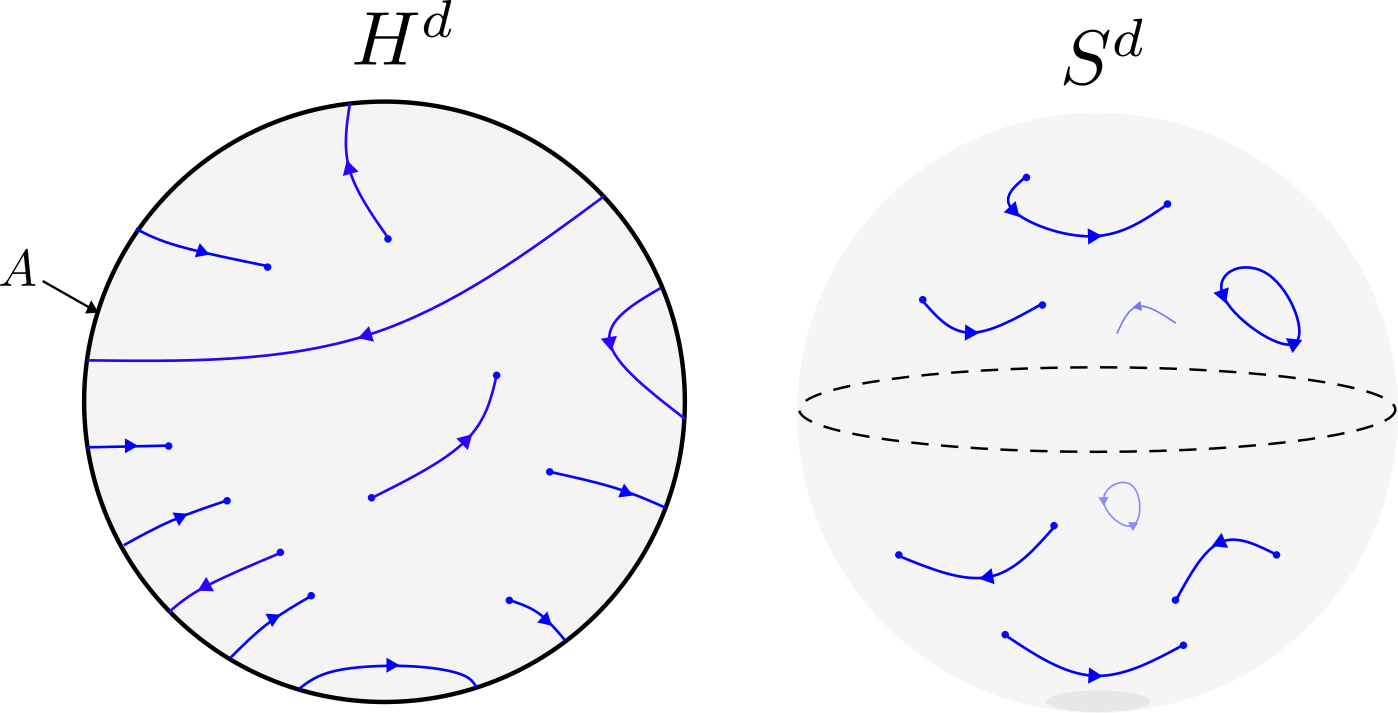}
    \caption{Quantum bit threads and closed universes I: the AdS bulk timeslice, the hyperbolic disk, is entangled with a closed universe $S^d$. The constraint in~\eqref{qmaxflow2} allows threads to end in the AdS bulk, but then also forces them to reappear in the closed universe, and then return to the AdS bulk, because it requires $\int_{\Sigma} \nabla \cdot v = 0$ for a pure boundary state, and we have taken $A = \del \Sigma$. As depicted, the net maximal boundary flux is $\int_A n\cdot v = S(A) = 0$. }
    \label{fig:ClosedUniverse1}
\end{figure}

\begin{figure}
    \centering
    \includegraphics[width=0.9\linewidth]{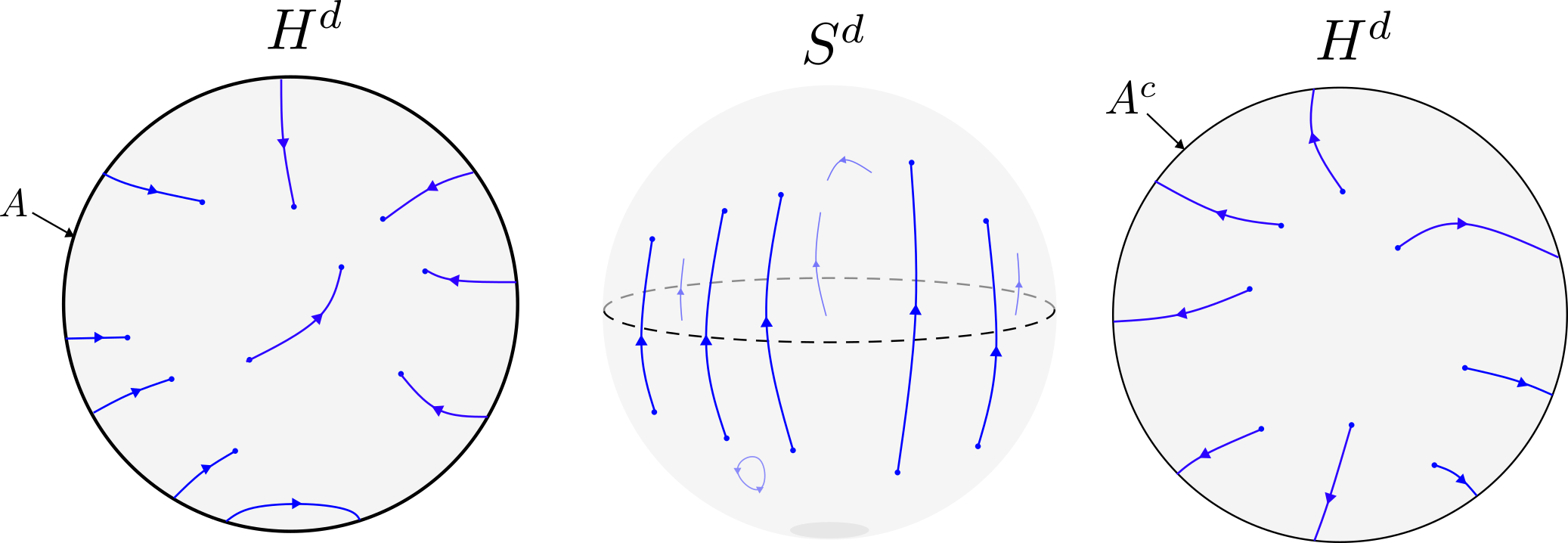}
    \caption{Quantum bit threads and closed universes II: the left AdS bulk is entangled with the southern hemisphere of the $S^d$, and the northern hemisphere is entangled with the right AdS bulk. The constraint in~\eqref{qmaxflow2} allows threads to end in the left bulk, but they have to reappear in the sphere's southern hemisphere. They can then end in the northern hemisphere, but they have to reappear in the right bulk. In this example, as depicted, quantum bit threads in optimal configurations have to pass through the closed universe when going from $A$ to $A^c$. }
    \label{fig:ClosedUniverse2}
\end{figure}

In this paper, we will extend the work of \cite{Rolph:2021hgz} with more bit thread reformulations of the QES formula \eqref{QES}. First, we will show that it is possible to impose a much stronger constraint on the divergence of $v$ without changing the maximum flux, giving the \textit{strict quantum flow} prescription:
\be\label{qmaxflow2}
S(A) = \max_v\int_An\cdot v\quad\text{subject to:}\quad|v|\le\frac1{4\GN}\,,\qquad\forall r\in\RR\,,\quad\left|\int_r\nabla\cdot v\right|\le S_{\rm b}(r)\,.
\ee
The constraint in \eqref{qmaxflow2} on the divergence of $v$ differs from that in \eqref{qmaxflow1} in two ways: first, it is applied to \emph{every} bulk region $r$, with $\RR$ defined as the set of all bulk regions, without any condition on their intersection with the boundary; second, the integrated divergence is bounded above by $S_{\rm b}(r)$ as well as below by $-S_{\rm b}(r)$. 
Both constraints are non-local and bound the flux into bulk subregions; the larger the bulk entropy of the region, the larger the number of bit threads that can end there. Unlike the loose constraint, but like classical flows, the strict constraint is independent of the choice of boundary region $A$. We will show that it is automatically obeyed in doubly holographic systems; for this purpose, we will prove a new theorem giving necessary and sufficient conditions for a function on the boundary of a Riemannian manifold to be the boundary flux of a classical flow.  More generally, the strict divergence constraint leads to an interesting generalization of the notion of an entanglement contour \cite{Vidal:2014aal}, which we call an \textit{entanglement distribution function (EDF)}. For a finite set of $N$ parties, the set of all EDFs for a given state is a convex polyhedron in $\R^N$ which we call the \emph{entropohedron}. We will determine its properties and study how it reflects the entanglement structure of the state. We also consider the generalization of strict quantum flows to multiflows.

From the bit thread perspective, \emph{entanglement islands} are highly entangled bulk regions where so many threads are forced to reappear because of the divergence constraint that the threads become maximally packed on a surface disconnected from the asymptotic boundary, and this bottleneck is the boundary of the island; see Fig.\ \ref{fig:QuantumThreadsIsland}. In Fig.\ \ref{fig:ClosedUniverse1}, we depict an example where the AdS bulk is entangled with a closed universe, and the quantum bit threads can jump across, but they have to return because the constraint in~\eqref{qmaxflow2} requires that $\int_{\Sigma} \nabla \cdot v = 0$ for a pure boundary state. For a given setup, there are generically many optimal flow configurations, and $v=0$ is the simplest for the example shown in Fig.\ \ref{fig:ClosedUniverse1}. In Fig.\ \ref{fig:ClosedUniverse2}, inspired by \cite{Antonini:2023hdh}, we depict an example where optimal quantum thread configurations in \eqref{qmaxflow2} must pass through an entangled closed universe on their way between the two AdS boundaries. 

Generalized entropy is independent of the value of UV regulator, and the separation of generalized entropy into classical and quantum pieces is, from a microscopic viewpoint, artificial, since both the Newton constant and the bulk entropy are regulator-dependent. In the quantum bit thread prescriptions discussed so far, $\GN$ appears in the density bound and $S_{\rm b}(r)$ appears in the divergence constraint, so the set of allowed thread configurations is regulator-dependent. However, we will see how they conspire to produce a macroscopic thread configuration that is regulator-independent. As an extreme example, we will investigate what happens to the threads in the induced-gravity limit $\GN\to\infty$ in which the area term disappears entirely.

On the topic of regulator-dependence, we will derive another formulation of quantum bit threads, in which the divergence and density constraints are combined into a single, \emph{cutoff-independent} bound. This is the \textit{strict, cutoff-independent quantum bit thread flow prescription:}
\bne S(A) = \sup_v \int_A v \qquad \text{subject to:} \quad \forall\, r\in\RR\,,\quad \int_{\eth r} |v| + \left| \int_r \nabla \cdot v \right |\leq \Sgen (r)\,. \label{eq:sciqb} \ene 
Like the cutoff-dependent quantum bit thread prescriptions~\eqref{qmaxflow1} and~\eqref{qmaxflow2}, this formulation has both strict and loose versions. 
The bound in~\eqref{eq:sciqb} depends on the generalized entropy of bulk regions, so it is regulator-independent. One thing that we will show is that every allowed thread configuration in~\eqref{qmaxflow1} and~\eqref{qmaxflow2} is an allowed thread configuration of \eqref{eq:sciqb}.

The last quantum bit thread flow prescription that we will mention here is both cutoff-independent and~\emph{divergenceless}:
\bne S(A) = \sup_v \int_A v \qquad \text{subject to:} \quad \nabla \cdot v = 0\,, \quad \forall\,r\in\RR_A \,,\quad \int_{\eth r} |v|  \leq \Sgen(r)\,.\label{eq:divloocu_intro} \ene
As the flow field is divergenceless, this prescription is, in a sense, the closest to the original classical formulation of~\cite{Freedman:2016zud}.

An alternative mathematical formulation of bit threads, instead of flows, is in terms of so-called \emph{thread distributions} \cite{Headrick:2022nbe}. In the classical case, a thread is a continuous bulk curve connecting boundary points, and a thread distribution is a measure on the set of threads subject to a density bound. Inspired by the quantum flows, we propose a quantum thread distribution formula, in which the quantum threads are allowed to jump from one point to another of the bulk, subject to the constraint that the number of threads leaving or entering any bulk region is bounded above by its entropy. We give some evidence, in particular based on double holography, that this thread distribution formula is equivalent to the QES formula. However, several statements concerning quantum bit thread distributions are left as conjectures.

In the form of the QES prescription given in \eqref{QES}, we are accounting for bulk quantum corrections but, implicitly, staying in a static or time reflection-symmetric setting --- as we will throughout this paper --- where the extremization part of the QES formula simplifies to a minimization over subregions of a single time slice.
In a forthcoming paper \cite{CQBT}, we will drop the static assumption, and derive \textit{covariant} quantum bit thread prescriptions, which synthesize and extend the quantum bit threads with the covariant classical bit threads of \cite{Headrick:2022nbe}, which are equivalent to the covariant Hubeny-Rangamani-Takayanagi (HRT) formula. 

The paper is organized as follows: In section \ref{sec:strictflows} we define and discuss strict quantum flows and entanglement distribution functions. We define and investigate our cutoff-independent formulations in section \ref{sec:cutin}, and our thread distribution formulations in section \ref{sec:qdist}. In section \ref{sec:entropohedron}, we define a new object called the entropohedron, determine its properties, and explain how it captures the entanglement structure of a given state. In appendix \ref{sec:gen_entropy}, we discuss generalised entropy, in particular its regulator independence and its value in the zero volume limit.

\subsection{Setup \& notation}

Throughout this paper, we fix a static or time-symmetric slice $\Sigma$ of a holographic spacetime, governed by Einstein gravity coupled to some set of quantum fields, whose state (which may be pure or mixed) we also fix. We assume that $\Sigma$ is complete, in the sense that its boundary $\partial\Sigma$ consists only of the conformal boundary, on which the dual field theory lives. We assume that the theory and state are such that the quantum extremal surface formula \eqref{QES} holds. For the purposes of this paper, \eqref{QES} may be taken as the definition of $S(A)$.

We use $A, B, C$ and so on to refer to boundary spatial regions. If such a region $A$ has a boundary (a non-empty entangling surface), then we assume that a regulator has been put in place to ensure that bulk surfaces homologous to $A$ have finite area. This is a UV cutoff from the boundary viewpoint.

In order to make the bulk entropy $S_{\rm b}(r)$ well-defined, we also need a bulk UV regulator. The associated length $\epsilon$ must be larger than the Planck length but smaller than the AdS length and any other geometrical invariant length scales of $\Sigma$. We consider only bulk regions that are larger than $\epsilon$ in all dimensions. We will not be precise about this restriction on bulk regions since, as we will show, the constraints associated with the smallest regions are not active. Except in subsection \ref{sec:changingcutoff}, where we investigate the effect of changing $\epsilon$, we will consider it to be fixed.

We will leave metric-derived measure factors in integrals ($\sqrt g$ etc.) implicit. For example, by $\int_An\cdot v$ we mean $\int_A\sqrt{h}\,n\cdot v$, where $h$ is the determinant of the induced metric on $A$. For orientations, for any bulk region $r$, including $\Sigma$, we take the unit normal vector $n$ to be outward-pointing on the $\eth r = \del r \backslash \del \Sigma$ component of the boundary, and inward-pointing on the $\del r \cap \del \Sigma$ component. In particular, $n$ is inward-pointing on $A \subset \del \Sigma$.

The bulk slice $\Sigma$ may include an end-of-the-world brane $Q$. Although topologically $Q$ is part of the boundary of $\Sigma$, physically it is part of the bulk, hence we reserve the term ``boundary'' and the symbol $\partial\Sigma$ to refer to the conformal boundary where the dual field theory lives. $Q$ may host quantum fields with their own entropy, and these will contribute to the bulk entropy $S_{\rm b}(r)$ appearing in the QES formula whenever $r$ intersects $Q$. We also do not include $Q$ in $\partial r$ or $\eth r$. To make the divergence theorem work, we then formally define the divergence $\nabla\cdot v$ to include a delta function on $Q$ proportional to $n\cdot v$, where $n$ is the inward-directed unit normal to $Q$. With these definitions in place, we will leave the possible presence of an end-of-the-world brane implicit in the rest of the paper.

For convenient reference, we collect our main notation here: 
\begin{align}
\Sigma&=\text{time-symmetric bulk slice }\\
x&=\text{point in $\Sigma$} \\
\partial\Sigma&=\text{slice of conformal boundary where dual CFT lives} \\
A&=\text{region of $\partial\Sigma$} \\
A^c&:=\partial\Sigma\setminus A  
\\
\A&:=\text{set of regions of $\partial\Sigma$} \\
n&:=\text{inward-directed unit normal vector on $\partial\Sigma$} \\
\epsilon&=\text{bulk UV cutoff length} \\
r&= \text{bulk region (codimension-0 subset of $\Sigma$, larger than $\epsilon$ in all dimensions)} \\
S_{\rm b}(r)&:=\text{entropy of bulk fields in $r$} \\
\eth r&:=\partial r\setminus\partial\Sigma \\
|\eth r|&:=\text{area of $\eth r$}\\
\RR&:=\text{set of bulk regions} \\
\RR_A&:=\{r\in\RR:r\cap\partial\Sigma=A\} \\
S(A)&:=\min_{r\in\RR_A}\left(\frac{|\eth r|}{4\GN}+S_{\rm b}(r)\right) 
\end{align}

\section{Strict quantum flows}
\label{sec:strictflows}

In this section, we will use a doubly holographic setup to motivate the strict quantum flows defined by the right-hand side of \eqref{qmaxflow2}. We will then study these flows, showing that they are equivalent to the QES formula and looking at a few simple examples. Along the way, a concept we call \emph{entanglement distribution function} will play an important role, so we will spend some time studying it. Some of the proofs are relegated to the last subsection to avoid interrupting the narrative. 

\subsection{Flows in double holography}
\label{sec:flowsdouble}

A doubly holographic theory is one where the bulk matter fields are themselves holographic, described by a higher-dimensional second bulk~\cite{Karch:2000gx,Karch:2000ct,Emparan:2002px,Almheiri:2019hni}.
Typically, it is assumed that the theory on the second bulk may be treated classically, hence that the RT formula applies for computing entropies of regions in the first bulk. The geometric nature of these bulk entropies makes the computation of entropies in doubly holographic theories more tractable than in generic holographic theories, and is a source of useful intuition. We will use double holography to motivate the definition of strict quantum flows.

\begin{figure}
    \centering
    \includegraphics[width=0.45\linewidth]{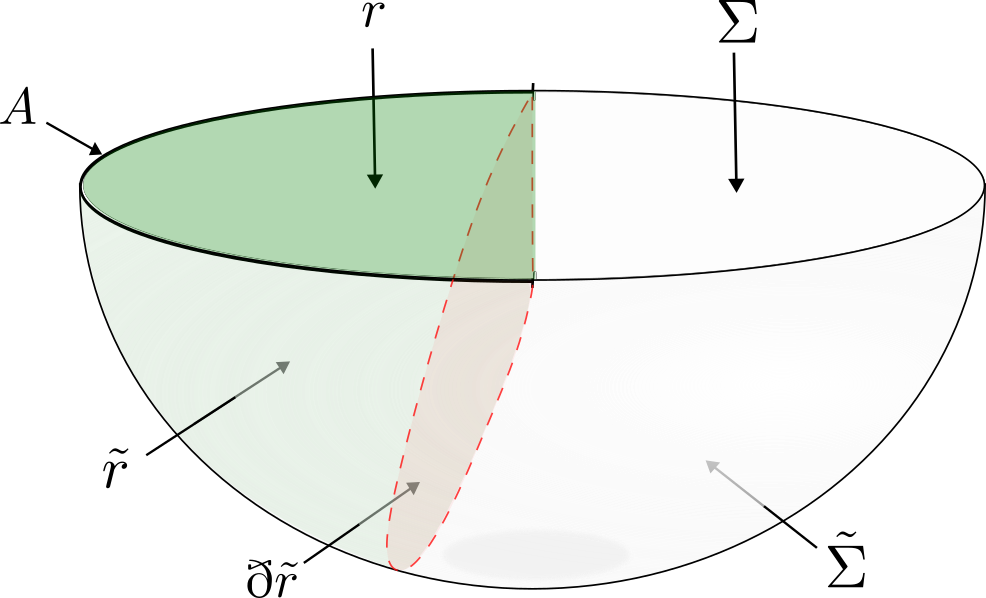}
    \caption{A doubly holographic setup. This figure illustrates the notation for geometric regions used in section~\ref{sec:flowsdouble}. The entropy of region $r\subseteq \Sigma$ is proportional to the area of the RT surface $\eth \tilde r$ in the higher-dimensional second bulk $\tilde \Sigma$. 
    }
    \label{fig:doubleh}
\end{figure}

We will call the second bulk $\tilde\Sigma$; the first bulk $\Sigma$ is thus part of $\partial\tilde\Sigma$, and the manifold $\partial\Sigma$ where the original field theory lives is a codimension-2 corner of $\tilde\Sigma$. See Fig.\ \ref{fig:doubleh}. (The boundary of $\tilde\Sigma$ may also include another region representing an auxiliary system with which the fields on $\Sigma$ may be entangled.) 

By the assumption that the theory on $\tilde\Sigma$ is classical, for any region $r\subseteq\Sigma$, the bulk entropy $S_{\rm b}(r)$ can be computed by applying the RT formula:
\be
S_{\rm b}(r) = \min_{\tilde r\in{\tilde\RR}_r}\frac{|\eth\tilde r|}{4\widetilde\GN}\,,
\ee
where: $\tilde\RR_r$ is the set of regions $\tilde r\subseteq\tilde\Sigma$ such that $\tilde r\cap\partial\tilde\Sigma=r$; $\eth\tilde r:=\partial\tilde r\setminus\partial\tilde\Sigma$; and $\widetilde\GN$ is the Newton constant in $\tilde\Sigma$. The QES formula \eqref{QES} can then be written as follows:
\be\label{RTdouble}
S(A) = \min_{\tilde r\in\tilde\RR_A}\left(\frac{|\eth(\tilde r\cap\Sigma)|}{4\GN}+\frac{|\eth\tilde r|}{4\widetilde\GN}\right).
\ee
Here $\tilde\RR_A$ is the set of regions $\tilde r\subseteq\tilde\Sigma$ such that $\tilde r\cap\partial\tilde\Sigma\subseteq\Sigma$ and $\tilde r\cap\partial\Sigma=A$.

\begin{figure}
    \centering
    \includegraphics[width=0.85\linewidth]{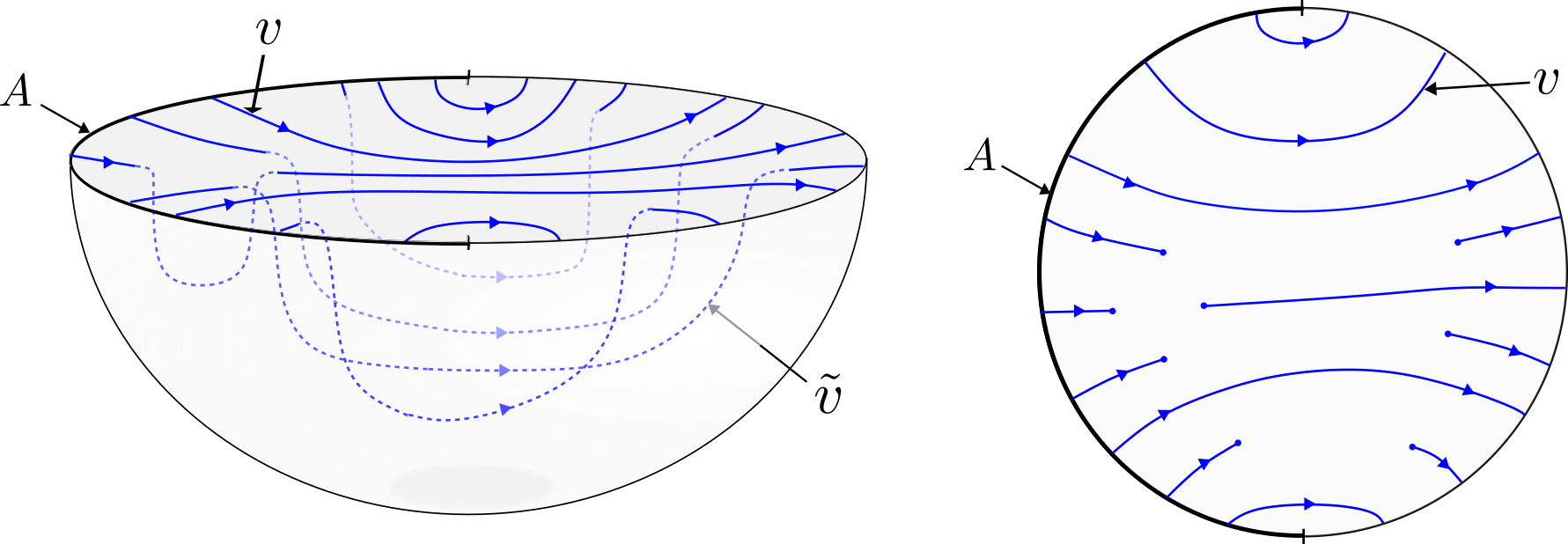}
    \caption{Quantum bit threads in double holography. Left: a maximal flow configuration from the highest-dimensional bulk perspective. The bit threads are classical (divergenceless) and their field-lines move between the two bulks $\Sigma$ and $\tilde{\Sigma}$. $\eth \tilde{r}_{\text{QES}}$, the classical RT surface, is the bottleneck to the double flow $(v,\tilde v)$. Right: the same flow configuration from the intermediate bulk perspective. Flow lines for $(v,\tilde v)$ that are continuous from the perspective of $\tilde{\Sigma}$ become discontinuous flow lines for $v$ from the perspective of $\Sigma$.
    }
    \label{fig:doubleh2}
\end{figure}

A \emph{classical flow} is a vector field $v$ obeying
\be
|v|\le\frac1{4\GN}\,,\qquad\nabla\cdot v=0\,.
\ee
Relaxing and dualizing the minimization problem \eqref{RTdouble}, as one does in the usual max flow-min cut theorem, yields a \emph{double flow}, which is a pair $(v,\tilde v)$ where $\tilde v$ is a classical flow on $\tilde\Sigma$ and $v$ is a vector field on $\Sigma$ obeying
\be
|v|\le\frac1{4\GN}\,,\qquad
\nabla\cdot v+\tilde n\cdot\tilde v=0\,,
\ee
where $\tilde n$ is the inward-directed unit normal to $\Sigma$ in $\tilde\Sigma$. We have:
\be\label{doublemaxflow}
S(A) = \max\int_An\cdot v\quad\text{over double flows $(v,\tilde v)$}\,.
\ee
(See theorem \ref{thm:doublemaxflow} for the proof of \eqref{doublemaxflow}.) The physics here is fairly intuitive: the field lines can pass through either the first or second bulk, respecting the local density bound, and can move back and forth between the first bulk and the second one. See Fig.\ \ref{fig:doubleh2}.

Our goal is to use \eqref{doublemaxflow} to inspire a flow formula for a holographic theory with general bulk matter. To do that, we need to rewrite \eqref{doublemaxflow} in terms of the bulk entropies $S_{\rm b}(r)$. What is the relation between the (classical) flow $\tilde v$ in the second bulk $\tilde\Sigma$ and the entropies $S_{\rm b}(r)$ in the first bulk $\Sigma$? By the max flow-min cut theorem, $S_{\rm b}(r)$ equals the maximal flux over $r$ for any flow $\tilde v$. Since the set of flows is invariant under $\tilde v\to-\tilde v$, we have
\be
\left|\int_r\tilde n\cdot\tilde v\right|\le S_{\rm b}(r)\,,
\ee
which implies
\be\label{nablavEDF}
\left|\int_r\nabla\cdot v\right|\le S_{\rm b}(r)\,.
\ee
In fact, the converse holds: if $v$ obeys \eqref{nablavEDF} for all regions $r\subseteq\Sigma$, then there exists a classical flow $\tilde v$ on $\tilde\Sigma$ such that
\be
\nabla\cdot v+\tilde n\cdot\tilde v=0\,.
\ee
See theorem \ref{thm:classicalflowEDF} for the general statement and proof. So maximizing $\int_An\cdot v$ over double flows $(v,\tilde v)$ is equivalent to maximizing it over vector fields $v$ on $\Sigma$ obeying
\be\label{strictdef}
|v|\le\frac1{4\GN}\,;\qquad
\forall\,r\in\RR\,,\quad\left|\int_r\nabla\cdot v\right|\le S_{\rm b}(r)\,,
\ee
where $\RR$ is the set of regions $r\subseteq\Sigma$.

Having thus succeeded in rewriting the max flow formula in double holography in a form that does not make direct reference to the second bulk, we can now ask whether the resulting formula is correct for a general bulk matter theory. We will turn to this question in the next subsection.

\subsection{Definition \& basic properties of strict flows}
\label{sec:strictflowsdef}

We now consider a general holographic theory, in which $S(A)$ is computed by the QES formula \eqref{QES}, with the bulk entropy $S_{\rm b}$ being a function on the set $\RR$ of bulk regions that obeys positivity and strong subadditivity but is otherwise general. Motivated by what we found in doubly holographic theories in the previous subsection, we define a \emph{strict quantum flow} as a vector field $v$ on $\Sigma$ obeying \eqref{strictdef}. 
This definition should be compared with that of the quantum flows defined in \cite{Rolph:2021hgz}, which we will call ``loose''.\footnote{One can also consider intermediate definitions, with $\RR_A$ and the absolute value, or with $\RR$ and no absolute value. Since loose and strict flows both have maximum flux equal to $S(A)$, these intermediate flows must also have this property.} A \emph{loose quantum flow} for a boundary region $A$ is a vector field $v$ on $\Sigma$ such that
\be
|v|\le\frac1{4\GN}\,,\qquad
\forall r\in\RR_A\,,\quad-\int_r\nabla\cdot v\le S_{\rm b}(r)\,,
\ee
where $\RR_A$ is the set of regions $r\subseteq\Sigma$ such that $r\cap\partial\Sigma=A$. The differences are, first, the divergence bound is applied only to regions in $\RR_A$, and second, the integrated divergence is bounded only in one direction. In \cite{Rolph:2021hgz}, it was shown that the maximum flux through $A$ for a loose flow equals $S(A)$.  
The same holds for strict flows:
\be\label{S(A)strict}
S(A) = \max\int_A n\cdot v\text{ over strict quantum flows $v$}\,.
\ee
See theorem \ref{thm:S(A)strict} for the proof of \eqref{S(A)strict}. In fact, we will give two proofs, both using Lagrange dualization: one dualizing from \eqref{S(A)strict} to the QES prescription, and one doing the reverse. Flow-to-surface dualizations are useful when one has a bit thread proposal in hand, like~\eqref{S(A)strict}, that one wishes to prove is equivalent to the corresponding surface-based prescription, while surface-to-flow dualizations are useful when, such as for the covariant QES prescription, there is not an obvious bit thread proposal, and one wishes to derive it starting from the surface-based prescription. 
So, the pair of proofs illustrate how the dualization works in both directions, and the surface-to-flow dualization will be useful for the covariant setup~\cite{CQBT}.
Also, in theorem~\ref{thm:looseS}, we will also give a new proof that the maximum flux over loose quantum flows equals $S(A)$. This proof is complementary to the proof given in~\cite{Rolph:2021hgz}, because the Lagrange dualization is in the reverse direction, and it is complementary to our second proof of theorem~\ref{thm:S(A)strict} because it smears out bulk regions a different way, using level sets rather than measures.
Notably, all of these proofs, as well as those of theorems \ref{thm:strictsubcontour} and \ref{thm:strictnesting}, make use of the assumption that $S_{\rm b}$ obeys strong subadditivity (SSA). 

Given \eqref{S(A)strict}, the set of maximal loose flows contains some strict flows. Since they both calculate $S(A)$, one can choose to work with either loose or strict flows. The loose flows have the virtue that they require only the same data as the QES itself, namely the metric and the entropy of every region in $\RR_A$. The strict flows require the entropy of every bulk region; the flip side is that they therefore encode much more information about the quantum state of the bulk fields. Furthermore, the definition doesn't depend on the region $A$; as with classical flows, one can speak of strict flows without specifying the region.

Because of the flow's divergence constraint, in particular the absolute value, the strict flow satisfies $\int_\Sigma \nabla \cdot v = 0$, assuming that the bulk state is pure, which means the net number of flow lines that end in any bulk region $r$ equals the net number that start in $r^c$. This is important for the pictorial interpretation that the quantum bit threads jump across the QES. In contrast, with the loose divergence constraint, max-flow lines end in the entanglement wedge of $A$%
\footnote{
Throughout this paper, what we mean by the entanglement wedge of $A$ is the intersection of the time-reflection-symmetric bulk time slice with the common definition of the entanglement wedge of $A$ (the domain of dependence of any codim-1 spacelike surface whose boundary is $A \cup \text{QES}_A$). Our entanglement wedge is a codim-0 region with respect to the bulk time slice.} 
but do not have to reappear in the complement region. 

In general, there is no advantage to have sources in the entanglement wedge of $A$, because bit thread flow lines attached to them cannot increase the flux through $A$. So, does there always exist a member in the set of optimal flow configurations such that the entanglement wedge of $A$ contains only sinks and no sources? No, and Fig.~\ref{fig:ClosedUniverse2} is a counterexample: the threads have to reappear in the southern hemisphere of the sphere, which is part of A's entanglement wedge. This is not just a peculiarity of disconnected bulks because, even if we connected the southern hemisphere to the left disk in Fig~\ref{fig:ClosedUniverse2} with a narrow wormhole, some threads would still have to be sourced in the sphere's southern hemisphere. 

The constraint on $\nabla\cdot v$ in the definition of a strict flow is quite interesting. We will now take a slight detour to study this constraint by defining something called an \emph{entanglement distribution function}. The results we find will be useful when we pursue our study of strict quantum flows.

\subsection{Entanglement distribution functions}
\label{sec:subcontours}

The setting for this subsection is a quantum field theory in a given state on a spatial manifold $M$, and its entanglement entropies $S(A)$ as a function of the spatial region $A\subseteq M$. When we return to holographic systems, the constructs we study here will be applicable to both the bulk and the boundary. More generally, the setting could be any multipartite quantum system in a fixed state, in which case the integrals below should be replaced by sums. We study this kind of system extensively in section \ref{sec:entropohedron}. Here, we will use the continuum language as that corresponds to our intended applications. However, the reader should imagine that the manifold $M$ has been discretized in some way at an ultraviolet scale $\epsilon$, so that every region $A$ we consider has size at least $\epsilon$ in every dimension and therefore has a well-defined entropy. We denote the set of such regions $\A$. Since this effectively turns the continuous system into a discrete one, for the proofs of some of the statements below, we will refer to the proofs of the corresponding theorems in the discrete setting given in section \ref{sec:entropohedron}.

We define an \emph{entanglement distribution function (EDF)}\footnote{EDFs are closely related to the notions of \emph{entanglement contour} and \emph{partial entanglement entropy} \cite{Vidal:2014aal, Wen:2019iyq, Kudler-Flam:2019oru, Rolph:2021nan}. Given a region $A$, these are non-negative function $s_A$ on $A$ obeying $s_A\ge0$, $\int_As_A=S(A)$, and several other conditions. In our case, since we are imposing \eqref{subcontourdef} for \emph{all} regions, we must allow $f$ to be negative; otherwise, for example, if the full state is pure, the only solution would be $f=0$. However, as we will see, for any fixed region $A$, an EDF can be found that is non-negative on $A$ and integrates to $S(A)$. For a discussion of entanglement contours, see \cite{Rolph:2021nan}. Other related notions include various quantities that have been given the name ``entanglement density'' \cite{Nozaki:2013wia, Bhattacharya:2014vja, Gushterov:2017vnr, Jeong:2022zea}.} as a function $f$ on $M$ such that, for all $A\in\A$,
\be\label{subcontourdef}
\left|\int_Af\right|\le S(A)\,.
\ee
Knowing the set of all EDFs is equivalent to knowing $S(A)$ for all $A\in\A$. The set of EDFs is convex, non-empty (since it contains the zero function), and symmetric under $f\to-f$. The extremal points of this set --- EDFs that are not convex combinations of other EDFs --- in some sense carry maximal information about the entropies. These are the functions that saturate \eqref{subcontourdef} on a maximal set of regions. This raises the question of what sorts of regions, or sets of regions, can be saturated.

First, consider a single region $A$. If we divide $A$ into subregions $A_1$, $A_2$ (so $A=A_1A_2$) then we have
\be
\int_Af=\int_{A_1}f+\int_{A_2}f\le S(A_1)+S(A_2)\,.
\ee
If it happened that $S(A_1)+S(A_2)<S(A)$, then we would not be able to saturate on $A$. Thankfully, the subadditivity property tells us that this is impossible, so at least this particular obstruction does not occur. Indeed, we will show using both subadditivity and strong subadditivity (SSA) that it \emph{is} possible to saturate on any given region.

In the same example, it is clear that if $\int_{A_1}f=S(A_1)$ and $\int_{A_2}f=S(A_2)$ then necessarily $\int_{A}f=S(A)$ as well, and moreover this can only happen if $I(A_1:A_2)=0$. For a QFT in a finite-energy state on a connected manifold, we do not expect the mutual information between spatial regions ever to vanish strictly. Similar reasoning shows that we can saturate on partially overlapping regions $AB$ and $BC$ only if the conditional mutual information $I(A:C|B)$ vanishes, which again we do not expect to happen in a QFT. So we should not expect to be able to saturate arbitrary sets of regions.

It turns out that the key property is \emph{nesting}: we can saturate on any nested set of regions $A_1\subset A_2\subset\cdots$, $\int_{A_i}f=S(A_i)$. Actually, we can do better than that, and have a second set of regions $B_1\subset B_2\subset\cdots$, disjoint from the $A_i$s, that are saturated with the other sign, $\int_{B_j}f=-S(B_j)$. And on top of that, we can require $f\ge0$ on $A_1$ and $f\le0$ on $B_1$. For the proof of this statement, see theorem \ref{thm:saturatepos} in subsection \ref{sec:EDFproofs}. It's easy to see why we cannot require $f\ge0$ on any of the $A_i$s except $A_1$ or $f\le0$ on any of the $B_i$s except $B_1$. For example, if $S(A_2)<S(A_1)$ and we saturate on $A_1$, then the integral of $f$ on $A_2\setminus A_1$ must be negative. Note that, as expected from the arguments above, the proof invokes the (strong) subadditivity property of the entropy.

Another useful fact is that, if $M$ is embedded in a larger manifold $M'$, then any EDF $f$ on $M$ can be extended to an EDF $f'$ on $M'$ such that $f'=f$ on $M$. Again, this theorem relies on SSA. For the proof, see theorem \ref{thm:restriction} in subsection \ref{sec:EDFproofs}.

\subsubsection{Examples}
\label{examples}
To get a feeling for EDFs and extremal EDFs, let us start with a few discrete, finite-dimensional examples before turning to QFTs:
\begin{itemize}
\item For a Bell pair on two qubits $a,b$, an EDF must obey $|f(a)|\le\ln2$, $|f(b)|\le\ln2$, and $f(a)=-f(b)$. The extremal EDFs are $f(a)=-f(b)=\pm\ln2$.
\item A 4-party perfect tensor is a pure state where every individual party has entropy $s_0>0$ and every pair has entropy $2s_0$. Then the extremal EDFs are
\be
f(a)=f(b)=-f(c)=-f(d)=s_0
\ee
and permutations thereof. (Note that this state has vanishing mutual information between any pair of parties, so they can be simultaneously saturated even though they're not nested. Also note that when we look at pairs of parties, these EDFs do not saturate except for the pairs $ab$ and $cd$.)
\item An $n$-party GHZ state is a pure state such that any subset of between 1 and $n-1$ parties has entropy $\ln 2$. The extremal EDFs are
\be\label{GHZEDF}
f(a)=-f(b)=\ln2\,,\qquad f(c)=\cdots=0
\ee
and permutations thereof.
\end{itemize}

Among QFTs, a simple example is provided by any (1+1)-dimensional CFT in the vacuum on the line. Given the classic formula $S(A)=(c/3)\ln(L/\epsilon)$ for a single interval of length $L>\epsilon$, an example of an EDF is
\be
f(x)=\begin{cases}
0\,,&|x|<\epsilon \\
\frac c{3x}\,,&|x|>\epsilon 
\end{cases}\,.
\ee
This saturates on all intervals of the form $[-L,0]$ or $[0,L]$. It is easy to check that any single interval obeys \eqref{subcontourdef}. For a region consisting of $n>1$ intervals, the entropy has a divergent piece $n\ln(1/\epsilon)$ and is therefore automatically larger than the integral of $f$.

Finally, as we already noted in subsection \ref{sec:flowsdouble}, in classical holography, a boundary function $f$ can equal the boundary flux of a classical flow if and only if $f$ is an EDF. (See  theorem \ref{thm:classicalflowEDF}.)

\subsection{More on strict flows}
\label{sec:morestrict}

Armed with the intuition concerning EDFs gathered in the previous subsection, we return to strict quantum flows and note a few properties before turning to examples.

First, strict flows share with classical flows the fact that their boundary flux is a (boundary) EDF, and conversely, every boundary EDF is the boundary flux of a strict flow. The first statement 
follows directly from the definitions: Fix a boundary region $A$. For any region $r\in\RR_A$,
\be
\int_An\cdot v=\int_{\eth r}n\cdot v-\int_r\nabla\cdot v\le \frac{|\eth r|}{4\GN}+S_{\rm b}(r)=S_{\rm gen}(r)\,.
\ee
Hence
\be
\int_An\cdot v\le\min_{r\in\RR_A}S_{\rm gen}(r)=S(A)\,.
\ee
The proof that $-\int_An\cdot v\le S(A)$ is similar. The converse is harder to prove; see theorem \ref{thm:strictsubcontour}. 

Also, as in the classical case, strict flows obey the nesting property: 
Given regions $A_1\subset \cdots\subset A_m$, $B_1\subset\cdots \subset B_n$ in $\RR_{\partial\Sigma}$ such that $A_m\cap B_n=\emptyset$, there exists a strict quantum flow $v$ such that
\be
\forall\, i=1,\ldots,m\,,\quad\int_{A_i}n\cdot v=S(A_i)\,;\qquad\forall\, j=1,\ldots,n\,,\quad\int_{B_j}n\cdot v=-S(B_j)\,.
\ee
For the proof, see theorem \ref{thm:strictnesting}. 
(Nesting was proven for loose flows in \cite{Rolph:2021hgz}.)

Alternatively, nesting follows as a corollary of two theorems quoted above: an EDF exists that saturates on any nested set of regions (see subsection \ref{sec:subcontours}); and a strict flow exists whose boundary flux equals any given EDF (see just above). However, the latter route assumes that $S$ obeys SSA (used in proving the existence of a saturating EDF), whereas the direct proof of the theorem on nesting of strict flows does not. The distinction is noteworthy because nesting can itself be used to prove SSA, by the following argument \cite{Freedman:2016zud}. Let $A, B, C$ be disjoint regions, and let $v$ be a max flow for both $B$ and $ABC$. Then we have $S(AB)\ge\int_{AB}n\cdot v$, $S(BC)\ge\int_{BC}n\cdot v$, so
\be
S(AB)+S(BC)\ge \int_{AB}n\cdot v+\int_{BC}n\cdot v=\int_Bn\cdot v+\int_{ABC}n\cdot v=S(B)+S(ABC)\,.
\ee

Classical flows also obey a stronger property, the existence of a so-called \emph{max multiflow} \cite{Cui:2018dyq}, which implies that $S$ obeys the monogamy of mutual information (MMI) inequality \cite{Hayden:2011ag},
\be\label{MMI}
S(AB)+S(BC)+S(AC)\ge S(A)+S(B)+S(C)+S(ABC)\,.
\ee
Unlike SSA, MMI is \emph{not} generally valid for the QES formula (a counterexample will be provided in the next subsection), so it must \emph{not} be the case that quantum max multiflows always exist. However, there are some circumstances where we would expect them to exist, such as when $S_{\rm b}$ obeys MMI. We will define and discuss quantum multiflows in subsection \ref{sec:nomultiflows}.

\subsection{Point-particle examples}
\label{sec:pointparticle}

Let us look at a few simple examples corresponding to the boundary examples in \ref{examples} to see how the EDF condition naturally encodes the entanglement structure of the bulk. We will focus on the contribution of a set of entangled point particles to the bulk entropy.

First, suppose the bulk contains two particles at locations $x_1, x_2$ in a Bell state. For any region $r$ containing $x_1$ but not $x_2$, there is a contribution of $\ln 2$ to $S_{\rm b}(r)$. 
Therefore, a strict flow may include flow lines with flux $\le\ln2$, beginning or ending on $x_1$. Similarly, flow lines with flux $\le\ln2$ may begin or end on $x_2$. 
However, the entropy of any bulk point's neighbourhood, except for $x_1$ and $x_2$, has zero entropy, so the flow lines cannot begin or end at any bulk point besides $x_1$ and $x_2$. 
Furthermore, if $r$ includes both $x_1$ and $x_2$, then the Bell pair makes no contribution to $S_{\rm b}(r)$, and so, however much flux ends on $x_1$, the same amount must begin on $x_2$; we can imagine that the flux ``jumps'' from $x_1$ to $x_2$.  Note that the same result could have been obtained by a purely geometric application of the RT formula (without the bulk entropy term) if the manifold $\Sigma$ had an extra Planckian handle attached, of cross section $4\GN\ln2$, connecting $x_1$ to $x_2$.

Next, suppose there are four particles at $x_1,\ldots,x_4$ jointly in a perfect tensor state with individual entropies $S(x_i)=s_0$. We can refer to the analysis of extremal EDFs in subsection \ref{sec:subcontours} to analyze the general strict flow. Every such flow is a convex combination of one with sources of flux $s_0$ at $x_1,x_2$ and sinks of flux $s_0$ at $x_3,x_4$, and permutations thereof. The same result would have been obtained by a purely geometrical application of the RT formula if a handle were attached to $\Sigma$ consisting of a sphere connected to each $x_i$ by a tube of area $4\GN s_0$.

Finally, for $n$ particles jointly in a GHZ state, a general strict flow is a convex combination of ones with a source of flux $\ln2$ at any one particle, an equal sink at any other particle, and no flux at the remaining particles. Thus, putting maximal flux through any pair of particles blocks the other ones from having any flux. For $n=3$, these entropies can be realized geometrically by attaching a sphere with three tubes to $\Sigma$, similar to the perfect tensor state of the last paragraph. However, for $n>3$, the entropies \emph{cannot} be realized geometrically by adding to $\Sigma$ a handle of any form, as they violate the MMI inequality. In this sense, the GHZ state is ``truly quantum''. The boundary entropies will likewise violate MMI if the classical and area terms in $S_{\rm gen}$ saturate the inequality, for example, if the boundary regions are sufficiently far apart that every QES is the union of QESs of individual regions.

To put this discussion in the context of the conceptual interpretation of bit threads, we have shown that, for some bulk states, it is possible to connect entangled bulk degrees of freedom with Planckian wormholes in such a way that classical bit threads have the same maximal boundary flux as the quantum bit threads. For these states, it is equivalent to have the threads either ``jump" between or to create and travel through Planckian wormholes connecting entangled bulk degrees of freedom. This is an intriguing picture and reminiscent of the ER/EPR paradigm; however, as demonstrated by the $n>3$ GHZ state example in the previous paragraph, this picture is not possible for all bulk states.

\subsection{Changing the cutoff}
\label{sec:changingcutoff}

We will now discuss, at a qualitative level, the behavior of maximal (strict or loose) flows, and what happens to them as we change the UV cutoff $\epsilon$. Given a boundary region $A$, let $\gamma_A$ be the QES and $r_A$ its homology region (so $\eth r_A=\gamma_A$). A maximal flow must saturate both the norm bound on $\gamma_A$ and the divergence bound for $r_A$. The bulk entropy contains a divergent area-law term $|\gamma_A|/\epsilon^{D-2}$, reflecting entanglement of ultraviolet modes down to a wavelength $\epsilon$ across the QES.

In simple cases, the QES coincides, or nearly coincides, with the classical RT surface. This happens when $S_{\rm b}$ is dominated by the area term, which is, of course, minimized at the RT surface, so $S_{\rm b}$ does not have a large gradient there. Then the area term can be accounted for by giving $v$ a negative divergence on one side of the QES and a positive divergence on the other; for example, a shell of thickness of order $\epsilon$ with a divergence of order $\epsilon^{1-D}$ will work. The corresponding flow lines thus ``jump across'' the QES. If the bulk entropy has a large gradient, moving the QES significantly away from its classical position, then the divergence will have to be positive (or negative) on both sides of the QES.

Now let us consider what happens if we change the cutoff $\epsilon$. Recall that, as usual in physics, the cutoff is not a physical parameter but a calculational device we impose. Under this change, the area term $|\eth r|/\epsilon^{D-2}$ and therefore the bulk entropy change, while $\GN$ makes a compensating change so that the physical quantity $S_{\rm gen}(r)$ is invariant \cite{Susskind:1994sm,Cooperman:2013iqr}. (Recall that $\GN$ is the bare parameter appearing in the Lagrangian.) In particular, if we decrease $\epsilon$, then both $S_{\rm b}$ and $\GN$ increase, so the divergence constraint becomes looser while the norm bound becomes tighter. The flux of $v$ through $A$ does not change, but what happens to that flux in the bulk changes, as less of it can go through the QES and more of it must jump across.

In so-called ``induced gravity'' theories, the Einstein-Hilbert term in the action is generated entirely by quantum effects of the matter fields. Hence, in the limit that $\epsilon$ approaches the Planck length, $\GN$ goes to infinity. When $\GN$ is very large, the norm bound largely excludes the flow from the bulk; only near the boundary, where space is very large so the flow can be very dilute, can there be an appreciable flux (see Fig.\ \ref{fig:inducedgrav}). In this limit, the maximal flux is controlled entirely by the divergence bound, and is equal to $\min_{r\in\RR_A}S_{\rm b}(r)$. Given that, in this limit, $S_{\rm gen}=S_{\rm b}$, the result agrees with the QES formula.

\begin{figure}
    \centering
    \includegraphics[width=0.5\linewidth]{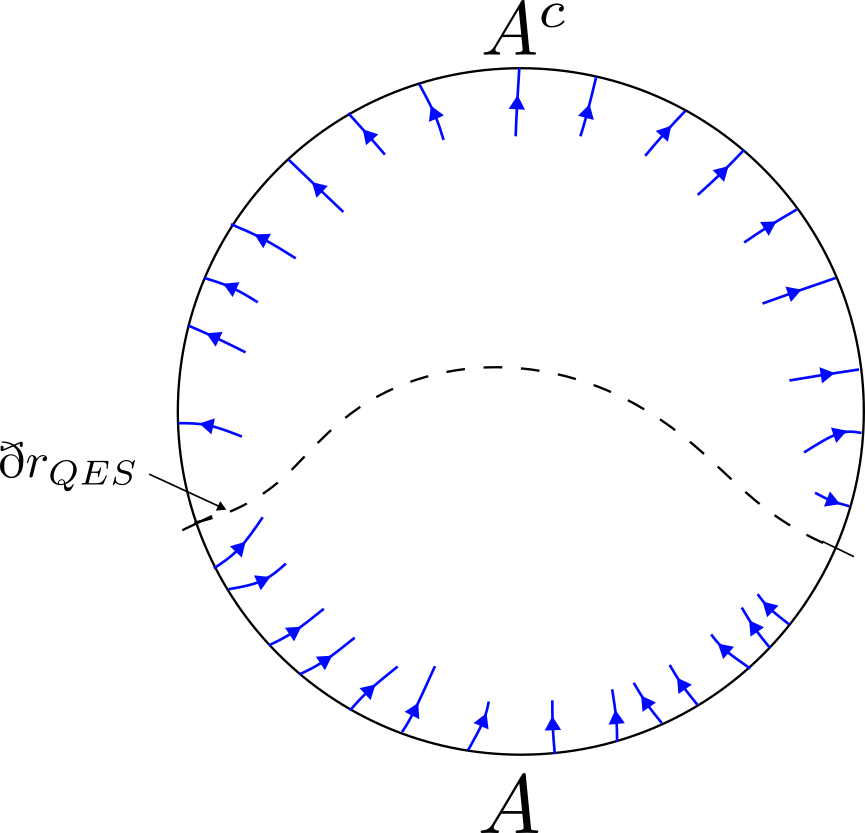}
    \caption{A maximal flow configuration in the $G_N \to \infty$ limit.}
    \label{fig:inducedgrav}
\end{figure}

\subsection{Theorems \& proofs}
\label{sec:proofs}

In this subsection, we state and prove several theorems that were used in the rest of the section. These proofs use the same tools of convex analysis, especially convex relaxation and dualization, as in the Riemannian max flow-min cut theorem and related theorems, which may be found in \cite{Headrick:2017ucz}. Therefore, we will be brief. We will also not make any attempt to be rigorous, particularly as regards the measure theory and functional analysis aspects of the proofs.

We start by defining the relevant terms and notation, before turning to the proofs.

\subsubsection{Definitions}

\begin{definition}Given a compact manifold-with-boundary $M$, a \emph{region} of $M$ is a compact codimension-0 submanifold-with-boundary.
\end{definition}

\begin{definition}
A \emph{classical holography setup} consists of a compact Riemannian manifold-with-boundary $\Sigma$ and positive constant $\GN$. $\RR$ is the set of regions of $\Sigma$ and $\mathcal{A}$ is the set of regions of $\partial\Sigma$. Given a region $r\in\RR$, $\eth r:=\partial r\setminus\partial\Sigma$. Given a region $A\in\mathcal{A}$,
\be\label{Radef}
\RR_A:=\{r\in\RR:r\cap\partial\Sigma=A\}\,.
\ee
$S$ is the function on $\A$ defined by \eqref{RT}. A \emph{classical flow} is a vector field $v$ on $\Sigma$ obeying
\be
|v|\le\frac1{4\GN}\,,\qquad \nabla\cdot v=0\,.
\ee
The inward-directed unit normal vector on $\partial\Sigma$ is denoted $n$.
\end{definition}

\begin{definition}
A \emph{quantum holography setup} consists of a compact Riemannian manifold-with-boundary $\Sigma$, a positive constant $\GN$, a subset\footnote{We allow for the possibility that $\Sb$ is defined only on a subset $\RR$ of the full set of bulk regions; for example, regions smaller than some cutoff size $\epsilon$ may not have a well-defined entropy.} $\RR$ of the set of regions of $\Sigma$, and a non-negative function $\Sb$ on $\RR$. $\RR$ must include the empty set and be closed under finite unions; 
and, for any region $A$ of $\partial\Sigma$, the set $\RR_A$, defined as in \eqref{Radef}, must not be empty. $\Sb$ must satisfy $\Sb(\emptyset)=0$ and, for any $r_1,r_2,r_3\in\RR$,
\begin{align}
\Sb(r_1r_2)+\Sb(r_2r_3)&\ge \Sb(r_2)+\Sb(r_1 r_2r_3)\label{SSA}\\
\Sb(r_1r_2)+\Sb(r_2r_3)&\ge \Sb(r_1)+\Sb(r_3)\,.\label{WM}
\end{align}
Given a region $r\in\RR$, $\eth r := r\setminus\partial\Sigma$. $\A$ is the set of regions of $\partial\Sigma$. $S$ is the function on $\A$ defined by \eqref{QES}, \eqref{Sgendef}. A \emph{strict quantum flow} is a vector field $v$ on $\Sigma$ obeying
\be
|v|\le\frac1{4\GN}\,,\qquad
\forall\,r\in\RR\,,\quad\left|\int_r\nabla\cdot v\right|\le\Sb(r)\,.
\ee
The inward-directed unit normal vector on $\partial\Sigma$ is denoted $n$.
\end{definition}

\begin{definition}
A \emph{double holography setup} consists of a quantum holography setup together with a compact Riemannian manifold-with-boundary $\tilde\Sigma$, such that $\Sigma$ is a region of $\partial\tilde\Sigma$, and a positive constant $\widetilde\GN$, with $\RR$ the set of all regions of $\Sigma$ and $\Sb$ defined by \eqref{RT} applied to $\tilde\Sigma$, $\widetilde\GN$. A \emph{double flow} is a classical flow $\tilde v$ on $\tilde\Sigma$ together with a vector field $v$ on $\Sigma$ obeying
\be
|v|\le\frac1{4\GN}\,,\qquad
\nabla\cdot v+\tilde n\cdot\tilde v=0\,,
\ee
where $\tilde n$ is the inward-directed unit normal to $\Sigma$ in $\tilde\Sigma$.
\end{definition}

\subsubsection{Theorems}
Lagrange duality will be used in many of our proofs to establish equivalence of optimisation problems. To aid the reader, let us give a one paragraph summary (see the textbook~\cite{boyd2004convex} for further details).
We start from a optimisation problem called the primal,
\bne \min_x f_0(x) \text{ subject to } f_i (x) \leq 0 \text{ for } (i = 1,\dots, m),\, h_j(x) = 0 \text{ for } (j=1,\dots,n), \ene
construct a Lagrangian with Lagrange multiplier terms to enforce the primal's constraints, then  optimise over the variables in the primal problem to arrive at the Lagrange dual problem. A sufficient condition for the primal and dual problems to be equivalent, i.e. to have to same optimal value, for strong duality to hold, is for the primal problem to be convex and for Slater's condition to be satisfied. This condition requires there to exist a strictly feasible point in  the domain of the primal problem, meaning a point that strictly satisfies the nonlinear constraints.

\begin{theorem}\label{thm:doublemaxflow}
In a double holographic setup, for any $A\in\A$,
\be
S(A) = \max\int_An\cdot v\quad\text{over double flows $(v,\tilde v)$}\,.
\ee
\end{theorem}

\begin{proof}
In this setup, the definition of $S(A)$ can be written as \eqref{RTdouble}. We convex-relax this minimization problem by 
introducing real functions $\phi$ on $\Sigma$ and $\tilde\phi$ on $\tilde\Sigma$, with boundary conditions
\be
\left.\tilde\phi\right|_\Sigma=\phi\,,\qquad
\left.\phi\right|_{A} = 1\,,\quad
\left.\phi\right|_{A^c} = 0\,.
\ee
The level sets for $\phi$ in the range $[0,1]$ are homologous to $A$, while those outside that range are null-homologous. The level sets for $\tilde\phi$ are homologous to the homology regions in $\Sigma$ for the level sets of $\phi$. The coarea formula says that $\int_\Sigma|\nabla\phi|$ equals the integrated area of the $\phi$ level sets, and $\int_{\tilde\Sigma}|\nabla\tilde\phi|$ equals the integrated area of the $\tilde\phi$ level sets. Therefore
\be
\min_{\phi,\tilde\phi}
\frac1{4\GN}\int_\Sigma|\nabla\phi|+\frac1{4\widetilde\GN}\int_{\tilde\Sigma}|\nabla\tilde\phi| = S(A)\,,
\ee
with the minimum achieved by setting $\phi$, $\tilde\phi$ equal to step functions supported on the minimal homology regions.

We now introduce vector field $w$, $\tilde w$ on $\Sigma$, $\tilde\Sigma$ respectively, and consider the convex program
\begin{multline}
\text{Minimize }
\frac1{4\GN}\int_\Sigma|w|+\frac1{4\widetilde\GN}\int_{\tilde\Sigma}|\tilde w|
\text{ over $w,\tilde w$, $\phi$, $\tilde\phi$ subject to:} \\
w=\nabla\phi\,,\quad
\tilde w=\nabla\tilde\phi\,,\quad
\left.\tilde\phi\right|_\Sigma=\phi\,,\quad
\left.\phi\right|_{A} = 1\,,\quad
\left.\phi\right|_{A^c} = 0\,.
\end{multline}
Introducing Lagrange multipliers $v$, $\tilde v$ (vector fields on $\Sigma$, $\tilde\Sigma$ respectively) for the first two constraints, and imposing the rest of them implicitly, we obtain \eqref{doublemaxflow} as the dual program. As there are no inequality constraints, Slater's condition is automatically satisfied.
\end{proof}

\begin{theorem}\label{thm:classicalflowEDF}
Given a classical holography setup, let $f$ be a function on $\partial\Sigma$. There exists a classical flow $v$ such that $n\cdot v=f$ 
if and only if, for all $A\in\A$,
\be\label{EDFdef1}
\left|\int_Af\right|\le S(A)\,.
\ee
\end{theorem}

\noindent Remarks:
\begin{itemize}
\item The graph version of this theorem can be found in theorem \ref{thm:graphEDF} in section \ref{sec:entropohedron}. The two theorems have similar proofs.
\item By theorem \ref{thm:restriction}, a function $f$ defined only within a region $M\subset\partial\Sigma$, and such that \eqref{EDFdef1} is obeyed on subregions $A\subseteq M$, may be extended to a function on all of $\partial\Sigma$ obeying \eqref{EDFdef1}. Therefore, theorem \ref{thm:classicalflowEDF} holds with such an $M$ and such an $f$.
\item Technically, this theorem is a corollary of theorem \ref{thm:strictsubcontour}, obtained by setting $\RR$ to be the set of all regions of $\Sigma$ and $S_{\rm b}=0$. However, we believe the proof will be easier to understand if we do the classical case first.
\end{itemize}

\begin{proof}
We consider the following convex program:
\be\label{classicalflowconvex}
\text{Minimize }\frac1{4\GN}\int_\Sigma|\nabla\phi|-\int_{\partial\Sigma}\phi\,f\text{ over $\phi$}\,.
\ee
We will first solve this program; we will then dualize it and relate its solution to the existence of a classical flow $v$ such that $n\cdot v=f$.

Given a function $\phi$, we define, for  $\hat\phi\in\R$, the following bulk and boundary regions
\be
r(\hat\phi):=\left\{x\in \Sigma:0\le \phi(x)\le \hat\phi\text{ or }\hat\phi\le\phi(x)\le0\right\},\qquad
A(\hat\phi):=r(\hat\phi)\cap\partial\Sigma\,,
\ee
and the objective of \eqref{classicalflowconvex} can be written as
\be
\int d\hat\phi\left[\frac{|\eth r(\hat\phi)|}{4\GN}-
\sgn(\hat\phi)\int_{A(\hat\phi)}f\right],
\ee
where the first term in the integrand was obtained using the coarea formula. That term is bounded below by $S(A(\hat\phi))$ by the definition of $S$. If \eqref{EDFdef1} holds, then the second term is bounded below by $-S(A(\hat\phi))$. Hence, the objective is non-negative, and its minimum value is 0, achieved by $\phi=0$. On the other hand, if \eqref{EDFdef1} is violated for some $A$, then by choosing $\phi$ to be constant on the minimal homology region for $A$ and 0 elsewhere, the dual objective can be made negative and arbitrarily large, so the minimum is $-\infty$.

Before dualizing \eqref{classicalflowconvex}, we introduce a vector field $w$ which we constrain to equal $\nabla\phi$:
\be\label{classicalflowconvex2}
\text{Minimize }\frac1{4\GN}\int_\Sigma|w|-\int_{\partial\Sigma}\phi\,f\text{ over $\phi,w$ subject to }w=\nabla\phi\,.
\ee
Note that this program has no inequality constraints, so Slater's condition is automatically obeyed and therefore strong duality holds. The Lagrangian is
\begin{align}
L[\phi,w,v]&=\int_\Sigma\left[\frac{|w|}{4\GN}+v\cdot(w-\nabla\phi)\right]-\int_{\partial\Sigma}\phi\,f \\
&= \int_\Sigma\left[\frac{|w|}{4\GN}+v\cdot w+\phi\nabla\cdot v\right]+\int_{\partial\Sigma}\phi(n\cdot v-f)\,.
\end{align}
Minimizing over $w$ and $\phi$ gives the following dual program:
\be\label{EDFprogram}
\text{Maximize 0 over $v$ subject to:} \qquad|v|\le\frac1{4\GN}\,,
\qquad\nabla\cdot v=0\,,
\qquad\left.n\cdot v\right|_{\partial\Sigma}=f\,.
\ee
This is a feasibility-type program: the maximum is 0 if the feasible set is non-empty and $-\infty$ if it is empty. This proves the theorem.
\end{proof}

For the next few theorems, we will need the following ``disentangling'' lemma, which can be obtained as the continuum limit of corollary \ref{thm:disentangling}, which is the analogous statement in the discrete setting. Corollary \ref{thm:disentangling} rests upon two lemmas: Lemma \ref{thm:disentangling1} states that, by applying the weak monotonicity property \eqref{WM}, one can remove overlaps between the sets of regions $\RR_+$, $\RR_-$, without changing $\phi$ and without increasing the integral of $S_{\rm b}$; and lemma \ref{thm:disentangling2} states that, by applying the strong subadditivity property \eqref{SSA}, one can obtain nested sets of regions, again without changing $\phi$ and without increasing the integral of $S_{\rm b}$. To prove these lemmas rigorously in the continuum would involve analysis and measure theory that are beyond the scope of this paper.

\begin{lemma}\label{thm:disentanglingcont}
Given a quantum holography setup, let $\mu_\pm$ be measures on $\RR$, and define
\be\label{phidef}
\phi:=\int_{\RR}(d\mu_--d\mu_+)\chi_r\,,
\ee
with $\chi_r$ the characteristic function, also known as the indicator function, of $r$. There exist nested sets of regions $\RR_+,\RR_-\subset\RR$ such that, for all $r_+\in\RR_+$ and $r_-\in\RR_-$, $r_+\cap r_-=\emptyset$ and measures $\mu_\pm'$ such that
\be
\mu_+(\RR\setminus\RR_+)=0\,,\qquad\mu_-(\RR\setminus\RR_-)=0
\ee
\be
\phi'=\phi
\ee
\be
\int_{\RR}(d\mu'_++d\mu'_-)S_{\rm b}(r)\le
\int_{\RR}(d\mu_++d\mu_-)S_{\rm b}(r)\,.
\ee
\end{lemma}

\begin{theorem}\label{thm:S(A)strict}
Given a quantum holography setup,
\be\label{S(A)strict2}
S(A) = \max\int_A n\cdot v\text{ over strict quantum flows $v$}\,.
\ee
\end{theorem}
As discussed in subsection \ref{sec:strictflowsdef}, where this theorem is used, we provide two proofs, both involving Lagrangian dualization but in opposite directions.
\begin{proof}[Proof 1: Flow to surface dualization]
The proof is similar to that in \cite{Rolph:2021hgz} for the corresponding theorem for loose flows. We write the RHS of \eqref{S(A)strict2} as a concave program:
\be\label{strictconcave}
\text{Maximize }\int_An\cdot v\text{ over $v$ subject to: }
|v|\le\frac1{4\GN}\,,\qquad
\forall\,r\in\RR,\quad\left|\int_r\nabla\cdot v\right|\le S_{\rm b}(r)\,.
\ee
The feasible configuration $v=0$ shows that Slater's condition is obeyed, so strong duality holds. We dualize \eqref{strictconcave}, introducing as Lagrange multipliers a scalar field $\psi\ge0$ for the first constraint, and measures $\mu_\pm$ on $\RR$ for $\pm\int_r\nabla\cdot v\le S_{\rm b}(r)$ respectively. The Lagrangian is
\begin{align}
&L[v,\psi,\mu_+,\mu_-] \nonumber\\
&=\int_An\cdot v
+\int_\Sigma\psi\left(\frac1{4\GN}-|v|\right)
+\int_\RR d\mu_+\left(S_{\rm b}(r)-\int_r\nabla\cdot v\right)
+\int_\RR d\mu_-\left(S_{\rm b}(r)+\int_r\nabla\cdot v\right) \nonumber\\
&=\int_{\partial\Sigma}n\cdot v\,(\chi_A-\phi)
+\int_\Sigma\left(\frac\psi{4\GN}-\psi|v|-v\cdot\nabla\phi\right)
+\int_\RR(d\mu_++d\mu_-)S_{\rm b}(r)\,,
\end{align}
where $\phi$ is defined in \eqref{phidef}. Maximizing the Lagrangian with respect to $v$, we find the following dual program:
\be
\text{Minimize }
\frac1{4\GN}\int_\Sigma\psi+\int_\RR(d\mu_++d\mu_-)S_{\rm b}(r)
\text{ over $\psi,\mu_\pm$ subject to: }\psi\ge|\nabla\phi|\,,\quad \left.\phi\right|_{\partial\Sigma}=\chi_A\,.
\ee
Minimizing over $\psi$ by setting it equal to $|\nabla\phi|$ gets us:
\be\label{finalconvex}
\text{Minimize }
\frac1{4\GN}\int_\Sigma|\nabla\phi|+\int_\RR(d\mu_++d\mu_-)S_{\rm b}(r)
\text{ over $\mu_\pm$ subject to: }
\left.\phi\right|_{\partial\Sigma}=\chi_A\,.
\ee
Thanks to lemma \ref{thm:disentanglingcont}, we can assume that the supports of $\mu_\pm$ are nested and mutually disjoint. Since $\phi$ is non-negative on $\partial\Sigma$, this implies that the support of $\mu_+$ includes only regions that do not intersect $\partial\Sigma$. To minimize the objective, we should therefore simply set $\mu_+=0$. The support of $\mu_-$, meanwhile, includes regions in $\RR_A$ with total weight 1, along with regions that do not intersect $\partial\Sigma$. To minimize the objective, we should zero out $\mu_-$ on the latter set. This leaves only regions in $\RR_A$. \eqref{finalconvex} is therefore equivalent to minimizing
\be
\int_{\RR_A}d\mu_-\left(\frac{|\eth r|}{4\GN}+S_{\rm b}(r)\right)
\ee
over measures $\mu_-$ on $\RR_A$ with nested support and total measure 1 (where we used the coarea formula to replace the integral of $|\nabla\phi|$ with the integrated area of its level sets). This is equivalent to the QES formula.
\end{proof}

\begin{proof}[Proof 2: Surface-to-flow dualization]
We start from the QES prescription written in the form~\eqref{finalconvex}. This is a convex problem, and Slater's condition is satisfied. The corresponding Lagrangian, leaving the definition of $\phi$ as an implicit constraint, is
\bne L[w,\mu_\pm,\psi,v]  = \int_\Sigma \frac{|w|}{4G_N} + \int_\cR (d\mu_+ + d\mu_-)S_b(r) + \int_\Sigma v\cdot  (\nabla \phi - w) + \int_{\del \Sigma} \psi (\chi_A - \phi) .\ene
The resulting dual problem is
\bne \sup_{v,\psi} \left(\int_A \psi \right), \qquad \text{ subject to } |v| \leq \frac{1}{4G_N} , \quad \forall r \in \cR:  \left|\int_{\eth r} v - \int_{\del \Sigma \cap r} \psi \right| \leq S_b (r). \ene
$\psi$ is a scalar function on $\del \Sigma$, and we can eliminate it by evaluating the second inequality on infinitesimal neighbourhoods of points on the boundary of $\Sigma$, which forces $\psi$ to equal to boundary flux density of $v$. Then, the dual problem simplifies to the strict quantum bit thread prescription
\bne \sup_v \int_A v \qquad \text{ subject to } |v| \leq \frac{1}{4G_N} , \quad \forall r \in \cR:  \left|\int_r \nabla \cdot v \right| \leq S_b (r). 
\ene
\end{proof}

\begin{theorem}\label{thm:looseS}
Given a quantum holography setup,
\be\label{eq:looseS}
S(A) = \max\int_A n\cdot v\text{ over loose quantum flows $v$}\,.
\ee
\end{theorem}
This theorem was proved in \cite{Rolph:2021hgz}. Here, we give an alternate proof that involves dualizing from surfaces to flows, similar to the second proof of theorem \ref{thm:S(A)strict}. On a technical note, this proof requires an additional assumption,
that there is a large volume near $\del \Sigma$, such as when $\Sigma$ is the region within a large radial cutoff surface of an asymptotically hyperbolic manifold.
\begin{proof}
We start from the QES prescription, written in the form~\eqref{QES}.
This is not a convex problem, so, before we dualize, we rewrite~\eqref{QES} in an equivalent convex-relaxed form. We smear out the homology surfaces with a scalar function $\psi$ whose level sets, where $\psi = s$, are surfaces $\eth r_s \in \cR_A$. We get
\bne S(A)
= \inf_\psi \left[ \int_\Sigma \frac{|d\psi|}{4G_N} + \int_0^1 ds S_b (r_s) \right] \text{ subject to } \psi|_{\del \Sigma} = \chi_A\ene
and where $r_s := \{x \in \Sigma: \psi(x) \geq s \}$. This is convex in $\psi$ because of the strong subadditivity of bulk entropies, and Slater's condition is satisfied.

The Lagrangian is
\bne L[\psi,w,v] = \int_\Sigma \left(\frac{|w|}{4G_N} + v\cdot (w-d\psi)\right) + \int_0^1 ds S_b (r_s), \ene
and the dual problem is
\bne \sup_v \inf_\psi \left[ \int_\Sigma \psi \nabla \cdot v + \int_0^1 ds S_b (r_s) + \int_A v\right], \qquad \text{subject to }  |v| \leq \frac{1}{4G_N}. \ene
We rewrite this, using $\psi(x) = \int_0^{\psi(x)} ds = \int_0^1 ds \chi_{r_s} (x)$, and get 
\bne \begin{split} 
=&\sup_v \left( \int_A v + \inf_{r \in \cR_A} \left( \int_r \nabla \cdot v + S_b (r) \right) \right) ,\qquad \text{ subject to } |v|\leq \frac{1}{4G_N}. \label{eq:qbitha} 
\end{split}\ene

Next, we will show that there is guaranteed to exist an optimal flow $v^*$ for~\eqref{eq:qbitha} satisfying
\bne \label{eq:odb}  \inf_{r\in \cR_A} \int_{r} \left( \nabla \cdot v^* + S_b(r )\right) = 0, \ene 
because this implies that
\bne  \forall r \in \cR_A : \quad\int_{r } \nabla \cdot v^* + S_b(r) \geq 0, \label{eq:qbvco} \ene
and then, since~\eqref{eq:qbvco} is satisfied by an optimal $v$, we can impose~\eqref{eq:qbvco} as a constraint on all $v$'s in~\eqref{eq:qbitha}, which gives us the loose quantum bit thread prescription that we are looking for:
\bne \sup_v \int_A v, \quad \text{subject to:}\quad |v| \leq \frac{1}{4G_N}, \quad \forall r \in \cR_A: \int_r \nabla \cdot v + S_b(r) \geq 0. \ene

\textit{(1) Classical.}
First, we consider the classical case, setting $S_b$ to zero.
In~\eqref{eq:qbitha}, think of $r$ as an adversary to $v$, with $r$ trying to minimise $\int_r \nabla \cdot v$ while $v$ is trying to maximise it. It is disadvantageous to $v$ to have $\nabla \cdot v < 0$ at any point, because $r$ can decrease $\int_r \nabla \cdot v^*$ by including that point. So, we cannot have $\nabla \cdot v^* < 0$ anywhere; we necessarily have $\nabla \cdot v^* \geq 0$. Furthermore, there is no advantage to $v$ to have $\nabla \cdot v > 0$ at any point, because $r$ will simply avoid that point, and it does not help with the norm bound; because of the norm bound, it can only be disadvantageous to $v$ to have $\nabla \cdot v >0$ anywhere. So, there must exist a $v^*$ with $\nabla \cdot v^* = 0$.

\textit{(2) Quantum.}
Now we extend the argument to include quantum corrections. 
Suppose, contrary to~\eqref{eq:odb}, that we are given optimal $v^*$ and $r^*$ for which 
\bne \int_{r^*} \nabla \cdot v^* + S_b(r^*) <0. \label{eq:suppose} \ene
Since $S_b (r) \geq 0$ for all $r$, there must be more sinks than sources in $r^*$. We can define a new optimal flow by removing sinks in $r^*$ and their attached field lines from $v^*$, without risk of violating the norm bound. 
The field lines from the sources have to end on $A$ or in the interior of $r^*$, otherwise $v^*$ was not optimal. 
The new flow is still optimal because removing sinks in $r^*$ cannot change the optimal flux, which is $\int_{{\eth r}^*} v^* + S_b(r^*)$,
since removing sinks cannot increase the flux through ${\eth r}^*$, and it cannot decrease the flux because that is inconsistent with the optimality of $v^*$. 
Furthermore, $r^*$ is also still optimal for this new flow, because 
the flux through ${\eth r}^*$ is unchanged and the flux through other ${\eth r}$'s can only increase as sinks with $A$-attached field are removed, so the set of minimising ${\eth r}$'s can only shrink, not grow. 
From $v^*$, we keep removing sources and their attached field lines in $r^*$ until we have reached an optimal $v$ for which~\eqref{eq:odb} is satisfied, and then we are done. We cannot run out of sources to remove before we are done, as can be seen from~\eqref{eq:suppose} and the positivity of bulk entropies.

Next, if the LHS of~\eqref{eq:suppose} is positive, then we modify $v^*$ by adding field lines that are anchored on $A$ and extend to sinks that are infinitesimally into the bulk. This decreases the integrated divergence in $r^*$ and we keep doing so until~\eqref{eq:odb} is satisfied. Both the new flow remains in the set of optimal flows and $r^*$ remains optimal, because the flux through $\eth r^*$ is unaffected. It is possible to modify $v^*$ this way without violating the norm bound, assuming that there is sufficient space in the neighbourhood of $\del \Sigma$.
\end{proof}

\begin{theorem}\label{thm:strictsubcontour}
Given a quantum holography setup, let $f$ be a function on $\partial\Sigma$. There exists a strict quantum flow $v$ such that $n\cdot v=f$ 
if and only if, for all $A\in\A$,
\be\label{strictEDFcondition}
\left|\int_Af\right|\le S(A)\,.
\ee
\end{theorem}

\begin{proof}
The proof is similar to that of theorem \ref{thm:classicalflowEDF}, except that since the divergence constraint is on the set of regions $\RR$, rather than points, we will have to enforce it using measures rather than a function.

We start with the following convex program:
\be\label{strictflowconvex}
\text{Minimize }\frac1{4\GN}\int_\Sigma|\nabla\phi|+\int_{\RR}(d\mu_++d\mu_-)S_{\rm b}(r)-\int_{\partial\Sigma}\phi\,f\text{ over $\mu_\pm$}\,,
\ee
where $\mu_\pm$ are measures on $\RR$ and $\phi$ is defined in \eqref{phidef}. We will first solve \eqref{strictflowconvex}, and then dualize it and relate its solution to the existence of a strict flow $v$ such that $n\cdot v=f$.

To solve \eqref{strictflowconvex}, we first note that, by lemma \ref{thm:disentanglingcont}, we can assume that the supports of $\mu_\pm$ are nested and mutually disjoint. Using the coarea formula on the first term, the objective can then be written
\be
\int_{\RR}d\mu_+\left(\frac{|\eth r|}{4\GN}+S_{\rm b}(r)+\int_{r\cap\partial\Sigma}f\right)+
\int_{\RR}d\mu_-\left(\frac{|\eth r|}{4\GN}+S_{\rm b}(r)-\int_{r\cap\partial\Sigma}f\right).
\ee
If \eqref{strictEDFcondition} holds, then the objective is non-negative; its minimum value is 0, achieved by $\mu_\pm=0$. On the other hand, if \eqref{strictEDFcondition} is violated for some $A$, then by choosing either $\mu_+$ or $\mu_-$ to be supported on the corresponding homology region, the objective can be made negative and arbitrarily large, so the minimum value is $-\infty$.

We now dualize \eqref{strictflowconvex}, first introducing as new variables a function $\tilde\phi$ and vector field $w$, constrained to equal $\phi$ and $\nabla\tilde\phi$ respectively:
\begin{multline}\label{strictflowconvex2}
\text{Minimize }\frac1{4\GN}\int_\Sigma|w|+\int_{\RR}(d\mu_++d\mu_-)S_{\rm b}(r)-\int_{\partial\Sigma}\tilde\phi\,f\text{ over $\mu_\pm,\tilde\phi,w$ subject to:}\\ 
w=\nabla\tilde\phi\,,\qquad \tilde\phi=\int_{\RR}(d\mu_--d\mu_+)\chi_r\,.
\end{multline}
The program involves only equality constraints, so Slater's condition is automatically satisfied and strong duality is guaranteed. To dualize, we introduce as Lagrange multipliers a vector field $v$ for the first constraint and a function $\psi$ for the second one. The Lagrangian is:
\begin{align}
L[\mu_+,\mu_-,\tilde\phi,w,v,\psi] 
&=\int_\Sigma\frac{|w|}{4\GN}+\int_{\RR}(d\mu_++d\mu_-)S_{\rm b}(r)-\int_{\partial\Sigma}\tilde\phi\,f\nonumber\\
&\qquad\qquad+\int_\Sigma v\cdot(w-\nabla\tilde\phi)+\int_\Sigma\psi\left(\tilde\phi-\int_{\RR}(d\mu_--d\mu_+)\chi_r\right)\nonumber \\
&=\int_\Sigma\left(\frac{|w|}{4\GN}+v\cdot w+\tilde\phi\nabla\cdot v+\psi\tilde\phi\right)+\int_{\partial\Sigma}\tilde\phi(n\cdot v-f) \nonumber\\
&\qquad\qquad
+\int_{\RR}d\mu_+\left(S_{\rm b}(r)+\int_r\psi\right)
+\int_{\RR}d\mu_-\left(S_{\rm b}(r)-\int_r\psi\right)
\end{align}
Minimizing with respect to $\mu_\pm,\tilde\phi,w$, we find that the dual program is
\begin{multline}
\text{Maximize }0\text{ over $v,\psi$ subject to:} \\
|v|\le\frac1{4\GN}\,,\qquad
\psi=-\nabla\cdot v\,,\qquad
n\cdot v=f\,,\qquad
\forall r\in\RR\,,\quad\pm\int_r\psi\le S_{\rm b}(r)\,.
\end{multline}
This is a feasibility-type program: if there exists a pair $v,\psi$ obeying the constraints, then the optimal value is 0; if not, then it is $-\infty$. Existence of such a pair is equivalent to the existence of a strict quantum flow obeying $n\cdot v=f$.
\end{proof}

\begin{theorem}\label{thm:strictnesting}
Given a quantum holography setup, let $A_1\subset \cdots\subset A_m$, $B_1\subset\cdots \subset B_n$ be regions in $\A$ such that $A_m\cap B_n=\emptyset$. There exists a strict quantum flow $v$ such that
\be
\forall\, i=1,\ldots,m\,,\quad\int_{A_i}n\cdot v=S(A_i)\,;\qquad\forall\, j=1,\ldots,n\,,\quad\int_{B_j}n\cdot v=-S(B_j)\,.
\ee
\end{theorem}

\begin{proof}
The theorem is a generalization of theorem \ref{thm:S(A)strict}. We consider the following concave program:
\begin{multline}\label{strictnestingprimal}
\text{Maximize }\sum_{i=1}^m\int_{A_i}n\cdot v-\sum_{j=1}^n\int_{B_j}n\cdot v\text{ over $v$ subject to:}\\
|v|\le\frac1{4\GN}\,,\qquad
\forall r\in\RR\,,\quad
\left|\int_r\nabla\cdot v\right|\le S_{\rm b}(r)\,.
\end{multline}
Since each term in the objective is bounded above by $S(A_i)$ or $S(B_j)$, if the objective attains the value
\be\label{target}
\sum_{i=1}^mS(A_i)+\sum_{j=1}^nS(B_j)\,,
\ee
then each term must achieve its upper bound, and the theorem is proven.

We now dualize \eqref{strictnestingprimal}, noting that the feasible configuration $v=0$ obeys Slater's condition, guaranteeing strong duality. Introducing as Lagrange multipliers a scalar field $\psi\ge0$ for the first constraint and measures $\mu_\pm$ on $\RR$ for $\pm\int_r\nabla\cdot v\le S_{\rm b}(r)$ respectively, the Lagrangian is:
\begin{align}
&L[v,\psi,\mu_+,\mu_-] \nonumber\\
&=\sum_{i=1}^m\int_{A_i}n\cdot v-\sum_{j=1}^n\int_{B_j}n\cdot v \nonumber \\
&\qquad\qquad+\int_\Sigma\psi\left(\frac1{4\GN}-|v|\right)
+\int_\RR d\mu_+\left(S_{\rm b}(r)-\int_r\nabla\cdot v\right)
+\int_\RR d\mu_-\left(S_{\rm b}(r)+\int_r\nabla\cdot v\right) \nonumber\\
&=\int_{\partial\Sigma}n\cdot v\,(\chi-\phi)
+\int_\Sigma\left(\frac\psi{4\GN}-\psi|v|-v\cdot\nabla\phi\right)
+\int_\RR(d\mu_++d\mu_-)S_{\rm b}(r)\,,
\end{align}
where $\phi$ is defined in \eqref{phidef}, and
\be
\chi:=\sum_{i=1}^m\chi_{A_i}-\sum_{j=1}^n\chi_{B_j}\,.
\ee
Maximizing the Lagrangian with respect to $v$, we find the following dual program:
\be
\text{Minimize }
\frac1{4\GN}\int_\Sigma\psi+\int_\RR(d\mu_++d\mu_-)S_{\rm b}(r)
\text{ over $\psi,\mu_\pm$ subject to: }\psi\ge|\nabla\phi|\,,\quad \left.\phi\right|_{\partial\Sigma}=\chi\,.
\ee
Minimizing over $\psi$ by setting it equal to $|\nabla\phi|$ gets us:
\be
\text{Minimize }
\frac1{4\GN}\int_\Sigma|\nabla\phi|+\int_\RR(d\mu_++d\mu_-)S_{\rm b}(r)
\text{ over $\mu_\pm$ subject to: }
\left.\phi\right|_{\partial\Sigma}=\chi\,.
\ee
Thanks to lemma \ref{thm:disentanglingcont}, we can assume that the supports of $\mu_\pm$ are nested and mutually disjoint. Given the boundary condition $\phi|_{\partial\Sigma}=\chi$, this implies that the supports include three types of regions: for each $i$, the support of $\mu_-$ includes regions in $\RR_{A_i}$ with total weight 1; for each $j$, the support of $\mu_+$ includes regions in $\RR_{B_j}$ with total weight 1; and there may be other regions that do not intersect $\partial\Sigma$. To minimize the objective, we should zero out the last type of region. The objective thus becomes
\be
\sum_{i=1}^m\int_{\RR_{A_i}}d\mu_-\left(\frac{|\eth r|}{4\GN}+S_{\rm b}(r)\right)
+\sum_{j=1}^n\int_{\RR_{B_j}}d\mu_+\left(\frac{|\eth r|}{4\GN}+S_{\rm b}(r)\right).
\ee
By the definition of $S$, the minimum of the objective equals \eqref{target}.
\end{proof}

\subsection{Multiflows}
\label{sec:nomultiflows}

The pattern followed by the theorems in the previous subsection was that every statement that was true about classical flows was true about strict quantum flows as well. We will now give an important example where this pattern breaks down, involving so-called multiflows.

\begin{definition}
Given a classical holographic setup with $\partial\Sigma$ partitioned into regions $A_1,\ldots,A_n$, a \emph{classical multiflow} is a set of vector fields $v_{ij}$ ($i,j=1,\ldots,n$) on $\Sigma$ such that
\be\label{multiflowdef}
v_{ij}=-v_{ji}\,,\qquad
\sum_{i<j}|v_{ij}|\le\frac1{4\GN}\,,\qquad
\nabla\cdot v_{ij}=0\,,\qquad
\left.n\cdot v_{ij}\right|_{A_k}=0\quad(k\neq i,j)\,.
\ee
\end{definition}
From the definition, each component vector field $v_{ij}$ is itself a classical flow. More than that, any linear combination $\sum_{i<j}c_{ij}v_{ij}$ where the coefficients obey $|c_{ij}|\le1$ is a classical flow. In particular, for any individual boundary region $A_i$, we can define the flow $v_i:=\sum_jv_{ij}$; its flux $\int_{A_i}v_i$ is bounded above by $S(A_i)$. The following theorem, proved in \cite{Cui:2018dyq}, shows that there exists a classical multiflow that simultaneously saturates these bounds for all $i$:
\begin{theorem}\label{thm:mmf}
Given a classical holographic setup with $\partial\Sigma$ partitioned into regions $A_1,\ldots,$ $A_n$, there exists a classical multiflow $\{v_{ij}\}$ such that, for all $i=1,\ldots,n$,
\be\label{mmfdef}
S(A_i)=\int_{A_i}n\cdot \left(\sum_{j=1}^nv_{ij}\right).
\ee
\end{theorem}
We call such a multiflow a \emph{max multiflow}. The MMI inequality is a corollary of its existence:
\begin{corollary}\label{thm:MMIproof}
Given a classical holographic setup with $\partial\Sigma$ partitioned into four regions $A,B,C,D$, the MMI inequality holds:
\be\label{MMIrep}
S(AB)+S(AC)+S(BC)\ge S(A)+S(B)+S(C)+S(ABC)\,.
\ee
\end{corollary}
\begin{proof}
Let $v_{ij}$ be a max multiflow. The vector field
\be
\tilde v_{AB}:=v_{AC}+v_{AD}+v_{BC}+v_{BD}
\ee
is a classical flow, and therefore its flux through $AB$ is bounded above by $S(AB)$:
\be\label{MMI1}
S(AB)\ge \int_{AB}n\cdot\tilde v_{AB}=
\int_An\cdot\left(v_{AC}+v_{AD}\right)+\int_Bn\cdot\left(v_{BC}+v_{BD}\right).
\ee
Similarly, we have
\begin{align}\label{MMI2}
S(AC)&\ge\int_An\cdot\left(v_{AB}+v_{AD}\right)+\int_Cn\cdot\left(v_{CB}+v_{CD}\right) \\
\label{MMI3}
S(BC)&\ge\int_Bn\cdot\left(v_{BA}+v_{BD}\right)+\int_Cn\cdot\left(v_{CA}+v_{CD}\right) .
\end{align}
Summing \eqref{MMI1}--\eqref{MMI3}, and using \eqref{mmfdef} and the fact that $S(D)=S(ABC)$ (since $ABCD$ cover the entire boundary) yields \eqref{MMIrep}.
\end{proof}

We can similarly define (strict) quantum multiflows in such a way that any sum of its component vector fields is a strict quantum flow:
\begin{definition}
Give a quantum holographic setup with $\partial\Sigma$ partitioned into regions $A_1,\ldots,A_n$, a \emph{strict quantum multiflow} is a set of vector fields $v_{ij}$ ($i,j=1,\ldots,n$) on $\Sigma$ such that
\begin{multline}
v_{ij}=-v_{ji}\,,\qquad
\sum_{i<j}^n|v_{ij}|\le\frac1{4\GN}\,,\qquad
\forall r\in\RR\,,\quad\sum_{i<j}^n\left|\int_r\nabla\cdot v_{ij}\right|\le \Sb(r)\,,\\
\left.n\cdot v_{ij}\right|_{A_k}=0\quad(k\neq i,j)\,.
\end{multline}
\end{definition}

If a max quantum multiflow exists, then by the same reasoning as in corollary \ref{thm:MMIproof}, the MMI inequality will hold. But we already know that this is not always the case; indeed, we saw a counterexample in subsection \ref{sec:pointparticle} in the form of four point particles in a 4-party GHZ state. Therefore, it cannot be that a quantum max multiflow always exists. Indeed, for the 4-party GHZ example, every particle has entropy $\ln2$; therefore, a max multiflow would have flux $\ln 2$ entering or leaving each particle. But such a multiflow does not exist, since putting flux $\ln2$ from any particle to any other particle blocks any flux from entering or leaving the other two particles.

Nonetheless, under certain conditions, one might expect a max quantum multiflow to exist. The following conjectures, which we leave to future work to prove or disprove, give two examples of such conditions.

\begin{conjecture}
Given a quantum holographic setup, with the boundary partitioned into three regions $A_1,A_2,A_3$, there exists a strict quantum multiflow $v_{ij}$ such that, for $i=1,2,3$, \eqref{mmfdef} holds.
\end{conjecture}

\begin{conjecture}
Given a quantum holographic setup, with the function $\Sb$ obeying the MMI inequality
\be
\Sb(r_1r_2)+\Sb(r_2r_3)+\Sb(r_1r_3)\ge \Sb(r_1)+\Sb(r_2)+\Sb(r_3)+\Sb(r_1r_2r_3)
\ee
for all disjoint $r_1,r_2,r_3\in\RR$, and with the boundary partitioned into regions $A_1,\cdots,A_n$, there exists a strict quantum multiflow $v_{ij}$ such that, for all $i=1,\ldots,n$, \eqref{mmfdef} holds.
\end{conjecture}

\section{Cutoff-independent flow prescriptions} \label{sec:cutin}

\subsection{The prescriptions}

In this section, we will discuss and derive cutoff-independent prescriptions, whose constraints are in terms of the generalised entropies of regions. In each of these prescriptions, we maximise the flux of $v$ through $A$, but with different possible cutoff-independent constraints on $v$. 

\subsubsection{Strict \& loose prescriptions}

There are two cutoff-independent prescriptions that we will prove are equivalent to the QES prescription. 

1) \textit{Loose, cutoff-independent quantum bit threads}:
\bne S(A) = \sup_v \int_A v \qquad \text{subject to:} \quad \forall r\in\RR_A : \int_{\eth r} |v| - \int_r \nabla \cdot v \leq \Sgen  (r) \label{eq:loocu} \ene

2) \textit{Strict, cutoff-independent quantum bit threads}:
\bne S(A) = \sup_v \int_A v \qquad \text{subject to:} \quad \forall r\in\RR : \int_{\eth r} |v| + \left| \int_r \nabla \cdot v \right |\leq \Sgen  (r) \label{eq:strcu}\ene

The strict prescription is stricter than the loose prescription, both because we take the absolute value of the integrated divergence term, and because the constraint is applied to the larger set of bulk subregions $\RR \supset \RR_A$.

Both prescriptions only have a single constraint each. These constraints are cutoff-independent, because $\Sgen (r)$ is cutoff-independent, in contrast to $G_N$ and $S_b$ that appear in the constraints of the cutoff-dependent prescriptions. 

Later in this section, we will prove that both~\eqref{eq:loocu} and~\eqref{eq:strcu} are equivalent to the QES prescription
and therefore equivalent to each other. 
The equivalence between flow prescriptions is non-trivial: for the strict constraint, the space of allowed flows is smaller than that allowed by the loose constraint, so the maximal flux clearly cannot be larger with the strict constraint, but it is not obvious that it is the same. 
Even though the set of allowed flows becomes smaller, $\sup_v \int_A v$ stays the same and is equal to the optimum of the QES formula: $S(A)$.

\subsubsection{Comparison to the cutoff-dependent prescriptions} \label{sec:compa}

The cutoff-independent prescriptions have only one constraint, while the cutoff-dependent prescriptions have two.
 
The cutoff-independent prescriptions are looser than their cutoff-dependent cousins: every allowed flow in the loose/strict cutoff-dependent prescription is an allowed flow in the loose/strict cutoff-independent prescription. 
The inequality constraint in the loose/ strict cutoff-independent prescription equals the sum of the inequality constraints in the loose/strict cutoff-dependent prescription, after integrating $|v| \leq 1/4G_N$ over $\eth r$, and the sum of any two inequalities is a looser constraint than the two inequalities applied separately. 

The cutoff-independent prescriptions do not capture the area and bulk area pieces of $\Sgen $ separately; they capture the whole of $\Sgen $ directly. In contrast, in the cutoff-dependent prescription, for a flux-maximising flow, it is the constraint on the norm of $v$ that causes the maximal flux to capture the area piece of $\Sgen (r) = |{\eth r}|/4G_N + S_{bulk}(r)$, and it is the constraint on the divergence of $v$ that accounts for the bulk entropy piece. 

\subsubsection{Other prescriptions} \label{sec:other_prescriptions}

A third cutoff-independent prescription that we will analyse applies the loose constraint from~\eqref{eq:loocu} plus an additional constraint that $v$ be divergenceless. 

3) \textit{Loose, divergenceless, cutoff-independent quantum bit threads:}

\bne S(A) = \sup_v \int_A v \qquad \text{subject to:} \quad \nabla \cdot v = 0, \quad \forall r\in\RR_A : \int_{\eth r} |v|  \leq \Sgen  (r). \label{eq:divloocu} \ene

We did not list this with the loose and strict constraints because it is not always QES-equivalent, though it is generically; in section~\ref{sec:dsfpd} we will derive the surface-based prescription that is equivalent to~\eqref{eq:divloocu} and see precisely when it is QES-equivalent. 

In section~\ref{sec:dsfpd}, we will also see that the \textit{strict} version of~\eqref{eq:divloocu} is not generically QES-equivalent: if we modified the norm bound constraint in~\eqref{eq:divloocu} to apply to the larger set of bulk subregions $\RR$, rather than $\RR_A$, then the constraints are too strict.

Flows that obey the loose constraint and are also divergenceless can be interpreted as the most ``classical'' of the flows allowed by the loose constraints. Flows in the loose cutoff-independent prescription allow for both highly ``quantum'' flow configurations (all the threads jumping over the QES and none passing through), and highly ``classical'' flows (as divergenceless as allowed, with the maximal flow passing through the QES), and everything in between these two extremes.

There are also cutoff-independent flow prescriptions that are similar-looking to the three we have discussed so far, but are trivially equivalent to the QES prescription, and we need to underscore what is different between the trivial and non-trivial prescriptions.
An example of such a trivial prescription is
\bne S(A) = \sup_v \int_A v, \qquad \text{subject to } \forall r \in \cR_A: \int_A v \leq \Sgen  (r). \label{eq:trivp} \ene
The flow field $v$ in this trivial flow prescription is totally unconstrained, except on the boundary subregion $A$. The flow field can do absolutely anything in the interior of the bulk, so any physical interpretation of the bit threads is gone, including that they capture the entanglement structure of entangled bulk fields, and that the QES is a bottleneck to the flow. Furthermore, physical properties of the bulk theory, such as strong subadditivity of $\Sgen $, are not needed to prove the equivalence of the trivial flow prescription to the QES prescription. For these reasons, we do not consider a trivial prescription like~\eqref{eq:trivp} to be an interesting or useful bit thread prescription. In contrast, in our prescriptions~\eqref{eq:loocu}, \eqref{eq:strcu}, and~\eqref{eq:divloocu}, as we will see, the flow field $v$ is strongly constrained in the bulk interior, the QES does act as a bottleneck, and proving the equivalence to the QES prescription requires bulk SSA and is non-trivial. 

In between the strict and loose prescriptions, one can also consider flow prescriptions with ``intermediate'' cutoff-independent constraints: (1)  with the absolute value on the $\int_r \nabla \cdot v$ term but only applying to all $r \in \cR_A$, and (2) no absolute value but applying to all $r \in \cR$. Every strict flow is an intermediate flow, and every intermediate flow is a loose flow. Since loose and strict flows both have maximum flux equal to $S(A)$, i.e. are QES-equivalent, the intermediate prescriptions are trivially also QES-equivalent. 
For these reasons, we do not consider the intermediate prescriptions to be sufficiently different from the loose and strict prescriptions to be of interest and will not discuss them further. 

\subsection{Strict, cutoff-independent flows}
We will determine properties of flows that satisfy the strict constraint because doing so will help us understand how the cutoff-independent prescriptions are equivalent to the QES prescription. 

Flows obeying the strict constraint automatically also obey the loose constraint; allowed strict flows are a subset of allowed loose flows.

\subsubsection{Different ways of writing the strict constraint} The strict constraint is:
\bne \forall r \in \cR: \qquad \int_{\eth r} |v| + \left| \int_r \nabla\cdot v \right| \leq \Sgen (r) .  \label{eq:stsgb} \ene

For the subset of bulk regions that are disjoint from $\del \Sigma$, $r\in \cR_{\varnothing} \subset \cR$, we have that $\eth r = \del r$, and~\eqref{eq:stsgb} can be written in a form that highlights that the constraint is on the behaviour of $v$ on the boundary of every bulk region:
\bne \forall r \in \cR_\varnothing: \qquad \int_{\del r} |v|+ \left | \int_{\del r} v \right| \leq \Sgen (r). \ene

It will also be useful that~\eqref{eq:stsgb} is equivalent to the pair of inequalities
 \bne \int_{\eth r} |v| + \int_{\eth r} v - \int_{r\cap \del \Sigma}v \leq \Sgen (r)  \label{eq:someq} \ene
 and
 \bne \int_{\eth r} |v| - \int_{\eth r} v + \int_{r\cap \del \Sigma}v\leq \Sgen (r)  \label{eq:someq2}. \ene 

We can also write this in terms of the normal and tangent components of $v$ on $\eth r$:
\bne \int_{\eth r}\left ( \sqrt{v_t^2 + v_n^2} \pm v_n \right ) \mp \int_{r\cap \del\Sigma} v \leq \Sgen (r). \ene
For a given $r$, this constraint acts to suppress flows with large components tangent to $\eth r$. This dovetails with a property of strict flows that we will determine next, that the max-flow must be normal to the QES on the QES.

\subsubsection{The strict constraint is saturated on the entanglement wedge
} We can show that the strict constraint~\eqref{eq:stsgb}  is tight. It is saturated for $\eth r$ equal to the QES when the flow is maximal, i.e. when $\int_A v = \Sgen ({\eth r}_{_{QES}})$. This is similar to how, for classical bit threads, the norm bound $|v| \leq 1/4G_N$ is saturated on the RT surface. 

We can actually show something more general that implies what we want: if $\int_A v = \Sgen  ({ r})$, for any ${\eth r}$ homologous to $A$, then~\eqref{eq:stsgb} is saturated for that ${\eth r}$, and $v$ is perpendicular to ${\eth r}$. To show this, we apply~\eqref{eq:stsgb} to any ${\eth r}$ homologous to $A$ and with $\int_A v = \Sgen  ({ r})$ gives
\bne \label{eq:corgb} \int_{\eth r} |v| + \left| \int_{\eth r} v - \Sgen  ({ r})\right| \leq \Sgen  ({r}) .\ene
This is equivalent to the pair of inequalities
\bne \int_{\eth r} |v| - \int_{\eth r} v \leq 0 \label{eq:ineon} \ene
and
\bne \int_{\eth r} |v| + \int_{\eth r} v \leq 2 \Sgen  ({r}). \label{eq:oteqn} \ene
If either one of these is saturated then~\eqref{eq:corgb} must be saturated. 
The inequality~\eqref{eq:ineon} together with the identity $\int_{\eth r} |v| - \int_{\eth r} v \geq 0$ imply that~\eqref{eq:ineon} is saturated, and therefore that (1) $v$ on ${\eth r}$ is parallel to the normal, and (2) that~\eqref{eq:corgb} is saturated.

\subsubsection{Constraints implied by the strict constraint} \label{sec:cibtsc}
\paragraph{Non-local norm and divergence bounds.} 
Using $\int_{\eth r} |v| \geq 0$ in~\eqref{eq:stsgb} gives an upper and lower bound on the divergence of $v$ in any bulk subregion:
\bne \left|\int_r \nabla \cdot v \right| \leq \Sgen  (r) \label{eq:indiv} \ene

Using $|\int_r \nabla \cdot v | \geq 0$ in~\eqref{eq:stsgb} gives%
\footnote{Two other ways to derive this: (1) using~\eqref{eq:jussq} in~\eqref{eq:someq} gives
	\bne \int_{\eth r} |v| + \int_{\eth r} v \leq 2 \Sgen  (r) \implies \int_{\eth r} v \leq \Sgen  (r) \ene
	and 
	(2) add the pair of inequalities together.
	}
\bne \int_{\eth r} |v| \leq \Sgen  (r) \label{eq:nlnbd} \ene 
which also implies $\int_{\eth r} v \leq \Sgen  (r)$.
So, one implication of~\eqref{eq:stsgb} is that the flux through any ${\eth r}$ is upper bounded by its generalised entropy.

From~\eqref{eq:stsgb}, for a given $r$,~\eqref{eq:indiv} and~\eqref{eq:nlnbd} cannot be saturated simultaneously, unless the generalised entropy is zero. If~\eqref{eq:nlnbd} is saturated, then we must have $\int_r \nabla\cdot v = 0$, i.e. the net flux into $r$ is zero. If~\eqref{eq:indiv} is saturated, then $\int_{\eth r} |v| = 0$, so $v=0$ on $\eth r$, and $|\int_{r\cap \del\Sigma} v| = \Sgen (r)$, which can only be true if $r$ is the entanglement wedge for $r\cap \del \Sigma$ and so $|\int_{r \cap \del \Sigma} v| = S(r\cap \del \Sigma)$. 

\textbf{Bound on boundary flux.} Using $\pm \int_{\eth r} v - \int_{\eth r} |v| \leq 0$,~\eqref{eq:someq2} implies that
\bne \left| \int_{r\cap \del \Sigma} v \right| \leq \Sgen  (r) \label{eq:jussq} \ene
which is consistent with the fundamental property of bit thread prescriptions: $\sup_v \int_{r\cap \del \Sigma} v  = S(r\cap \del \Sigma) = \inf_r \Sgen (r)$.

\textbf{A local norm bound from the strict constraint applied to small regions.} We can derive a local norm bound on $|v|$. Applying the constraint~\eqref{eq:nlnbd} to a region sufficiently small that $|v|$ is approximately constant gives a local upper bound on the norm of $v$:
\bne \begin{split} |v(x)| &\leq \inf_{\eth r \ni x} \frac{\Sgen (r)}{|\eth r|}\\
&=  \frac{1}{4G_N} + D S_b(x) , \qquad D S_b (x) =  \inf_{\eth r \ni x } \frac{S_{b} (r)}{|\eth r|} \label{eq:locnb} \end{split} \ene 

We minimise over the set of small $\eth r$ to make the bound as tight as possible, but we only minimise over those $\eth r$ whose scale is much larger than the UV cutoff. The bound~\eqref{eq:locnb} depends on just how small we allow $\eth r$ to be.

In the classical limit,~\eqref{eq:locnb} reduces to the classical bit thread norm bound $|v| \leq \frac{1}{4G_N}$. 
For generic, weakly entangled states, $DS_b (x)$ is a small correction to the classical norm bound $|v| \leq 1/4G_N$. This is important because, as in the classical case, up to small corrections, 
we can think of the threads as tubes that have a cross-sectional area equal to a quarter in Planck units that limits how closely they can be packed together.

In appendix \ref{sec:gen_entropy}, we discuss $\Sgen $ and $D S_b (x)$; in particular, their cutoff-dependence and finiteness.

The upper bound~\eqref{eq:locnb} is useful for our intuition because it is local; however, it is not necessarily tight, even for maximal flows on the QES, especially for highly entangled bulk states, and when there is an entanglement island, because small local regions do not know the global entanglement structure.  

\subsubsection{Maximally classical \& quantum flows}

The set of allowed flows in the strict/loose cutoff-independent prescription is large: it is a superset of the flows allowed in the strict/loose cutoff-\textit{dependent} prescriptions, as explained in section~\ref{sec:compa}.

The most ``classical" maximal flow configurations are those where the flow is as divergenceless as possible, and the flux through $A$ equals the flux through the QES, i.e. $S(A)$ is captured by the bit threads passing through the QES. The local norm bounds $|v| \lesssim 1/4G_N$ tell us that the flux through the QES must be approximately evenly distributed over the QES. 

Nothing we have derived from~\eqref{eq:stsgb} forbids having maximal flows with $v = 0$ on QES. These are the most ``quantum'' flow configurations, with all the threads jumping over the QES. But where these threads start and end in the bulk is constrained. Any threads that end in the entanglement wedge of $A$ have to reappear in the complementary entanglement wedge, for pure bulk states, because of~\eqref{eq:indiv}. Where the threads end in the entanglement wedge is also constrained by~\eqref{eq:indiv}, they cannot all end in an infinitesimal region, and there is a limit to how closely they can be packed together, as can be seen from, for example, the local norm bound~\eqref{eq:locnb}.
Fig.~\ref{fig:class_quant} depicts the flows that we have been discussing.

\begin{figure}
    \centering
    \includegraphics[width=0.99\linewidth]{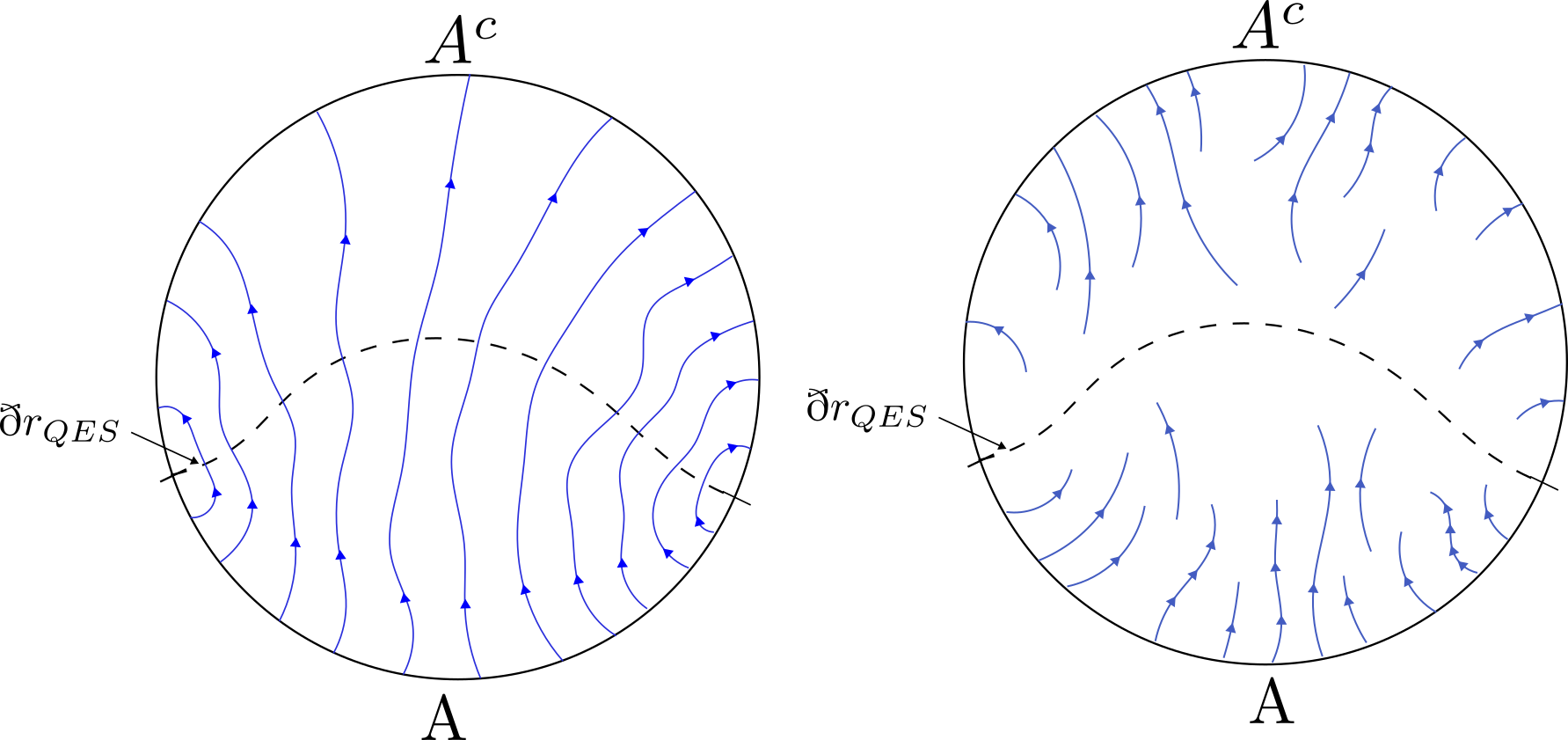}
    \caption{Opposite extremes in the set of max-flows allowed by the strict cutoff-independent constraint~\eqref{eq:stsgb}. Left: the most classical, with zero divergence, and all the threads passing through the QES. Right: the most quantum, with all the threads jumping over the QES. Both max-flows have the same flux on $A$: $S(A)$.}
    \label{fig:class_quant}
\end{figure}

If we decrease the bulk cutoff length $\epsilon$, then the bulk entropies increase, while the area piece of the $\Sgen $ decreases, precisely to keep $\Sgen $ $\epsilon$-independent. In the cutoff-dependent prescriptions, this change allows more threads to jump over the QES and fewer to pass through it, but the max flux through $A$ is unaffected. The flow becomes less classical, in the sense that we have used the term in this subsection. 

Furthermore, there is a similarity between the set of all allowed flows in the cutoff-independent prescriptions, and the unions of sets of all allowed flows in the cutoff-dependent prescriptions across all cutoff lengths $\epsilon$, allowing for a spatially-varying cutoff\footnote{In the cutoff-dependent prescriptions, if $\epsilon$ is constant on $\Sigma$, then $|v|$ must be constant for a max flow on the QES; this is not required in the cutoff-independent prescriptions, so the set of allowed flows is definitely not the same.}. The sets are not identical, as there are always threads passing through the QES in the cutoff-dependent prescriptions for any finite ($\epsilon$-dependent) $G_N$.

\subsection{Divergenceless, loose, cutoff-independent flows}

The divergenceless, loose, cutoff-independent prescription is~\eqref{eq:divloocu}. These divergenceless flows have similarities and differences with respect to the strict flows. The divergenceless flow lines cannot start or end at points in the bulk. From their definition, the flows obey~\eqref{eq:nlnbd}, just like the strict flows, but only for $r \in \RR_A$. The divergenceless flows do obey the same local norm bound~\eqref{eq:locnb} as the strict flow, but we will have to derive it another way, because in this prescription we cannot apply~\eqref{eq:nlnbd} to arbitrary bulk regions.

The constraint $\int_{ \eth r} |v| \leq \Sgen (r)$ is tight: it is saturated for a maximal flow when $r = r_{QES}$. 
This is because $\Sgen (r_{QES}) = S(A)$, which equals $\int_A v$ because $v$ is maximal, which equals $\int_{\eth r_{QES}} v$ using the divergencelessness of $v$. The constraint becomes $\int_{\eth r_{QES}} |v| \leq \int_{\eth r_{QES}} v$, which can only be satisfied if (1) the constraint is saturated for $r = r_{QES}$, and (2) $v$ on the QES is perpendicular to the QES.

We can derive a local norm bound for the divergenceless flows, that is the same as~\eqref{eq:locnb}, by using that the inequality constraint~\eqref{eq:divloocu} is saturated for max-flows when applied to $r = r_{QES}$. We subtract $\int_{\eth r_{QES}} |v| = \Sgen (r_{QES})$ from the inequality constraint applied to $r = r_{QES} \cup B$, where $B$ is a small ball region anywhere in the interior%
\footnote{We require it to be the interior because it is important for the argument that $B \cap \eth r_{QES} = \varnothing$.}
 of $r_{QES}^c$, and, using subadditivity of generalised entropies, this gives
\bne \int_{\del B} |v| \leq \Sgen  (B). \ene
Dividing this through by the surface area of the ball $|\del B|$, and taking $|\del B|$ to be sufficiently small that $|v|$ is approximately constant, gives
\bne  \qquad |v| \leq \frac{1}{4 G_N} + DS_b(x). \label{eq:locnb3} \ene
for all $ x \in r_{QES}^c$, with $DS_b (x)$ defined in~\eqref{eq:locnb}.
Running the same argument by subtracting a ball region from the interior%
\footnote{Again, interior is important so that $B \cap \eth r_{QES} = \varnothing$, but also so that $B \cap \del \Sigma = \varnothing$ because the argument needs that $r_{QES} \ B \in \RR_A$.}
 of $r_{QES}$ gives the same inequality applied to all $x \in r_{QES}$, and so the local norm bound~\eqref{eq:locnb3} holds for all $x \in \Sigma$. 

For a weakly entangled bulk state, $\Sgen  (r) \approx |\eth r|/4G_N$ for all $r \in R$, and the bulk entropy correction to the norm bound in~\eqref{eq:locnb3} to the classical norm bound is subleading, so~\eqref{eq:locnb3} becomes $|v| \lesssim 1/4G_N$ on the QES. Furthermore, in order that $\int_{\eth r_{QES}} v = \Sgen (r_{QES}) \approx |\eth r_{QES}|/4G_N$, we must have $|v| \approx 1/4G_N$ on the QES, so~\eqref{eq:locnb3} must be approximately saturated, which shows that~\eqref{eq:locnb3} is reasonably tight, and that the QES acts as a bottleneck to the flow.
$|v| \approx 1/4G_N$ on the QES with a small correction, and this additional flux density is how the bulk entropy contribution to $S(A)$ is captured. 

We can change the bulk state to increase bulk entropies, and then~\eqref{eq:locnb3} allows for a larger $|v|$, though, of course, this does not necessarily mean that the maximal flow makes use of this; it depends on whether~\eqref{eq:locnb3} is saturated at a point on the QES, and whether the bound is weakened at that point in the new state. 

When we Lagrange dualise this flow prescription in section~\ref{sec:dsfpd}, we will see when it is QES-equivalent. One sufficient but not necessary condition is that $A$ has a boundary.

\subsubsection{Islands}
Maximal flows in this divergenceless prescription cannot enter the island. Using that $\frac{v}{|v|} = n$ on $\del I$, we have that $0= \int_I \nabla \cdot v = \int_{\del I} |v|$, so $v = 0$ on $\del I$. See Fig.~\ref{fig:divergenceless} for a max-flow configuration when there is an island. The boundary of the island does \textit{not} act as a bottleneck to the divergenceless flow; the flow goes around the island. The importance of the island to the flow is that $\Sgen (a\cup I) < \Sgen (a)$, so that, compared to $r =a$, the inequality constraint in~\eqref{eq:divloocu} applied to $r = a \cup I$ gives a tighter constraint on the flux density on $\eth a$.

\begin{figure}
    \centering    \includegraphics[width=0.4\linewidth]{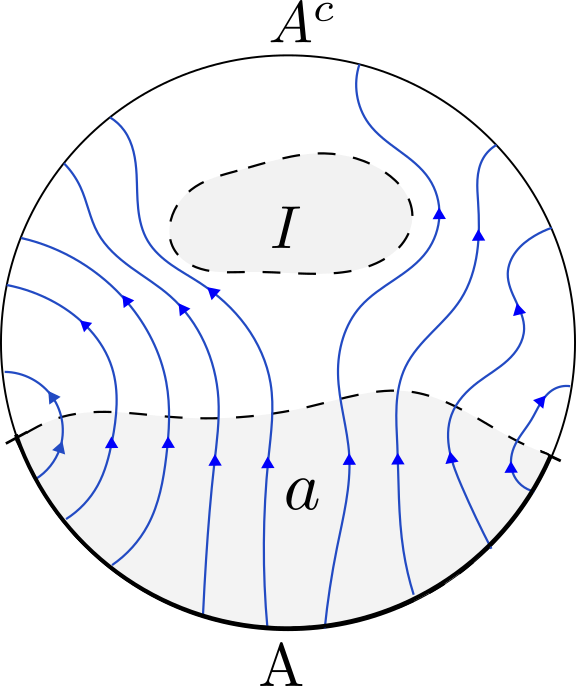}
    \caption{A maximal flow for the loose, divergenceless, cutoff-independent prescription when there is an island. The flow avoids the island.}
    \label{fig:divergenceless}
\end{figure}

The fact that $v = 0$ on $\del I$ implies that $\int_A v = \int_{\eth a} v$ and leads to an apparent tension between $|v| \lesssim 1/4G_N$ and  
\bne  \max_v \int_A v = \Sgen  (a \cup I) \geq \frac{|\eth a| + |\del I|}{4G_N}. \label{eq:somte} \ene 
There would be a contradiction if $|v| \lesssim 1/4G_N$ implied that $\int_A v = \int_{\eth a} v \leq \frac{|\eth a|}{4G_N}$, but it does not, only that $\int_{\eth a} v \lesssim \frac{|\eth a|}{4G_N}$. While the subleading terms from the bulk entropy in this inequality are smaller than $|\eth a|/4G_N$, they can be large, such as an IR divergence from the infinite volume of $a$, and sufficiently large to avoid the tension with $\int_{\eth a} v \geq \frac{|\eth a| + |\del I|}{4G_N}$.

\subsection{Proofs of equivalence to the QES prescription}
In this subsection, we will prove that our cutoff-independent flow prescriptions are equivalent to the QES prescription. We will do so using Lagrange dualisation.
\subsubsection{Loose flow prescription}
The primal problem is the loose, cutoff-independent flow prescription
\bne \sup_v \int_A v \qquad \text{subject to } \int_{{\eth r}} |v| - \int_r \nabla \cdot v \leq \Sgen  ({r}) \label{eq:concon} 
 \ene
and the Lagrangian for this problem is
\bne L[v,\mu] = \int_A v  + \int_{\cR_A} d\mu(r) \left( \int_\Sigma \left( \chi_r  \nabla \cdot v - |v| |d \chi_r| \right ) + \Sgen  (r) \right ) .\ene

We rewrite this Lagrangian by integrating the $\nabla \cdot v$ term by parts:
\bne L[v,\mu] = \int_A v - \int_{\del \Sigma} \psi v - \int_{\cR_A} d\mu \left( \int_\Sigma \left(d\chi_r \cdot v + |v| |d\chi_r| \right) - \Sgen  (r) \right), \qquad \psi := \int_{\cR_A} d\mu \chi_r .\ene

To reach the dual problem, we maximise with respect to $v$ in the bulk and on the boundary. In the bulk,
\bne \sup_v  \left ( - \int_{\cR_A} d\mu \int_\Sigma \left(d\chi_r \cdot v + |v| |d\chi_r| \right)  \right) = 0 \ene
because
\bne  \left |  \int_{\cR_A} d\mu d \chi_r \right| \leq \int_{\cR_A} d\mu |d \chi_r| \ene
for any measure $\mu$. From maximising $v$ on the boundary, we get the constraint
\bne \psi|_{\del \Sigma} = \chi_A  \label{eq:psi_constraint2}\ene
and, because $\cR_A$ is the set of bulk regions with $r\cap \del \Sigma = A$, this is equivalent to requiring $\mu$ to be a probability measure on $\cR_A$: $\int_{\cR_A} d\mu(r) = 1$.
The resulting dual problem is
\bne \inf_\mu \int_{\cR_A} d \mu (r) \Sgen  (r) , \qquad \text{subject to } \int_{\cR_A} d\mu (r) = 1, \quad \mu \geq 0 . \ene 
This is the QES prescription.

\subsubsection{Strict flow prescription: flow to surface dualisation}\label{sec:sfpfts}
The primal problem is the strict cutoff-independent flow prescription
\bne \sup_v \int_A v \qquad \text{subject to } \int_{\eth r}|v| + \left|\int_r \nabla \cdot v \right| \leq \Sgen (r) \label{eq:qsrgi} \ene
for which the Lagrangian is 
\begin{multline}
L = \int_A v + \int_{\cR} d\mu_+ (r) \left(\Sgen (r) - \int_{\eth r} |v| + \int_r \nabla\cdot v \right) + \\
\int_{\cR} d\mu_- (r) \left(\Sgen (r) - \int_{\eth r} |v| - \int_r \nabla\cdot v \right). 
\end{multline}
We use the divergence theorem on the $\nabla \cdot v$ terms, then rewrite the bulk $v$ terms as an integral over the whole Cauchy slice using
\bne -\int_{\eth r} |v| + \int_{\eth r} v = \int_\Sigma \left ( - |d\chi_r||v| + v \cdot d\chi_r \right). \ene 

We get the dual problem by maximising with respect to $v$, which in the bulk is unbounded, using the general result that
\bne \sup_v \int_\Sigma \left (A\cdot v - B |v| \right) = +\infty, \qquad \text{unless } |A| \leq B, \ene
so the supremum is unbounded unless
\bne  \left | \int_\cR d\mu_+ d\chi_r - \int_\cR \mu_- d\chi_r \right | \leq \int_\cR d\mu_+  |d\chi_r| + \int_\cR d\mu_- |d\chi_r| \ene
but this is identically true.

The dual problem is therefore
\bne \inf_{\mu_+, \mu_-} \left( \int_{\cR} d\mu_+ \Sgen (r) + \int_{\cR} d\mu_- \Sgen (r)\right) \quad \text{subject to:} \int_\cR d\mu_+ \chi_{r\cap \del \Sigma} - \int_\cR d\mu_- \chi_{r\cap \del \Sigma} = \chi_A 
\ene
with $\mu_\pm$ non-negative. 

The optimal $\mu_-$ is zero: using bulk SSA, we can assume that the regions on which $\mu_+$ and $\mu_-$ have support are mutually disjoint, and that the regions within each measure are nested. Then, to satisfy the constraint, $\mu_+$ must be a probability measure with support on $\cR_A$ alone, and $\mu_-$ can only have support on $\cR_{\varnothing}$, the set of bulk regions with $r \cap \del \Sigma = \varnothing$,
and the optimal $\mu_-$ is clearly zero. The dual problem becomes
\bne \inf_{\mu_+} \left( \int_{\cR_A} d\mu_+ \Sgen (r)\right) \qquad \text{subject to }\int_{\cR_A} d\mu_+ = 1
\ene
which is the QES prescription.
\subsubsection{Strict flow prescription: surface to flow dualisation}
Here we will dualise from the QES prescription to the strict flow prescription. 
Since we already proved the equivalence between the prescriptions in section~\ref{sec:sfpfts}, this derivation is redundant, but we include it to show how it can be done because surface-to-flow dualisations are needed when deriving covariant prescriptions~\cite{CQBT}.

The primal problem is the QES prescription, written in the following way:
\begin{multline}
\inf_\mu \int_\cR |d\mu| \Sgen  (r) \qquad \text{ subject to:}\\
\int_\cR d\mu \chi_r \big|_{\del \Sigma} = \chi_A, \quad |w| - \int_\cR |d\mu| |d\chi_r| \leq 0, \quad w= \int_\cR d\mu d\chi_r \,. \label{eq:noqth} 
\end{multline}
$\mu$ is not constrained to be non-negative. The inequality constraint and the last constraint together enforce a mathematical identity, but without imposing it as a constraint, one does not derive the flow prescription. Also, the inequality constraint is not convex in $\mu$, so, strictly speaking, we are not guaranteed strong duality, but if one replaces the inequality with the pair of convex constraints, $\xi = \int_\cR |d\mu||d\chi_r|$ and $|w| - \xi \leq 0$, then one gets the same result.

The Lagrangian for this problem is
\begin{multline}
L = \int_\cR |d\mu| \Sgen (r) + \int_{\del \Sigma} \lambda \left(\int_\cR d\mu \chi_{r\cap \del \Sigma} - \chi_A\right) \\
+ \int_\Sigma \lambda'\left(|w| - \int_\cR |d\mu||d\chi|\right) + \int_\Sigma v\cdot\left(\int_\cR d\mu d\chi_r - w\right). 
\end{multline}

The dual problem is
\bne \sup_{\lambda, v} \left ( - \int_A \lambda \right) \qquad \text{subject to } \forall r: \int_{\eth r} |v| + \left | \int_{\del \Sigma \cap r} (\lambda + n\cdot v) + \int_r \nabla \cdot v \right | \leq \Sgen (r).\label{eq:noqth2} \ene
We can show that there exists an optimal $\lambda$ equal to $- n\cdot v$ on $\del \Sigma$. A given $\lambda$ is optimal for~\eqref{eq:noqth2} if the supremum over $v$ for that $\lambda$ gives $S(A)$. If $\lambda = - n\cdot v$ on $\del \Sigma$ then~\eqref{eq:noqth2} becomes the strict cutoff-independent quantum bit thread prescription~\eqref{eq:qsrgi}, whose supremum over $v$ we proved in section~\ref{sec:sfpfts} is $S(A)$. So, we are free to set $\lambda$ equal to $- n\cdot v$ on $\del \Sigma$ in~\eqref{eq:noqth2} without affecting the optimum, and we get~\eqref{eq:qsrgi}.

\subsubsection{Divergenceless loose flow prescription} \label{sec:dsfpd}
The divergenceless, loose, cutoff-independent prescription is
\bne \max_v \int_A v \qquad \text{ subject to } \nabla\cdot v = 0, \quad \forall r \in \RR_A : \int_{\eth r} |v| \leq \Sgen (r) .\label{eq:clarg} \ene
This is the same as the loose prescription with the additional constraint that $v$ is divergenceless, so it is stricter than the loose prescription, but it is neither stricter nor looser than the strict prescription~\eqref{eq:strcu}. 

Since we have not proven that there always exist optimal flows to the loose prescription that are divergenceless, we do not know for certain that~\eqref{eq:clarg} is equivalent to the QES prescription. We dualise~\eqref{eq:clarg} to find an equivalent problem, which will tell us when the divergenceless loose prescription is equivalent to the QES prescription. 

The Lagrangian for the problem~\eqref{eq:clarg} is
\bne L = \int_A v + \int_{\cR_A} d\mu (r)\left(\Sgen (r) - \int_\Sigma |v| |d\chi_r|\right) + \int_\Sigma \lambda(x) \nabla \cdot v \ene
and from this, the dual problem is
\bne \min_{\mu,\lambda} \int_{\cR_A} d\mu (r) \Sgen (r) \qquad \lambda(x)|_{\del \Sigma} = \chi_A , \quad |d\lambda | \leq \int_{\cR_A} d\mu (r) |d\chi_r|. \label{eq:duarg} \ene
We have derived that our loose and divergenceless flow prescription~\eqref{eq:clarg} is equivalent to~\eqref{eq:duarg}, so we want to know whether~\eqref{eq:duarg} is QES-equivalent, and this, as we will now see, depends on $\Sigma$ and the choice of $A$. 

If $A$ has a boundary, then~\eqref{eq:duarg} is equivalent to the QES prescription.
To show this, it is useful to note that an implication of the constraints in~\eqref{eq:duarg} is
\bne |d\chi_A| \leq \int_{\cR_A} d\mu(r) |d\chi_{r\cap \del \Sigma}|. \label{eq:abdyh} \ene
If $A$ has a boundary, then $|d \chi_A|$ is a delta-function on that boundary. Then, from~\eqref{eq:abdyh}, we must have $\int_{\cR_A} d\mu \geq 1$, and so~\eqref{eq:duarg} is lower bounded by $S(A)$. But it is also upper bounded by $S(A)$, because $\lambda = \chi_{r_{QES}}$ and $\mu(r) = \delta(r-r_{QES})$ is always feasible. So,~\eqref{eq:duarg} gives $S(A)$, if $A$ has a boundary.
	
In contrast, if there exists a feasible $\lambda$ with $d\lambda = 0$, then any $\mu$ is feasible, including $\mu = 0$, so the minimum of~\eqref{eq:duarg} is zero, even when $S(A) \neq 0$. This happens, for example, when $\Sigma$ is a set of disjoint regions $\bigcup_i \Sigma_i$, and $d\chi_A = 0$, such as when $A = \del \Sigma_1$; then $\lambda = \chi_{\Sigma_1}$ is feasible and has $d\lambda = 0$ so $\mu=0$ is feasible, for which the objective function vanishes.

An example: consider the case where $\Sigma$ has two asymptotic AdS boundaries. If the bulk of $\Sigma$ is two disconnected components, and $A$ is the whole of one of the AdS boundaries, then $\mu = 0$ is feasible and~\eqref{eq:duarg} is not QES-equivalent. But if we connect the disconnected components of $\Sigma$, with even an arbitrarily narrow wormhole, then $\lambda$ must have a non-zero gradient somewhere in the bulk, and we must have $\int_{\RR_A} d\mu \geq 1$, so~\eqref{eq:duarg} becomes QES-equivalent%
\footnote{Loosely speaking, in double holography, we can think of the loose, divergenceless prescription as when we push all the threads in the highest dimensional bulk onto the brane so that, from the brane perspective, the flow is divergenceless. This makes it clear why even an arbitrarily narrow wormhole makes this prescription work; in double holography, if we have two AdS-branes which are only connected through the highest-dimensional bulk, then there is no way to push all the threads out of the bulk and onto the branes.}.
We also get QES-equivalence if we make $A$ anything less than the whole of one of the AdS boundaries, or add a subregion from the other AdS boundary. The conclusion is that~\eqref{eq:duarg} is generically QES-equivalent.

\begin{figure}[h!]
    \centering
    \includegraphics[width=0.4\linewidth]{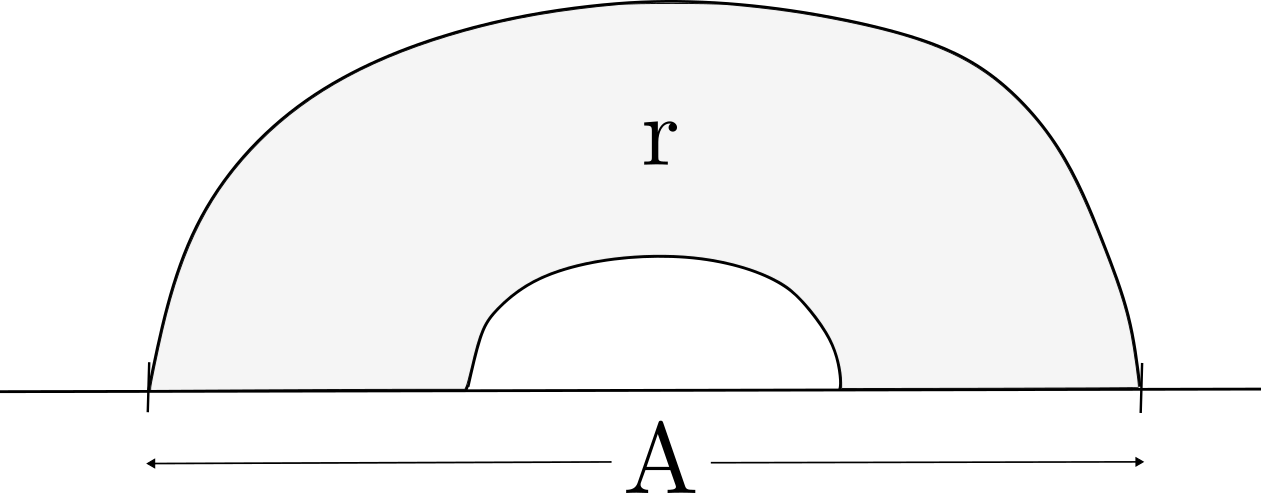}
    \caption{A bulk region $r \notin \RR_A$ that satisfies $|d\chi_{r\cap \del\Sigma}| \geq |d\chi_A|$.}
    \label{fig:StrictDivergencelessCounterexample}
\end{figure}

As mentioned in section~\ref{sec:other_prescriptions}, there is also a divergenceless, \textit{strict}, cutoff-independent prescription that one can consider, but this prescription is not QES-equivalent. The prescription is~\eqref{eq:clarg} but with the norm bound applied to the larger set of bulk subregions $\forall r \in \RR$. If one dualises this prescription, then one arrives at~\eqref{eq:duarg} with $\RR_A \mapsto \RR$, and this allows $\mu$ to have support on regions $r$ that are not in $\RR_A$, such as that shown in Fig.~\ref{fig:StrictDivergencelessCounterexample}, and these will generically allow for $\Sgen (r) < S(A)$ in~\eqref{eq:duarg}. As far as QES-equivalence is concerned, from the flow perspective, imposing the strict constraint is fine; the strict prescription allows for flows with $\int_A v = S(A)$, but imposing divergencelessness as well is a step too far.

\section{Quantum thread distributions} \label{sec:qdist}

Thread distributions, introduced in \cite{Headrick:2022nbe}, are an alternative mathematical representation of bit threads, replacing the flow $v$ with a measure $\muth$ over bulk curves connecting boundary points. Although they were defined both in the static and covariant settings, here we will focus (as we do throughout this paper) with the static case, and explain a possible method to reproduce the QES formula using thread distributions. Unfortunately, we will have to leave several key statements as conjectures.

We first recall the basic facts about classical thread distributions. A classical thread $p$ is a connected bulk curve. While it is perhaps conceptually more satisfying to have the threads be unoriented, here it will be notationally convenient to make them oriented. Letting $\PP$ be a set of classical threads, a \emph{classical thread distribution (TD)} $\muth$ is a (non-negative) measure on $\mathcal{P}$ obeying the density bound
\be\label{cldensitybound}
\forall x\in \Sigma,\quad \int_{\mathcal{P}} d\muth(p)\,\Delta(x,p)\le\frac1{4\GN}\,,
\ee
where $\Delta(x,p)$ is a delta function on $\Sigma\times\PP$ that clicks when $p$ passes through $x$. If we let $\PP_A$ be the set of curves starting on $A$ and ending on $A^c$, then maximizing the total number of threads $\muth(\PP_A)$ gives the RT formula. This can be shown by dualizing the program of maximizing $\muth(\PP_A)$ subject to \eqref{cldensitybound}, which yields the relaxed min cut program. A classical TD $\muth$ can also be converted into a classical flow $v$ and vice versa, where the threads are essentially the field lines of the flow; under this mapping, the objectives agree:
\be
\muth(\PP_A)=\int_An\cdot v\,.
\ee

In this section, we will look for a thread-distribution analogue of the strict quantum flows studied in section \ref{sec:strictflows}. In that case, the effect of bulk entanglement was to allow the flow to have sources and sinks. We can therefore guess that the threads will jump between points in the bulk. In doing so, they will carry a key extra piece of information beyond what the flows carry, namely, which point is connected to which point. Specifically, whereas a flow does not associate a given sink with any particular source, a thread does jump from a particular point to a specific other point.

For simplicity, throughout this section we will assume that the full bulk $\Sigma$ is in a pure state, $S_{\rm b}(\Sigma)=0$. This assumption is without loss of generality, since we can always formally adjoin extra purifying regions to the bulk and boundary.

\subsection{Entanglement pair functions}

In upgrading classical TDs to quantum ones, we will need the concept of an entanglement pair function, which represents the bulk entropies and plays the same role for quantum TDs that the entanglement density function (EDF), defined and studied in subsection \ref{sec:subcontours} and section \ref{sec:entropohedron}, plays for quantum flows. We define an \emph{entanglement pair function} (EPF) as a non-negative function $g$ on $\Sigma^2:=\Sigma\times\Sigma$ such that\footnote{Because our application of EPFs will be in holography, in defining EPFs, we are using the notation corresponding to bulk entropies in holography. However, the concept can be applied to any state of a field theory for which entropies of spatial regions are defined, or even more generally, to any multiparty quantum state. Note also that \eqref{EPFdef} is a property shared by the so-called two-point partial entanglement entropy $\mathcal{I}(x,y)$, defined in \cite{Lin:2023rxc}.}
\be\label{EPFdef}
\forall r\in\RR\,,\quad
\int_rdx\int_{r^c}dy\,(g(x,y)+g(y,x))\le S_{\rm b}(r)\,.
\ee
Clearly, only the symmetric part $g(x,y)+g(y,x)$ enters in the definition, so we could have required $g$ to be symmetric. However, when we relate EPFs to EDFs and to threads, it will be handy to also have the antisymmetric part $g(x,y)-g(y,x)$.

Just as with EDFs, knowing the full set of EPFs is equivalent to knowing $S_{\rm b}(r)$ for all $r\in\RR$. Like the EDFs, the EPFs form a convex set that includes 0. (However, unlike the EDFs, since $g$ is required to be non-negative, this set is not symmetric about 0.) As we will see in subsection \ref{sec:QTDdh}, for any classical TD in the second bulk, the distribution of endpoint pairs defines an EPF.

EPFs and EDFs have a direct relation to each other. An EPF $g$ can be converted into an EDF by integrating over one of the points in the pair:
\be
f_{g}(x):=\int_{\Sigma}dy\,(g(x,y)-g(y,x))\,.
\ee
To see that $f_{g}$ is indeed an EDF, integrate over any region $r\in\RR$:
\be
\int_r dx\,f_{g}(x)=\int_rdx\int_\Sigma dy\,(g(x,y)-g(y,x)) =\int_rdx\int_{r^c} dy\,(g(x,y)-g(y,x))\,;
\ee
therefore
\be
\left|\int_rdx\,f_{g}(x)\right|
\le\int_rdx\int_{r^c}dy\,\left|g(x,y)-g(y,x)\right|
\le\int_rdx\int_{r^c}dy\,\left(g(x,y)+g(y,x)\right)
\le S_{\rm b}(r)\,.
\ee

According to the following conjecture, conversely, every EDF comes from an EPF:
\begin{conjecture}\label{conj:EDFEPF}
Given an EDF $f$, there exists an EPF $g$ such that $f_g=f$, and $g(x,y)=0$ unless $f(x)>0$ and $f(y)<0$.
\end{conjecture}
\noindent We will use this conjecture below when we discuss converting quantum flows into quantum TDs. In subsection \ref{sec:QTDdh}, we will show that conjecture \ref{conj:EDFEPF} holds in double holography. This implies that the conjecture holds in any state whose entropies are in principle derivable from a second bulk. However, there are states, such as the $n$-party GHZ state for $n\ge4$, that are not derivable from any bulk, as they violate the MMI inequality \eqref{MMI}. Nonetheless, it's easy to show by explicit construction that conjecture \ref{conj:EDFEPF} holds for GHZ states as well (for example, for the extremal EDF given in \eqref{GHZEDF}, we can simply take $g(a,b)=\ln 2$, and 0 for all other arguments). These examples support the validity of conjecture \ref{conj:EDFEPF}.

As mentioned above, the set of all EPFs contains the same information as the set of all EDFs (namely, the entropy of every bulk region). And if the above conjecture is correct, then EPFs and EDFs can be converted into each other. However, being a function of two variables rather than one, an individual EPF can carry more information than an individual EDF. Consider, for example, a state on 4 qubits at $x_1,\ldots,x_4$, where we are given that the state consists of two Bell pairs. The EDF $f(x_1)=f(x_2)=-f(x_3)=-f(x_4)=\ln2$ tells us that the Bell pairs are either $(x_1,x_3)$, $(x_2,x_4)$ or $(x_1,x_4)$, $(x_2,x_3)$. On the other hand, the EPF $g(x_1,x_3)=g(x_2,x_4)=\ln2$ resolves the ambiguity, in favor of the first option.

It would be interesting to further study the set of EPFs, for example, the extremal points and under what conditions they can saturate a given set of regions, etc., similar to the investigations concerning EDFs in section \ref{sec:entropohedron}.

\subsection{Quantum thread distributions: definition}
\label{sec:QTDs}

We define a \emph{quantum thread} $p$ as a sequence $p_1,\ldots,p_n$ of connected curves in $\Sigma$, which we call \emph{segments}, such that the starting point of the first segment $p_1$ and ending point of the last one $p_n$ are on the boundary $\partial\Sigma$.\footnote{As with classical threads, one again has the choice to make the quantum threads oriented or unoriented. Conceptually, one can think of them as being unoriented, in the sense that $p$ should be identified with its reverse thread $p'$, where $p'_i$ is the reverse of $p_{n-i+1}$. However, notationally, it will be simpler to maintain the orientation.} Writing $s_i$, $e_i$ as the starting and ending points respectively of the segment $p_i$, we have $s_1,e_n\in\partial\Sigma$. We say that the pair of points $(x,y)\in\Sigma^2$ is a \emph{jump} of $p$ if there is an $i$ such that $e_i=x$ and $s_{i+1}=y$. Given a thread $p$ and a region $r\in\RR$, we define $N(r,p)$ as the number of times $p$ jumps into or out of $r$, in other words the number of jumps $(x,y)$ with $x\in r$, $y\in r^c$ or $x\in r^c$, $y\in r$.

Let $\PPq$ be the set of all quantum threads. A \emph{quantum TD} is a measure $\muth$ on $\PPq$ such that\footnote{This is the analogue of the strict quantum flows, since we are imposing the bound on the number of jumps on all regions $r\in\RR$. The analogue of the loose quantum flows would be to replace $\RR$ in \eqref{jumpbound} with $\RR_A$.}
\begin{align}\label{densitybound}
\forall\,x\in\Sigma\,,\quad&\int_\PPq d\muth(p)\,\Delta(x,p)\le \frac1{4\GN}\\
\label{jumpbound}
\forall\,r\in\RR\,,\quad&\int_\PPq d\muth(p)\,N(r,p)\le S_{\rm b}(r)\,.
\end{align}
For any measure $\muth$ on $\PPq$, the jumps of the threads define a function on $\Sigma^2$; specifically, letting $\Delta_{\rm jump}(p,x,y)$ be a delta function on $\PPq\times\Sigma^2$ that clicks when $p$ contains a jump $(x,y)$, we define
\be\label{gmudef}
g_\muth(x,y):=\int_\PPq d\muth(p)\,\Delta_{\rm jump}(p,x,y)\,.
\ee
Then \eqref{jumpbound} is equivalent to the statement that $g_\muth$ is an EPF. In the rest of this section, we will give evidence that this definition of quantum TDs is sensible and has the properties one would expect for a representation of boundary entanglement that takes into account bulk entanglement.

Let $A$ be a boundary region, and $\PPqA$ the set of quantum threads starting on $A$ and ending on $A^c$. Given a bulk region $r\in\RR_A$, a quantum thread $p\in\PPqA$ must leave $r$ at least once, either by one of its segments crossing $\eth r$ or by jumping across it. (It may leave $r$ multiple times if it re-enters it.) Given a quantum TD $\muth$, the total number that cross $\eth r$ cannot exceed $|\eth r|/(4\GN)$ by \eqref{densitybound}, while the total number that jump across cannot exceed $S_{\rm b}(r)$ by \eqref{jumpbound}. Therefore,
\be\label{muPPbound}
\muth(\PPqA)\le S_{\rm gen}(r)\,.
\ee
Since \eqref{muPPbound} holds for all $r\in\RR_A$, we have
\be\label{QTDbound}
\muth(\PPqA)\le S(A)\,.
\ee
We expect this bound to be achieved, and therefore that quantum TDs correctly compute boundary entropies:
\begin{conjecture}\label{conj:qtdmax}
\be
S(A) = \max\muth(\PPqA)\quad\text{over quantum TDs $\muth$}\,.
\ee
\end{conjecture}
\noindent In the next subsection, we will show that conjecture \ref{conj:qtdmax} is implied by conjecture \ref{conj:EDFEPF}, and in subsection \ref{sec:QTDdh}, we will show that it holds in double holography.

In an optimal configuration, both the number crossing the QES $\eth r_A$ and the number jumping across must be maximal, so both bounds \eqref{densitybound} and \eqref{jumpbound} must be saturated there:
\begin{align}
\forall x\in\eth r_A\,,\quad \int_{\PPqA}d\muth(p)\,\Delta(x,p)&=\frac1{4\GN} \label{crosssat}\\
\int_{\PPqA}d\muth(p)\,N(r_A,p)&=S_{\rm b}(r_A)\,.\label{jumpsat}
\end{align}
Furthermore, every thread either crosses the QES once or jumps across it once, contributing to the integral in either \eqref{crosssat} or \eqref{jumpsat} (in the latter case, with $N(r_A,p)=1$).

\subsection{Converting between quantum thread distributions \& quantum flows}

\subsubsection{From thread distributions to flows} As in the classical case, we can convert a quantum TD $\muth$ into a quantum strict flow $v_\muth$:
\be
v_\muth(x) := \int_\PPq d\muth(p)\,\Delta(x,p)\,\dot x\,,
\ee
where $\dot x$ is the unit tangent vector to $p$ at $x$. By the Cauchy-Schwarz inequality together with the density bound \eqref{densitybound}, $v_\muth$ obeys the norm bound, $|v_\muth|\le1/(4\GN)$. Due to the jumps in the threads, $v_\muth$ may have a non-zero divergence:
\be
\nabla\cdot v_\muth(x) = 
\int_\Sigma dy\int_\PPq d\muth(p)
\,(\Delta_{\rm jump}(p,y,x)-\Delta_{\rm jump}(p,x,y))\,.
\ee
Therefore, for any region $r\in\RR$,
\be
\int_r dx\,\nabla\cdot v_\muth(x) = 
\int_r dx\int_{r^c} dy\int_\PPq d\muth(p)
\,(\Delta_{\rm jump}(p,y,x)-\Delta_{\rm jump}(p,x,y))\,,
\ee
so
\be
\left|\int_r dx\,\nabla\cdot v_\muth(x) \right|\le 
\int_rdx\int_{r^c}dy\,g_\muth(x,y)\le S_{\rm b}(r)\,.
\ee
Hence $v_\muth$ also obeys the strict divergence bound \eqref{nablavEDF}. Finally, the flux of $v_\muth$ through $A$ is given by the difference between the number of threads starting on $A$ and ending on $A^c$ and the number going the other way:
\be
\int_An\cdot v_\muth=\muth(\PPqA)-\muth(\PP^{\rm q}_{A^c})\,.
\ee

\subsubsection{From flows to thread distributions}
To convert a quantum flow $v$ to a quantum TD $\muth_v$, we will need conjecture \ref{conj:EDFEPF}. The vector field $v$ induces a measure $\muth'_v$ on segments $p'$ such that, for any oriented surface $\sigma$,
\be
\int d\muth'(p')\,N(\sigma,p')=\int_\sigma n\cdot v\,,
\ee
where $N(\sigma,p')$ is the net number of times $p'$ crosses $\sigma$ in the direction of the unit normal vector $n$. Roughly speaking, the segments $p'$ are the flow lines of $v$, starting at sources and ending at sinks. Therefore, for an infinitesimal neighborhood of a point $x\in\Sigma$ of volume $dV$, $\nabla\cdot v(x)dV$ equals the number of segments $p'$ that start in the neighborhood (or, if it is negative, minus the number that end in it). Conjecture \ref{conj:EDFEPF} can then be used to tie the sources and sinks to each other. Since $\nabla\cdot v$ is an EDF, we can apply conjecture \ref{conj:EDFEPF} to it, yielding an EPF $g(x,y)$ such that
\be
\forall \,x\in \Sigma_\pm\,,\quad\nabla\cdot v(x)=\pm\int_{\Sigma_\mp}dy\,g(x,y)\,,
\ee
where
\be
\Sigma_+:=\{x\in\Sigma:\nabla\cdot v>0\}\,,\qquad
\Sigma_-:=\{x\in\Sigma:\nabla\cdot v<0\}\,.
\ee
Therefore, for every sink, in other words, every $x\in\Sigma_-$, $g(x,y)$ defines a distribution of jumps $(x,y)$ over points $y\in\Sigma_+$, which we can use to connect the segments ending at $x$ to segments starting at positions $y\in\Sigma_+$. With this prescription, the total number of jumps ending at a given $y\in\Sigma_+$ is $\int_{\Sigma_+}dx\,g(x,y)=\nabla\cdot v(y)$, equal to the number of segments starting at $y$. Thus, the bulk starting and ending points of all segments are accounted for by jumps. With these jumps, every segment becomes a member of some sequence. Those sequences that do not start and end on the boundary (e.g.\ that form loops) are discarded, leaving a quantum TD $\muth_v$. The flux of $v$ through a given boundary region $A$ is related to the TD $\muth_v$ in the same way as for $v_\muth$ and $\muth$ from the previous paragraph, namely
\be
\int_An\cdot v=\muth_v(\PPqA)-\muth_v(\PP^{\rm q}_{A^c})\,.
\ee

Given a maximal quantum flow $v$, in view of \eqref{QTDbound}, necessarily $\muth_v(\PPqA)=S(A)$ and $\muth_v(\PP^{\rm q}_{A^c})=0$. Therefore, conjecture \ref{conj:EDFEPF} implies conjecture \ref{conj:qtdmax}.

\subsection{Double holography}
\label{sec:QTDdh}

As with quantum flows, we can use doubly holographic setups, in which the bulk entropy $S_{\rm b}$ is derived geometrically via the RT formula in the second bulk, to gain intuition about quantum TDs. Here, the quantum threads in the first bulk will be related to classical threads in the full (first plus second) bulk. We use the same setup and notation as in subsection \ref{sec:flowsdouble}.

We first note that any classical TD $\tilde\muth$ in the second bulk $\tilde\Sigma$ defines an EPF $g_{\tilde\muth}$ in the first bulk.\footnote{This statement and its converse, conjecture \ref{conj:EPFlift}, are not fundamentally statements about double holography. Rather, the statements are that, in holography, any classical TD in the bulk defines an EPF for the boundary, and any boundary EPF lifts to a classical TD. Here, we use the language and notation of double holography, as that corresponds to our present application.} Letting $\tilde\PP$ be the set of classical threads in $\tilde\Sigma$, we defined $\Delta_{\rm end}(\tilde p,x,y)$ as a delta function on $\tilde\PP\times\Sigma^2$ that clicks when $\tilde p$ begins at $x$ and ends at $y$. Then,
\be\label{gtildemudef}
g_{\tilde\muth}(x,y):=\int_{\tilde\PP}d\tilde\muth(\tilde p)\,\Delta_{\rm end}(\tilde p,x,y)\,.
\ee
To see that $g_{\tilde\muth}$ is indeed an EPF, note that, since $\tilde\muth$ is a classical TD, for any surface $\sigma$ in $\tilde\Sigma$, the total number of times the threads intersect $\sigma$ is bounded by its area:
\be\label{intbound}
\int_{\tilde\PP}d\tilde\muth(\tilde p)\,\#(\tilde p\cap\sigma)\le\frac{|\sigma|}{4\widetilde{\GN}}\,.
\ee
Now fix a region $r\in\RR$, and let $\tilde\PP_r$ be the set of threads in $\tilde\PP$ that start in $r$ and end in $r^c$ and $\tilde\PP_{r^c}$ the set that goes the other way. Given any surface $\sigma$ in $\tilde\Sigma$ homologous to $r$, any thread $\tilde p$ in $\PP_r$ or $\PP_{r^c}$ must intersect $\sigma$ at least once. Therefore,
\be\label{gbound}
\int_rdx\int_{r^c}dy\,(g_{\tilde\muth}(x,y)+g_{\tilde\muth}(y,x))=
\int_{\tilde\PP_r}d\tilde\muth(\tilde p)
+\int_{\tilde\PP_{r^c}}d\tilde\muth(\tilde p)
\le\int_{\tilde\PP}d\tilde\muth(\tilde p)\,\#(\tilde p\cap\sigma)\,.
\ee
Applying \eqref{intbound} and \eqref{gbound} to the RT surface for $r$ shows that $g_{\tilde\muth}$ obeys the definition \eqref{EPFdef} of an EPF.

It seems reasonable to believe that the converse holds, namely any EPF $g$ arises from some TD on $\tilde\Sigma$. However, we do not have a proof, so we leave it as a conjecture:
\begin{conjecture}\label{conj:EPFlift}
Given an EPF $g$ on $\Sigma$, there exists a classical TD $\tilde\muth$ on $\tilde\Sigma$ such that $g_{\tilde\muth}=g$.
\end{conjecture}
\noindent This is the analogue of theorem \ref{thm:classicalflowEDF}, which states that any EDF $f$ on $\Sigma$ can be lifted to a classical flow $\tilde v$ on $\tilde\Sigma$ such that $\tilde n\cdot\tilde v=f$ (where $\tilde n$ is the inward-directed unit normal to $\Sigma$ in $\tilde\Sigma$).

In the double holography setting, making use of the above mapping from classical TDs to EPFs, we can also prove conjecture \ref{conj:EDFEPF}. Namely, starting from an EDF $f$ on $\Sigma$, we lift it to a classical flow $\tilde v$ in $\tilde\Sigma$ using theorem \ref{thm:classicalflowEDF}; we then map $\tilde v$ to a classical TD $\tilde\muth$ on $\tilde\Sigma$ via the flow lines of $\tilde v$, as described in subsection 5.1.1 of \cite{Headrick:2022nbe} (except that, rather than restricting to threads connecting specified boundary regions, we keep all threads connecting any two boundary points, and we retain their orientations); finally, using \eqref{gtildemudef} we map $\tilde\muth$ to an EPF $g$. Through this series of transformations, $f(x)$ becomes the boundary flux $\tilde n\cdot\tilde v(x)$, which in turn becomes the number of threads beginning (if $f(x)>0$) or ending (if $f(x)<0$) at $x$, finally implying that $g$ obeys $\int_\Sigma dy\,g(x,y)=|f(x)|$.

The generalization of classical TDs to double holography is fairly straightforward. In double holography, a \emph{double thread} $\hat p$ is a connected curve in $\Sigma\cup\tilde\Sigma$ starting and ending on $\partial\Sigma$. Along the way, $\hat p$ may pass from $\Sigma$ to $\tilde\Sigma$ and vice versa an arbitrary number of times. It will be convenient to require the first and last segments of $\hat p$ to lie in $\Sigma$ (not $\tilde\Sigma$); this is without loss of generality since we can simply append infinitesimal segments in $\Sigma$ at the beginning and end. Defining $\hat\PP$ as the set of all double threads, a \emph{double TD} is a measure $\hat\muth$ on $\hat\PP$ subject to density bounds on both $\Sigma$ and $\tilde\Sigma$:
\begin{align}
&\forall x\in \Sigma,\quad \int_{\hat\PP} d\hat\muth(\hat p)\,\Delta(x,\hat p)\le\frac1{4\GN} \label{densbound1}\\
&\forall \tilde x\in\tilde \Sigma,\quad \int_{\hat\PP} d\hat\muth(\hat p)\,\Delta(\tilde x,\hat p)\le\frac1{4\widetilde{\GN}}\,.\label{densbound2}
\end{align}
Let $\hat\PP_A$ be the set of double threads that start on $A$ and end on $A^c$. The maximum number of such double threads, $\hat\muth(\hat\PP_A)$, equals $S(A)$, either by dualization to obtain the relaxed min cut program or by converting between double TDs and double flows. (We leave the proof as an exercise for the reader.)

We will now show that any double TD can be converted into a quantum TD. A double thread $\hat p$ consists of segments in $\Sigma$ and segments in $\tilde\Sigma$. The former segments together define a quantum thread $p$ in $\Sigma$, while each of the latter segments is a classical thread in $\tilde\Sigma$. Therefore, a measure $\hat\muth$ on $\hat\PP$ induces a measure $\muth$ on the set $\PPq$ of quantum threads together with a measure $\tilde\muth$ on the set $\tilde\PP$ of classical threads on $\tilde\Sigma$. These are related by the fact that the distribution of jumps of the former equals the distribution of endpoints of the latter, so
\be\label{jumpsends}
g_\muth=g_{\tilde\muth}=:g\,.
\ee
The constraint \eqref{densbound2} is simply the statement that $\tilde\muth$ is a classical TD, hence $g$ is an EPF. Meanwhile, the constraint \eqref{densbound1} is equivalent to \eqref{densitybound}. So $\muth$ is a quantum TD. The fact that double TDs reduce to quantum TDs is evidence that the definition we gave for the latter is the physically relevant one.

In particular, a maximal double TD $\hat\muth$, which has $\hat\muth(\hat\PP_A)=S(A)$, induces a quantum TD $\muth$ that saturates the bound \eqref{QTDbound}, i.e.\ such that $\muth(\PPqA)=S(A)$. Therefore, conjecture \ref{conj:qtdmax} holds in double holography.

Conjecture \ref{conj:EPFlift} would allow us to go the other way, and convert any quantum TD $\muth$ to a double TD, by lifting the EPF $g_\muth$ to a classical TD $\tilde\muth$ on $\tilde\Sigma$ such that $g_{\tilde\muth}=g_\muth$. We can then replace every jump $(x,y)$ in each quantum thread $p$ by a classical thread in $\tilde\Sigma$ going from $x$ to $y$, thereby producing a double thread.

In order to put the notions of EPF and of quantum TD that we have defined in this section on a solid footing, it would be desirable to prove the three conjectures. We leave these investigations to future work.

\section{The entropohedron}
\label{sec:entropohedron}

This section may be read independently of the rest of the paper.

In the definition of strict quantum flows given in section\ \ref{sec:strictflows}, the concept of \emph{entanglement distribution function} (EDF) played a key role. In that setting, the relevant entropies were those of spatial regions of the bulk in a holographic spacetime. The notion of EDF is also applicable for a quantum state on a finite number of parties. For $N$ parties, the EDFs elegantly repackage the entropy function, which is a function of $2^N-1$ variables (or equivalently the entropy vector, a vector in $2^N-1$ dimensions), into a convex polytope in just $N$ dimensions, which we will call the \emph{entropohedron}. Whereas the entropy function treats all sets of parties on an equal footing, the entropohedron makes use of the local structure, i.e.\ which sets of parties are contained within which other sets. It also naturally builds in important properties of entropies such as strong subadditivity. In this section, we give a self-contained presentation of the entropohedron and its properties. In this finite-dimensional setting, we will be able to prove the relevant properties rigorously, as no analysis is required.

Some aspects of this discussion appeared previously in the original bit thread paper \cite{Freedman:2016zud}. The entropohedron is also closely related to several objects that have been previously studied in the context of submodular functions, including the so-called \emph{symmetric submodular polytope} (see \cite{bach2013learning}) and the \emph{generalized polymatroid} (see \cite{frank2014characterizing}). The simple relation between these objects and network flows, discussed in subsection \ref{sec:graphflows}, has not been noted before as far as we know. The entropohedron is also related to the notion of \emph{entanglement contour}, as we will discuss in subsection \ref{sec:saturation}.

To avoid interrupting the narrative, we relegate the lengthier proofs to subsection \ref{sec:EDFproofs}.

\subsection{Definitions \& examples}

Everywhere in this section, $X$ will denote a finite set and $2^{X}$ its power set (the set of all subsets of ${X}$). We will use $x,y,\ldots$ to represent elements of ${X}$ and $A,B,\ldots$ to represent subsets of $X$. We denote by $\R^{X}$ the $|{X}|$-dimensional real vector space with basis vectors labelled by elements of ${X}$.

\begin{definition}
An \emph{entropy function}\footnote{$S$ is also sometimes called an \emph{entropy vector}, as it may be considered a vector in $\R^{2^{X}\setminus\{\emptyset\}}$, where we remove the empty set from the basis since we require $S(\emptyset)=0$.}$^,$\footnote{The strong subadditivity condition \eqref{props3} is also called submodularity. In the literature on submodular functions, in particular in defining the symmetric submodular polytope, the condition \eqref{props1} is also usually imposed, along with monotonicity, $S(AB)\ge S(A)$ \cite{bach2013learning}. Monotonicity is obeyed by the entropies of classical probability distributions; however, for quantum states it is replaced by the weak monotonicity condition \eqref{props4}.} for ${X}$ is a function $S:2^{X}\to[0,\infty)$, satisfying the following conditions:\footnote{Note that \eqref{props4}, applied to the case $B=A$, together with \eqref{props1}, implies $S(A)\ge0$.}
\begin{align}\label{props1}
S(\emptyset) &= 0 \\
\forall \,A,B\subseteq X \,,\quad S(A)+S(B) &\ge S(A\cap B)+S(A\cup B) \qquad\text{(SSA)} \label{props3} \\
\forall \,A,B\subseteq X \,,\quad S(A)+S(B) &\ge S(A\setminus B)+S(B\setminus A) \qquad\text{(WM)}\,.\label{props4}
\end{align}
\end{definition}

The entropies of a multipartite quantum system in a fixed state obey \eqref{props1}--\eqref{props3}, hence the name ``entropy function''.\footnote{The name ``entropy function'' is a bit of an abuse, since not every entropy function represents the entropies of a quantum state. In particular, it is known that the entropies of quantum states on at least four parties obey a further set of constrained inequalities \cite{Linden:2004ebt,Cadney:2011vix}.} For this reason, we will sometimes use the language of ``states'' and ``parties''. The only aspects of the state that are relevant for our purposes, however, are the entropies. Another natural source of functions obeying these properties is minimal cuts on weighted graphs. We will discuss these further in subsection \ref{sec:graphflows}. Of course, minimal cuts on graphs are related to quantum states via the RT formula.

\begin{definition}Given an entropy function $S$ for $X$, an \emph{entanglement distribution function} (EDF) is a function $f:X\to\R$ such that, for all $A\subseteq X $,
\be\label{EDFdef}
\left|\sum_{x\in A}f(x)\right|\le S(A)\,.
\ee
$f$ may also be thought of as a vector in the $|X|$-dimensional vector space $\R^X$.
\end{definition}
\begin{definition}
Given an entropy function $S$, the \emph{entropohedron} $F_S$ is the set of EDFs of $S$.
\end{definition}

\begin{figure}
    \centering
    \includegraphics[width=0.99\linewidth]
    {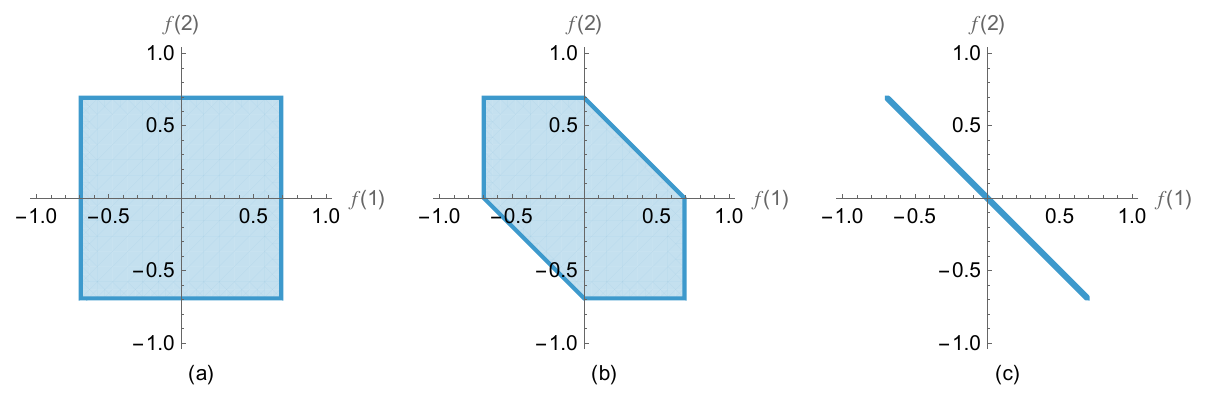}
    \caption{Entropohedra for various states on two qubits: (a) maximally mixed, $\rho=\frac14(\ket{0}\bra{0}+\ket{1}\bra{1})\otimes(\ket{0}\bra{0}+\ket{1}\bra{1})$. (b) Maximally classically correlated, $\rho=\frac12(\ket{00}\bra{00}+\ket{11}\bra{11})$. (c) Maximally entangled (Bell pair), $\rho=\frac12(\ket{00}+\ket{11})(\bra{00}+\bra{11})$.}
    \label{fig:twoqubit}
\end{figure}

\VerbatimFootnotes

\begin{figure}
    \centering
    \includegraphics[width=0.99\linewidth]
    {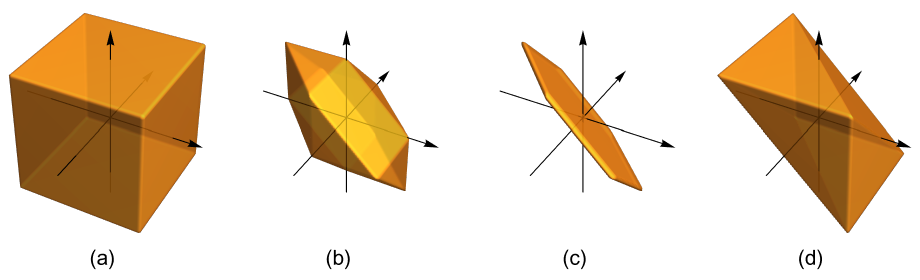}
    \caption{Entropohedra for various states on three parties: (a) maximally mixed. (b) Maximally classically correlated. (c) Maximally entangled (e.g.\ GHZ or W state). (d) Marginal of four-party perfect tensor. Axes are $f(1),f(2),f(3)$. Figures are scaled to have equal single-party entropies.\protect\footnotemark
    }
    \label{fig:threequbit}
\end{figure}

{\interfootnotelinepenalty=10000
    \footnotetext{
        To manipulate these 3D figures, paste the following commands into \emph{Mathematica}:\\
    \verb|RegionPlot3D[{Abs[f1] <= 1 && Abs[f2] <= 1 && Abs[f3] <= 1},| \\
    \verb| {f1, -1.1, 1.1}, {f2, -1.1, 1.1}, {f3, -1.1, 1.1}, PlotPoints -> 100]| \\
    \verb|RegionPlot3D[{Abs[f1] <= 1 && Abs[f2] <= 1 && Abs[f3] <= 1 && | \\
    \verb| Abs[f1 + f2] <= 1 && Abs[f2 + f3] <= 1 && Abs[f1 + f3] <= 1 && Abs[f1 + f2 + f3] <= 1}, |\\
    \verb| {f1, -1.1, 1.1}, {f2, -1.1, 1.1}, {f3, -1.1, 1.1}, PlotPoints -> 100]| \\
    \verb|RegionPlot3D[{Abs[f1] <= 1 && Abs[f2] <= 1 && Abs[f3] <= 1 && | \\
    \verb| Abs[f1 + f2] <= 1 && Abs[f2 + f3] <= 1 && Abs[f1 + f3] <= 1 && Abs[f1 + f2 + f3] <= 0.05},| \\
    \verb| {f1, -1.1, 1.1}, {f2, -1.1, 1.1}, {f3, -1.1, 1.1}, PlotPoints -> 100]| \\
    \verb|RegionPlot3D[{Abs[f1] <= 1 && Abs[f2] <= 1 && Abs[f3] <= 1 && Abs[f1 + f2 + f3] <= 1},| \\
    \verb| {f1, -1.1, 1.1}, {f2, -1.1, 1.1}, {f3, -1.1, 1.1}, PlotPoints -> 100]|
    }
}

Examples of entropohedra for various quantum states on two and three parties are shown in Figs.\ \ref{fig:twoqubit} and \ref{fig:threequbit} respectively.

\subsection{Basic properties}

Proofs of the propositions in this subsection are straightforward, and we leave them to the reader. The following properties of the entropohedron follow directly from the definition:
\begin{proposition}
Given an entropy function $S$ for $X$, the entropohedron $F_S$ is a non-empty compact convex polyhedron in $\R^{{X}}$, and is invariant under $f\to-f$.
\end{proposition}

Note that $F_S$ may be lower-dimensional than $\R^X$. Specifically, we have:
\begin{proposition}
Given an entropy function $S$ for $X$, the codimension of the entropohedron $F_S$ in $\R^X$ equals the maximum number of disjoint non-empty sets $A\subseteq X$ such that $S(A)=0$.
\end{proposition}

Under conical combinations of entropy functions, entropohedra combine in the sense of the Minkowski sum of sets.
\begin{proposition}
Given a set of pairs $\{(\alpha_i,S_i)\}$, where $\alpha_i$ is a non-negative number and $S_i$ is an entropy function for $X$, the function $S:2^X\to\R$ defined by
\be
S(A):=\sum_i\alpha_iS_i(A)
\ee
is an entropy function, and its entropohedron is
\be
F_S=\sum_i\alpha_iF_{S_i}:=\left\{\sum_i\alpha_if_i:f_1\in F_{S_1},f_2\in F_{S_2},\ldots\right\}.
\ee
\end{proposition}

Under direct sums of entropy functions (e.g.\ for tensor products of quantum states), the entropohedron is the Cartesian product.

\begin{proposition}
Given entropy functions $S_1,S_2$ on disjoint sets ${X}_1,{X}_2$ respectively, the function $S:2^{X_1\cup X_2}\to\R$ defined by
\be\label{setunion}
S(A_1\cup A_2):=S_1(A_1)+S_2(A_2)\qquad(A_1\subseteq {{X}_1},A_2\subseteq {{X}_2})
\ee
is an entropy function, and its entropohedron is
\be
F_S=F_{S_1}\times F_{S_2}\,.
\ee
\end{proposition}

\subsection{Saturation \& positivity}
\label{sec:saturation}

The boundary of $F_S$ is the locus where one or more of the defining inequalities \eqref{EDFdef} is saturated. One can therefore ask under what circumstances some inequalities can be saturated in a single EDF. One may also want to know whether an EDF can be made everywhere non-negative, or non-positive, on a given set of parties while saturating an inequality. The following theorem addresses these questions, showing that nesting is the key property for saturating multiple inequalities. Specifically, one can saturate the inequality on any nested set of subsets, while keeping the EDF non-negative on the smallest subset. In fact, one can do better, saturating the inequality with one sign on one nested set of subsets and with the other sign on another nested set of subsets disjoint from the first set.
\begin{theorem}\label{thm:saturatepos}
Let $S$ be an entropy function for $X$, and let $A_1\subset\cdots\subset A_m$ and $B_1\subset\cdots\subset B_n$ be subsets of $X$ such that $A_m\cap B_n=\emptyset$. Then there exists an EDF $f$ such that:
\begin{itemize}
\item $\forall\,i=1,\ldots,m$, $\sum_{x\in A_i}f(x)=S(A_i)$;
\item $\forall\,j=1,\ldots,n$, $\sum_{y\in B_j}f(y)=-S(B_j)$;
\item $\forall\,x\in A_1$, $f(x)\ge0$;
\item $\forall\,y\in B_1$, $f(y)\le0$.
\end{itemize}
\end{theorem}

The proof may be found in subsection \ref{sec:EDFproofs}. It is easy to see the reason for each of the conditions in theorem \ref{thm:saturatepos}. For example, for any sets $A,A'$, we have
\be
\sum_{x\in A}f(x)+\sum_{x\in A'}f(x)
=\sum_{x\in A\cup A'}f(x)+
\sum_{x\in A\cap A'}f(x)\le S(A\cup A')+S(A\cap A')\,.
\ee
Therefore, it is possible to positively saturate on both $A$ and $A'$ only if
\be
S(A)+S(A')=S(A\cup A')+S(A\cap A')\,.
\ee
If $A_1,A_2$ are nested, then this is always true, but otherwise it need not be. This is why the theorem requires the $A_i$s to be nested. A similar argument shows that the $B_j$s must also be nested, and disjoint from the $A_i$s.

We can also see why $f$ can be guaranteed to be non-negative only on $A_1$: if $S(A_2)<S(A_1)$, and we want to saturate on $A_1$, then $f(x)$ must be negative for at least one $x\in A_2\setminus A_1$. Similarly, we can guarantee that $f$ is non-positive only on $B_1$.

Specialized to the case $m=1$, $n=0$, theorem \ref{thm:saturatepos} gives us a function $f$ that, within $A_1$, is non-negative and sums to $S(A_1)$. This makes it an example of an \emph{entanglement contour} \cite{Vidal:2014aal}. The purpose of the entanglement contour is to characterize the distribution of entanglement between $A_1$ and its complement. However, our viewpoint here is that, more than any one EDF, it is the set of \emph{all} EDFs --- the entropohedron --- that best captures the distribution of entanglement in the given state.

A compact convex polytope is determined by its extremal points, being equal to their convex hull. The extremal points of $F_S$ are the EDFs that saturate $|X|$ independent inequalities, since $F_S$ is in an $|X|$-dimensional space. Furthermore, each one can be saturated with either sign. In view of theorem \ref{thm:saturatepos}, the extremal EDFs can thus be classified: Choosing a partition of ${X}$ into two subsets $A,B$, and an order for each one, $A=\{x_1,x_2,\ldots\}$, $B=\{y_1,y_2,\ldots\}$, we can saturate all of the subsets $\{x_1\},\{x_1,x_2\},\ldots$ positively and all of the subsets $\{y_1\},\{y_1,y_2\},\ldots$ negatively. This construction gives a total of $(|{X}|+1)!$ extremal points (although, in a non-generic state, they may not all be distinct).

The dual polyhedron $F_S^*$ is defined as the set of linear functionals $g$ on $\R^X$ such that, for any $f\in F_S$, $g(f)\le1$. $F_S^*$ is the convex hull of the linear functionals $g_A^\pm$, where
\be
g_A^\pm(f):=\pm\frac1{S(A)}\sum_{x\in A}f(x)\,,
\ee
for all non-empty $A\subseteq X$. (If $S(A)=0$ for some non-empty $A$, then the dual polyhedron extends to infinity in the corresponding direction in the dual space to $\R^X$.) Because every region $A$ can be saturated, i.e.\ there exist EDFs $f$ such that
\be\label{satdef}
\sum_{x\in A}f(x)=\pm S(A)\,,
\ee
every functional $g^\pm_A$ is on the boundary of $F_S^*$. Furthermore, in a generic state, where the equations \eqref{satdef} define faces of $F_S$, the functionals $g_A^\pm$ are vertices (extremal points) of $F_S^*$.

\subsection{Information quantities}

\begin{figure}
    \centering
    \includegraphics[width=0.99\linewidth]
    {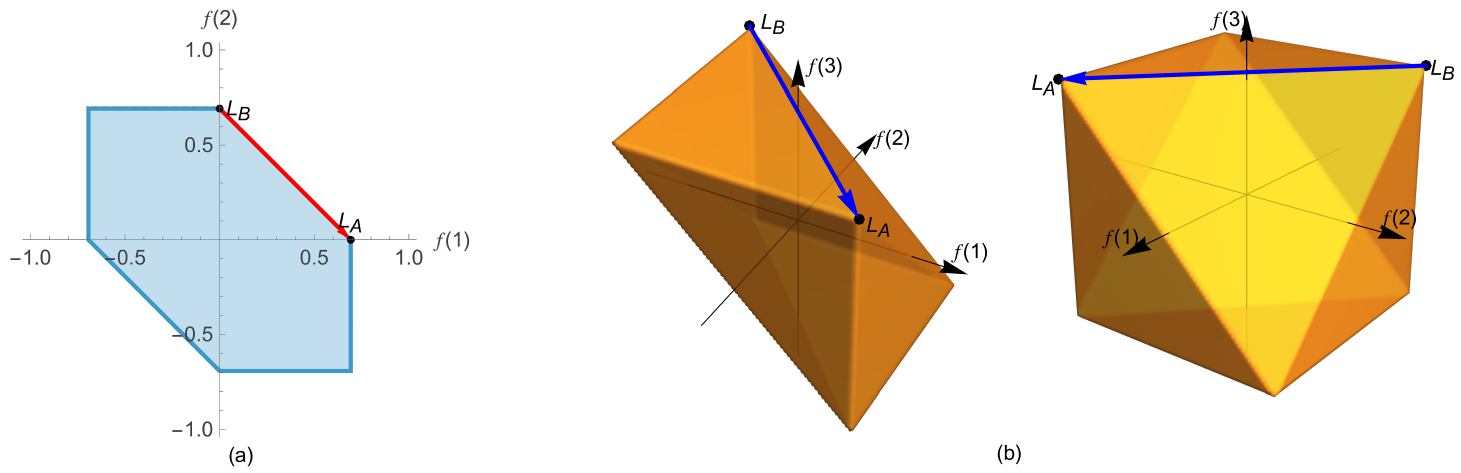}
    \caption{Geometric manifestation of information quantities on the entropohedron: (a) The entropohedron for a correlated pair of bits (as shown in Fig.\ \ref{fig:twoqubit}(b)); with $A=\{1\}$, $B=\{2\}$, we show the loci $L_A$, $L_B$ and their relative displacement, the vector $\vec I(A:B)$ (red), which quantifies the mutual information (see \eqref{MIvecdef}). (b) Two views of the entropohedron for the marginal of a four-party perfect tensor (as shown in Fig.\ \ref{fig:threequbit}(d)); with $A=\{1\}$, $B=\{2\}$, $C=\{3\}$, we show the loci $L_A$, $L_B$ and their relative displacement,, the vector $\vec I(A:B|C)$ (blue), which quantifies the conditional mutual information (see \eqref{CMIvecdef}).
    }
    \label{fig:IQs}
\end{figure}

The fact that nested sets can be saturated means that information quantities like the mutual information and conditional mutual information have simple geometric manifestations in the entropohedron. Let us start with the mutual information,
\be
I(A:B):=S(A)+S(B)-S(AB)\,,
\ee
where $A,B$ are disjoint sets and $AB:=A\cup B$. Let $L_A$ be the locus in $F_S$ where $A$ and $AB$ are both saturated, which is generically a codimension-2 edge on the boundary of $F_S$. On $L_A$,
\be
\sum_{x\in A}f(x)=S(A)\,,\qquad\sum_{x\in B}f(x)=S(AB)-S(A)\,.
\ee
Let $L_B$ be the locus where $B$ and $AB$ are saturated, also generically a codimension-2 edge of $F_S$. On $L_B$,
\be
\sum_{x\in A}f(x)=S(AB)-S(B)\,,\qquad\sum_{x\in B}f(x)=S(B)\,.
\ee
$L_A$ and $L_B$ both adjoin the codimension-1 face of $F_S$ where $AB$ is saturated; in fact, they are opposite edges of that face. They sit in parallel planes in $\R^X$, defined by constancy of the same linear functionals, namely $\sum_{x\in A}f(x)$ and $\sum_{x\in B}f(x)$. The relative displacement between those planes is
\be\label{MIvecdef}
\vec I(A:B)=
I(A:B)\left(\vec v(A)-\vec v(B)\right),
\ee
where $\vec v(A)$ is the vector in $\R^X$ with $x$-component 1 if $x\in A$ and 0 if $x\notin A$. See Fig.\ \ref{fig:IQs}(a) for an example. Note that, whereas the scalar quantity $I(A:B)$ is symmetric in its arguments, the vector $\vec I(A:B)$ is antisymmetric.

Similar reasoning allows us to read off the conditional mutual information:
\be
I(A:B|C):=S(AC)+S(BC)-S(C)-S(ABC)\,,
\ee
where $A,B,C$ are disjoint. We define $L_1$ as the locus in $F_S$ where $C,AC,ABC$ are all saturated, and $L_2$ as the locus where $B,BC,ABC$ are all saturated. On $L_1$ we have
\be
\sum_{x\in A}f(x)=S(AC)-S(C)\,,\qquad
\sum_{x\in B}f(x)=S(ABC)-S(AC)\,,\qquad
\sum_{x\in C}f(x)=S(C)\,,
\ee
while on $L_2$ we have
\be
\sum_{x\in A}f(x)=S(ABC)-S(BC)\,,\qquad
\sum_{x\in B}f(x)=S(BC)-S(C)\,,\qquad
\sum_{x\in C}f(x)=S(C)\,.
\ee
These are generically codimension-3 edges of $F_S$. They are opposite edges of the codimen\-sion-2 face of $F_S$ where $C,ABC$ are saturated. They sit in parallel planes in $\R^X$, which are displaced relative to each other by the vector
\be\label{CMIvecdef}
\vec I(A:B|C)=I(A:B|C)\left(\vec v(A)-\vec v(B)\right).
\ee
See Fig.\ \ref{fig:IQs}(b) for an example.

More complicated information quantities can be obtained by combining these. For example, for the (negative) tripartite information,
\be\label{I3def}
-I_3(A:B:C):=S(AB)+S(BC)+S(AC)-S(AB)-S(B)-S(C)-S(ABC)\,,
\ee
we can compute the difference between a mutual information vector and a conditional mutual information vector:
\be
\vec I(A:B|C)-\vec I(A:B)=-I_3(A:B:C)\left(\vec v(A)-\vec v(B)\right).
\ee

\subsection{Transformations of the entropy function}

We will now discuss several ways that an entropy function may be transformed into another one: adding and removing parties; merging and splitting parties; and minimizing over a set of parties. We will see that, in every case, the entropohedron transforms in a simple way.

\subsubsection{Adding \& removing parties, purification}

Next, we will see that EDFs and entropohedra are well-behaved under operations that replace the set ${X}$ with a larger or smaller set.

\begin{proposition}\label{thm:restriction1}
If $S$ is an entropy function for $X$ and ${X}'\subset{X}$, then the restriction of $S$ to subsets $A\subseteq X'$ is an entropy function $S'$ for $X'$. Furthermore, the restriction to ${X}'$ of any EDF $f$ for $S$ is an EDF for $S'$.
\end{proposition}
These statements follow directly from the definitions. More non-trivial is the other direction, that an EDF for $S'$ can always be completed to an EDF for $S$:

\begin{theorem}\label{thm:restriction}
Let $S$ be an entropy function for ${X}$, $S'$ its restriction to ${X}'\subset{X}$, and $f'$ an EDF for $S'$. Then there exists an EDF $f$ for $S$ that equals $f'$ on $X'$.
\end{theorem}

The proof of this theorem, which may be found in subsection \ref{sec:EDFproofs}, makes use of the SSA and WM properties \eqref{props3}, \eqref{props4}. Together, proposition \ref{thm:restriction1} and theorem \ref{thm:restriction} imply that, under restriction, the entropohedron transforms simply by projection:

\begin{corollary}\label{thm:restrictprojection}
Let $S$ be an entropy function for ${X}$ and $S'$ its restriction to ${X}'\subset{X}$. Then $F_{S'}$ equals the image of $F_S$ under the projection from $\R^{X}$ to $\R^{{X}'}$.
\end{corollary}

\begin{figure}
    \centering
    \includegraphics[width=0.4\linewidth]
    {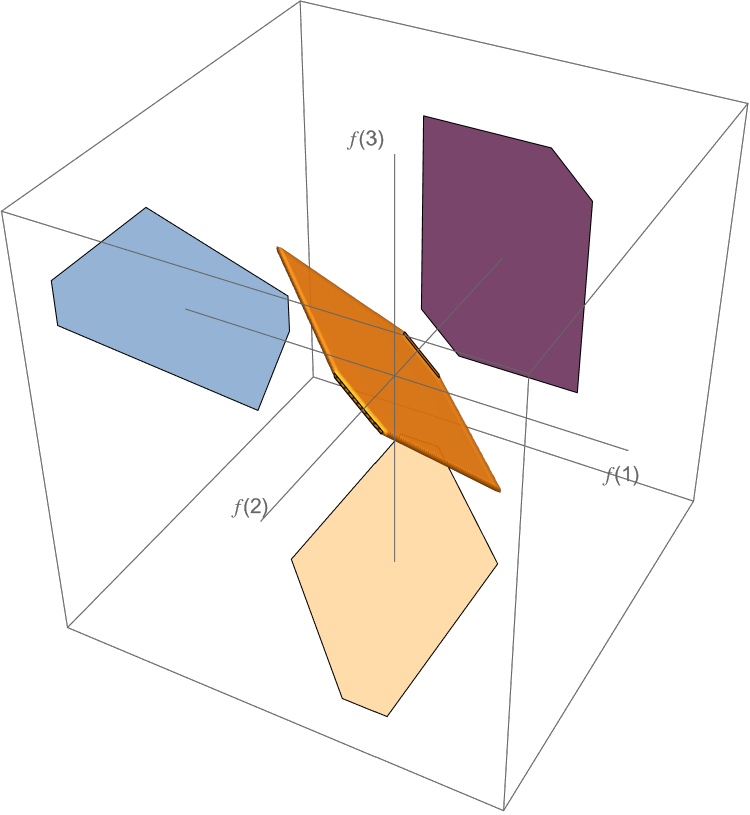}
    \caption{Entropohedron $F_S$ for a pure state on three parties labelled $1,2,3$ (orange), along with the entropohedra for the same state reduced to two parties $1,2$ (tan), $2,3$ (blue), and $1,3$ (purple), obtained by projecting $F_S$ along the axes. Due to the purity of the state on $123$, these projections are bijective.}
    \label{fig:purification}
\end{figure}

An application of this corollary is to purification. The entropy function $S$ on ${X}$ is said to \emph{purify} the entropy function $S'$ on ${X}'\subset{X}$ if $S({X})=0$ and, for all $A\subseteq X'$, $S(A)=S'(A)$. Given ${X}'$ and $S'$, a minimal purification can be constructed by adding a single element $O$ to ${X}'$, so ${X}={X}'\cup\{O\}$, and defining $S$ as follows:
\be
S(A) := S'(A)\quad(A\not\ni O)\,,\qquad
S(A):=S'({X}\setminus A)\quad (A\ni O)\,.
\ee
Since $S({X})=0$, for any EDF on $S$ we must have $\sum_{x\in{X}}f(x)=0$. Given an EDF $f'$ for $S'$, there therefore exists a unique extension $f$ for $S$, setting $f(O)=-\sum_{x\in{X}'}f(x)$. Geometrically, the entropohedron $F_{S'}\subset \R^{{X}'}$ gets lifted onto the ``purity'' hyperplane
\be
\left\{f\in\R^X:\sum_{x\in X}f(x)=0\right\}
\ee
to give $F_S$. Removing from ${X}$ any element of ${X}'$ to obtain a set ${X}''$ then projects $F_S$ down to $F_{S''}\subset \R^{{X}''}$. All of these maps are bijective (since they are linear and dimension-preserving), so no information is lost in the process. This is illustrated in Fig. \ref{fig:purification}.

\subsubsection{Merging \& splitting parties}

We can also ask how EDFs and the entropohedron transform under merging and splitting parties. Given a subset ${X}'$ of ${X}$, we can merge the elements of ${X}\setminus{X}'$ into a single party, which we call $R$, making the set $Y:={X}'\cup \{R\}$. We can then turn an entropy function $S$ for $X$ into an entropy function $T$ for $Y$, and an EDF for $S$ into an EDF for $T$, in the obvious ways:
\begin{proposition}\label{thm:merging}
If $X'\subset X$ then, given an entropy function $S$ for $X$, we obtain an entropy function $T$ for $Y:={X}'\cup \{R\}$ as follows:
\be\label{mergedEV}
\forall\,A\subseteq X'\,,\qquad T(A):=S(A)\,,\qquad T(A\cup\{R\}):=S(A\cup({X}\setminus{X}'))\,.
\ee
Furthermore, given an EDF $f$ for $S$, then we obtain an EDF $g$ for $T$ as follows:
\be\label{compatible}
g(x)=f(x) \quad(x\in{X}')\,,\qquad
g(R)=\sum_{x\in{X}\setminus{X}'}f(x)\,.
\ee
\end{proposition}
\eqref{compatible} defines a linear map $M:\R^X\to\R^Y$, that acts as the identity on the subspace $\R^{X'}$ and by projection along the vector $(1,\ldots,1)$ on the subspace $\R^{X\setminus X'}$.

More non-trivial is to go in the other direction, that is, given an EDF $g$ for $T$ on the coarsened set $Y$, to define an EDF $f$ for $S$ on the refined set ${X}$ that is compatible with $T$, in the sense of \eqref{compatible}. The following theorem says that this is possible.
\begin{theorem}\label{thm:refine}
Let $S$ be an entropy function for ${X}$ and $T$ the entropy function for $Y:={X}'\cup \{R\}$ defined by \eqref{mergedEV}. Given an EDF $g$ for $T$, there exists an EDF $f$ for $S$ obeying \eqref{compatible}.
\end{theorem}

Like that of theorem \ref{thm:restriction}, the proof of this theorem, again found in subsection \ref{sec:EDFproofs}, makes use of the SSA and WM properties \eqref{props3}, \eqref{props4}. Together, proposition \ref{thm:merging} and theorem \eqref{thm:refine} imply that, under merging of parties, the entropohedron transforms from $\R^X$ to $\R^Y$ under the map $M$:
\begin{corollary}
Let $S$ be an entropy function for $X$ and $T$ the entropy function for $Y:=X'\cup\{R\}$ defined by \eqref{mergedEV}. Then $F_T=M(F_S)$, where $M:\R^X\to\R^Y$ is defined by \eqref{compatible}.
\end{corollary}

\subsubsection{Partial minimization}

Yet another way to obtain an entropy function on a subset of $X$ is to minimize over the complement:
\begin{proposition}
Given an entropy function $S$ for $X$, we obtain an entropy function $S'$ for $X'\subset X$ as follows:
\be\label{S'def}
S'(A\subseteq X'):=\min_{\hat A\subseteq X\setminus X'}S(A\cup \hat A)\,.
\ee
\end{proposition}
\begin{proof}
It's clear that $S'(\emptyset)=0$. To show that $S'$ obeys \eqref{props3}, let $\hat A_{\rm min}$, $\hat B_{\rm min}$ be minimizers for $S(A\cup\hat A)$, $S(B\cup \hat B)$ respectively. We have
\begin{align}
S'(A\cap B)&\le S((A\cap B)\cup(\hat A_{\rm min}\cap\hat B_{\rm min}))=S((A\cup \hat A_{\rm min})\cap(B\cup\hat B_{\rm min}))\\
S'(A\cup B)&\le S((A\cup B)\cup(\hat A_{\rm min}\cup\hat B_{\rm min}))=S((A\cup \hat A_{\rm min})\cup(B\cup\hat B_{\rm min}))\,,
\end{align}
hence
\begin{align}
S'(A\cap B)+S'(A\cup B)&\le S((A\cup \hat A_{\rm min})\cap(B\cup \hat B_{\rm min}))+S((A\cup \hat A_{\rm min})\cup(B\cup\hat B_{\rm min})) \nonumber\\
&\le S(A\cup \hat A_{\rm min})+S(B\cup\hat B_{\rm min})\nonumber \\
&= S'(A)+S'(B)\,,
\end{align}
where in the second inequality we applied \eqref{props3}. Similar reasoning shows that $S'$ obeys \eqref{props4}.
\end{proof}

For the graph min cut function (see subsection \ref{sec:graphflows}), $S'$ is the min cut function obtained by converting the boundary vertices in $X\setminus X'$ into internal (i.e.\ non-boundary) vertices. This shows that partial minimization maps the holographic entropy cone (defined in that subsection) to itself. In the quantum setting, we do not know what kind of operation, if any, acting on the state on the parties $X$ with entropies $S$ gives the state on the parties $X'$ with entropies $S'$.\footnote{An interesting question in this context is whether partial minimization preserves the constrained inequalities valid for quantum states \cite{Linden:2004ebt,Cadney:2011vix}.}

We will now show that the entropohedra for $S'$ are related in a very simple way to the ones for $S$. Given a function $f':X'\to\R$, we define the function $f_{f'}:X\to\R$ by
\be
f_{f'}(x):=\begin{cases}f'(x)\,,\quad &x\in X' \\ 0\,,\quad &x\notin X'\end{cases}\,.
\ee
\begin{proposition}
Let $S$ be an entropy function for $X$ and $S'$ the entropy function for $X'\subset X$ defined by \eqref{S'def}. Then the function $f':X'\to \R$ is an EDF for $S'$ if and only if $f_{f'}$ is an EDF for $S$.
\end{proposition}
\begin{proof}
Suppose $f'$ is an EDF for $S'$. Then, for any $A\subseteq X$,
\be
\left|\sum_{x\in A}f_{f'}(x)\right| = \left|\sum_{x\in A\cap X}f'(x)\right|\le S'(A\cap X)\le S(A)\,,
\ee
where in the first inequality we used the fact that $f'$ is an EDF for $S'$ and in the last one we used \eqref{S'def}.

Conversely, if $f_{f'}$ is an EDF for $S$, then for any $A\subseteq X'$ and $\hat A\subseteq X\setminus X'$,
\be\label{sumf'}
\left|\sum_{x\in A}f'(x)\right|=\left|\sum_{x\in A\cup \hat A}f_{f'}(x)\right|\le S(A\cup \hat A)\,.
\ee
Choosing $\hat A$ to be the minimizer in the definition \eqref{S'def} of $S'$, the right-hand side of \eqref{sumf'} equals $S'(A)$, so $f'$ is an EDF for $S'$.
\end{proof}
\begin{corollary}
Let $S$ be an entropy function for $X$ and $S'$ the entropy function for $X'\subset X$ defined by \eqref{S'def}. Then $F_{S'}$ is the intersection of $F_S$ with the subspace $\R^{X'}$ in $\R^X$.
\end{corollary}

\subsection{Relation to network flows}
\label{sec:graphflows}

Entropy functions can be derived not just from quantum states but also from networks. Specifically, given a weighted undirected graph, with a subset ${X}$ of the vertex set labelled as ``external'' vertices, the min cut function on $X$ satisfies \eqref{props1}--\eqref{props4}. (This function defines a ``pure state'' in the sense that $S({X})=0$.) For a given set $A\subseteq X $, by the max flow-min cut theorem, $S(A)$ equals the maximal flux through the graph from $A$ to ${X}\setminus A$. As we will show in theorem \ref{thm:graphEDF} below, there is a direct relationship between the set of flows and the entropohedron: A point $f\in\R^{X}$ is in $F_S$ if and only if it equals the boundary flux of some flow. This is the graph version of theorem \ref{thm:classicalflowEDF} from section \ref{sec:strictflows}.

Not every entropy function can be realized as the min cut function on a network. Network min cut functions obey additional inequalities beyond \eqref{props3}, \eqref{props4}, such as $-I_3(A:B:C)\ge0$ (where the negative tripartite information $-I_3$ was defined in \eqref{I3def}). The set of entropy functions obtainable from networks, for a given number of parties, is called the \emph{holographic entropy cone} \cite{Bao:2015bfa}. Its structure for a general number of parties remains unknown (see \cite{Grimaldi:2025jad} for an overview and references). It would be interesting if the entropohedron might have some special properties for such entropy functions, and if it could be used to shed some light on the holographic entropy cone.

Before formally stating theorem \ref{thm:graphEDF}, we remind the reader of the relevant definitions and facts about graphs. We'll use the term \emph{network} to refer to a weighted undirected graph with some vertices labelled as ``external'' and the rest as ``internal''. Specifically, a network $(X,I,E,w)$ consists of 
a set ${X}$ of external vertices, a set $I$ of internal vertices, a set $E$ of edges, and a positive function $w$ on $E$. Each edge $e\in E$ is an unordered pair of distinct vertices. A cut $r$ is a subset of $I\cup{X}$, its edge set $\eth r$ is the set of edges with one vertex inside and the other outside $C$, and its weight $|\eth r|$ is the total weight of its edge set:
\be
|\eth r|:=\sum_{e\in \eth r}w(e)\,.
\ee
A cut for $A$ is one such that $r\cap{X}=A$; we write $r\sim A$. The entropy function $S$ on $X$ is defined by the minimal-weight cut:
\be
S(A):=\min_{r\sim A}|\eth r|\,.
\ee
Although the graph is undirected, in order to define flows we need to assign an arbitrary fiducial orientation to the edges; we then define a sign $s(i,e)=\pm1$ for $i\in e$, with $s(i,e)=+1$ if the orientation of $e$ points toward the vertex $i$ and $s(i,e)=-1$ if it points away. A \emph{flow} on the graph is a function $v:E\to\R$ obeying
\be
\forall \,e\in E\,,\quad |v(e)|\le w(e)\,;\qquad
\forall\,i\in I\,,\quad \sum_{e\ni i}s(i,e)v(e)=0\,.
\ee
At an external vertex $x\in{X}$, the flux of a flow is given by
\be
\Phi_v(x):=\sum_{e\ni x}s(x,e)v(e)\,.
\ee
The set $\mathcal{V}$ of flows is convex and invariant under $v\to-v$. 
The max flow-min cut theorem states that the maximal flux through $A$ equals $S(A)$:
\be
S(A)=\sup_{v\in\mathcal{V}}\sum_{x\in A}\Phi_v(x)\,.
\ee

We are now ready to state and prove the relation between flows and the min cut entropohedron.

\begin{theorem}\label{thm:graphEDF}
Given a network $(X,I,E,w)$, 
let $S$ be the corresponding min cut function on $2^{X}$. A point $f\in\R^{X}$ is in $F_S$ if and only if $f=\Phi_v$ for some flow $v$ on the network.
\end{theorem}

\begin{proof}
It follows directly from the definitions that, for any flow $v$, $\Phi_v\in F_S$. The more non-trivial statement is the converse: for any $f\in F_S$, there exists a flow $v$ such that $\Phi_v=f$. To prove this, we set up the following convex program:
\begin{align}
\text{Minimize $0$ over $v:E\to\R$ subject to:}&\nonumber\\
\forall \,e\in E\,,&\quad |v(e)|\le w(e) \nonumber\\
\forall\,x\in{X}\,,&\quad\sum_{e\ni x}s(x,e)v(e)=f(x)\nonumber\\
\forall\,i\in I\,,&\quad \sum_{e\ni i}s(i,e)v(e)=0\,.
\end{align}
Such a convex program with objective 0 is a feasibility test: If there exists a feasible flow, then the optimal value is 0; if not, it is $+\infty$. Dualizing the program, we will see which one is correct. Since the program can be written as a linear program, strong duality automatically holds.

We will impose the first constraint implicitly and introduce Lagrange multipliers $\lambda(x)$, $\lambda(i)$ for the second and third constraints respectively. The dual program is
\be
\text{Maximize }g[\lambda]:=
\sum_{x\in{X}}\lambda(x)f(x)
-\sum_{e\in E}\left|\Delta\lambda(e)\right|
\text{ over $\lambda:I\cup X\to\R$}\,,
\ee
where $\Delta\lambda(e)$ is the difference between the values of $\lambda$ at the two vertices of the edge $e$. (The dual program has no constraints.) Defining, for $\hat\lambda\in\R$, the cut $r(\hat\lambda)$ and boundary set $A(\hat\lambda)$ by,
\be
r(\hat\lambda):=\{i\in I\cup{X}:0<\lambda(i)<\hat\lambda\text{ or }\hat\lambda<\lambda(i)<0\}\,,\qquad
A(\hat\lambda):=r(\hat\lambda)\cap X\,,
\ee
we can rewrite the dual objective as
\be
\int d\hat\lambda\left[
\sgn(\hat\lambda)\sum_{x\in A(\hat\lambda)}f(x)-|\eth r(\hat\lambda)|
\right].
\ee
By the assumption that $f\in F_S$, the first term in the integrand is bounded above by $S(A(\hat\lambda))$. The second term is bounded above by $-S(A(\hat\lambda))$ by the definition of $S$. So the total is non-positive for all $\hat\lambda$. Therefore, the maximum is 0, achieved by $\lambda=0$. (As a check, note that, if $f$ violates the EDF condition for some $A\subseteq X $, then by making $\lambda$ constant on the minimal cut for $A$ and 0 elsewhere, the integral can be made arbitrarily large and positive. Therefore, the maximum is $+\infty$, so as expected the primal program does not admit a feasible point.)
\end{proof}

\subsection{Proofs}
\label{sec:EDFproofs}

In this subsection, we will prove theorems \ref{thm:saturatepos}, \ref{thm:restriction}, and \ref{thm:refine}.

\subsubsection{Disentangling lemmas}

Given a function $\mu:2^{X}\to\R$, we define the function $\phi_\mu:X\to\R$ as follows:
\be
\phi_\mu(x):=\sum_{A\ni x}\mu(A)\,.
\ee
We will need the following ``disentangling'' lemmas in the proofs in the following subsubsection.

\begin{lemma}\label{thm:disentangling1}
Given functions $\mu,\nu:2^X\to[0,+\infty)$, there exist functions $\mu',\nu':2^X\to[0,+\infty)$ with non-overlapping supports such that
\be
\phi_{\mu'}-\phi_{\nu'} = \phi_\mu-\phi_\nu
\ee
and
\be
\sum_{A\subseteq {X}}\left(\mu'(A)+\nu'(A)\right)S(A)\le
\sum_{A\subseteq {X}}\left(\mu(A)+\nu(A)\right)S(A)\,.
\ee
By ``non-overlapping supports'', we mean that, if $\mu'(A)>0$, $\nu'(B)>0$, then $A\cap B=\emptyset$.
\end{lemma}

\begin{lemma}\label{thm:disentangling2}
Given a function $\mu:2^X\to[0,+\infty)$, there exist subsets of ${X}$, $A_1\subset \cdots \subset A_m$, and a function $\mu':2^X\to[0,+\infty)$ supported on $\{A_1,\ldots,A_m\}$ such that
\be
\phi_{\mu'} = \phi_\mu
\ee
and
\be
\sum_{i=1}^m\mu'(A_i)S(A_i)\le
\sum_{A\subseteq {X}}\mu(A)S(A)\,.
\ee
\end{lemma}

\begin{corollary}\label{thm:disentangling}
Given functions $\mu,\nu:2^X\to[0,+\infty)$, there exist subsets of ${X}$, $A_1\subset \cdots \subset A_m$, $B_1\subset \cdots \subset B_n$ such that $A_m\cap B_n=\emptyset$, and functions $\mu',\nu':2^X\to[0,+\infty)$ supported on $\{A_1,\ldots,A_m\}$, $\{B_1,\ldots,B_n\}$ respectively such that
\be
\phi_{\mu'}-\phi_{\nu'} = \phi_\mu-\phi_\nu
\ee
and
\be
\sum_{A\subseteq {X}}\left(\mu'(A)+\nu'(A)\right)S(A)\le
\sum_{A\subseteq {X}}\left(\mu(A)+\nu(A)\right)S(A)\,.
\ee
\end{corollary}

\begin{proof}[Proof of lemma \ref{thm:disentangling1}] 
Suppose there exists a pair of sets $A,B\subseteq X $ with $\mu(A)>0$, $\nu(B)>0$, and $A\cap B\neq\emptyset$. Setting $\alpha:=\min(\mu(A),\nu(B))$, do the following replacement:
\be\label{overlapstep}
\mu(A\setminus B)\to \mu(A\setminus B)+\alpha\,,\quad
\nu(B\setminus A)\to \nu(B\setminus A)+\alpha\,,\quad
\mu(A)\to \mu(A)-\alpha\,,\quad
\nu(B)\to \nu(B)-\alpha\,,
\ee
leaving $\mu$ and $\nu$ unchanged for all other subsets of ${X}$. By the WM property \eqref{props4}, this does not increase $\sum_A(\mu(A)+\nu(A))S(A)$. It has the following effect on $\phi_\mu$, $\phi_\nu$:
\be
\phi_\mu\to\phi_\mu-\alpha\,,\quad
\phi_\nu\to\phi_\nu-\alpha\quad
\text{ on $A\cap B$}\,;\qquad
\phi_\mu\to\phi_\mu\,,\quad
\phi_\nu\to\phi_\nu\quad\text{ elsewhere}\,.
\ee
Therefore it leaves $\phi_\mu-\phi_\nu$ unchanged.

Now we check that repeating the above procedure converges to a pair of functions $\mu,\nu$ with non-overlapping supports. We define the following quantity, which quantifies the total overlap between $\mu$ and $\nu$:
\be\label{Odef}
O:=\sum_{x\in X}\phi_\mu(x)\phi_\nu(x) = \sum_{A,B\subseteq X}\mu(A)\nu(B)|A\cap B|\,,
\ee
Under the step described in the above paragraph, the change in $O$ is
\be
\Delta O = -\sum_{x\in A\cap B}\alpha\left(\phi_\mu(x)+\phi_\nu(x)-\alpha\right)\le-\mu(A)\nu(B)|A\cap B|\,,
\ee
where the inequality follows from the fact that, for any $x\in A\cap B$, $\phi_\mu(x)\ge \mu(A)$ and $\phi_\nu(x)\ge \nu(B)$. Therefore, the step reduces the quantity $O$ by at least the contribution of that pair $A,B$ to it in the second sum of \eqref{Odef}.

Let $n$ be the number of pairs $A,B$ that overlap, and therefore could potentially contribute to the second sum in \eqref{Odef}. The pair that contributes the most to the sum contributes at least $O/n$. Running the above step for that pair therefore reduces $O$ by at least a factor $1-1/n$. Therefore, in the limit of an infinite number of steps, $O$ must go to zero.\footnote{We suspect that the algorithm actually always terminates after a finite number of steps, but we have not proven that.} Furthermore, the change in the functions $\mu,\nu$ also goes to zero. So they converge to functions whose supports have no overlaps.
\end{proof}

\begin{proof}[Proof of lemma \ref{thm:disentangling2}]
The proof is similar to the previous one.  Suppose there exists a pair of distinct sets $A,B\subseteq {X}$ with $\mu(A)>0$, $\mu(B)>0$, and neither $A\subset B$ nor $B\subset A$. With $\alpha:=\min(\mu(A),\mu(B))$, do the following replacement:
\begin{multline}
\mu(A\cup B)\to\mu(A\cup B)+\alpha\,,\quad
\mu(A\cap B)\to\mu(A\cap B)+\alpha\,,\\
\mu(A)\to\mu(A)-\alpha\,,\quad
\mu(B)\to\mu(B)-\alpha\,,
\end{multline}
leaving $\mu$ unchanged for all other subsets of ${X}$. This 
leaves $\phi_\mu$ unchanged, and (by the SSA property \eqref{props3}) does not increase $\sum_A\mu(A)S(A)$.

Now we check that repeating the above procedure converges to a function $\mu$ with nested support. We define the following quantity, which quantifies the degree to which the support of $\mu$ is not nested:
\be\label{Vdef}
V:=\sum_{(A,B)}\mu(A)\mu(B)|A\setminus B|\,|B\setminus A|
\ee
(where the sum is over unordered pairs). Under the above step, the term in the sum in \eqref{Vdef} corresponding to the pair $A,B$ that we are acting on gets deleted. However, there are also effects on terms involving a third region. The change in the other terms is
\begin{align}
\alpha\sum_{C\neq A,B}\mu(C)&\left(
|(A\cup B)\setminus C|\,|C\setminus (A\cup B)| 
+|(A\cap B)\setminus C|\,|C\setminus (A\cap B)| \right.\nonumber\\
&\qquad\qquad\qquad\qquad\qquad\qquad\qquad\qquad\qquad
\left.-|A\setminus C|\,|C\setminus A|
-|B\setminus C|\,|C\setminus B|
\right) \nonumber\\
&=-\alpha\sum_{C\neq A,B}\mu(C)\left(|A\setminus(B\cup C)|\,|(B\cap C)\setminus A|+|B\setminus(A\cup C)|\,|(A\cap C)\setminus B|\right) \nonumber\\
&\le0\,,
\end{align}
where the equality follows from a short inclusion-exclusion calculation. The important result is that the contribution of the other terms to the change in $V$ is non-positive, so the change in $V$ in one step is
\be
\Delta V\le -\mu(A)\mu(B)|A\setminus B|\,|B\setminus A|\,.
\ee

Let $n$ be the total number of pairs $A,B$ that are non-nested, and therefore could potentially contribute to the sum in \eqref{Vdef}. The pair that contributes the most to the sum contributes at least $V/n$. Running the above step for that pair therefore reduces $V$ by at least a factor of $1-1/n$. Therefore, in the limit of an infinite number of steps, $V$ must go to zero. Furthermore, the change in the function $\mu$ also goes to 0, so it converges to a function whose support is nested.
\end{proof}

\subsubsection{Proofs about EDFs}

With the disentangling lemmas in hand, we are now in a position to prove the theorems quoted in the rest of the section. All of the proofs rely on the strong duality of convex programs. For all of the convex programs considered in this subsection, the constraints are affine (or can be written in affine form), so Slater's condition is automatically satisfied and strong duality is guaranteed.

\begin{proof}[Proof of theorem \ref{thm:saturatepos}] Consider the following convex program:
\begin{multline}\label{satthmprog}
\text{Maximize }\sum_{i=1}^m\sum_{x\in A_i}f(x)-\sum_{j=1}^n\sum_{x\in B_j}f(x)\text{ over $f:X\to\R$ subject to:} \\
\forall\,A\subseteq X \,,\quad\left|\sum_{x\in A}f(x)\right|\le S(A)\,;\qquad \forall \,x\in A_1\,,\quad f(x)\ge0\,;\qquad
\forall\,x\in B_1\,,\quad f(x)\le0\,.
\end{multline}
Clearly, the objective of this program cannot exceed
\be\label{objmax}
\sum_{i=1}^mS(A_i)+\sum_{j=1}^nS(B_j)\,,
\ee
and can achieve this value only if, for all $i=1,\ldots,m$, $\sum_{x\in A_i}f(x)=S(A_i)$ and, for all $j=1,\ldots,n$, $\sum_{x\in B_j}f(x)=S(B_j)$.

We now dualize \eqref{satthmprog}, using Lagrange multipliers $\mu$ and $\nu$, functions on $2^{X}$, for the constraints $\pm\sum_{x\in A}f(x)\le S(A)$ respectively, and imposing the last two constraints implicitly. Defining the function $\psi:X\to\mathbf{Z}$ by
\be
\psi(x):=|\{i:A_i\ni x\}|-|\{j:B_j\ni x\}|\,,
\ee
the dual program is the following:
\begin{align}\label{satdual}
\text{Minimize }\sum_{A\subseteq X }\left(\mu(A)+\nu(A)\right)S(A)&\text{ over functions $\mu,\nu:2^{X}\to[0,\infty)$ subject to:} \\
\phi_\mu(x)-\phi_\nu(x)&=\psi(x)\quad(x\not\in A_1\cup B_1)\,;\label{dualconstraint1}\\
\phi_\mu(x)-\phi_\nu(x)&\ge\psi(x)\quad(x\in A_1)\,;
\\
\phi_\mu(x)-\phi_\nu(x)&\le\psi(x)\quad(x\in B_1)\,.\label{dualconstraint3}
\end{align}

Appealing to corollary \ref{thm:disentangling}, we can assume that the minimizing pair $\mu,\nu$ is supported on nested and disjoint sets. The dual constraints \eqref{dualconstraint1}--\eqref{dualconstraint3}, together with the non-negativity of $\mu,\nu$, then imply that $\mu$ is supported precisely on $A_1,\ldots,A_m$, with $\mu(A_i)=1$, together possibly with some subsets of $A_1$; and that $\nu$ is supported precisely on $B=1,\ldots,n$, with $\nu(B_j)=1$, together possibly with some subsets of $B_1$. These extra subsets of $A_1$ and $B_1$ only increase the objective, so to minimize it we should eliminate them.\footnote{
As noted below the statement of theorem \ref{thm:saturatepos}, we cannot in general require $f$ to be non-negative on any of the $A_i$s except $A_1$: suppose $A_2\supset A_1$. If $S(A_2)<S(A_1)$ and we saturate on $A_1$, then the integral of $f$ on $A_2\setminus A_1$ must be negative. The place where the proof goes wrong if we attempt to make $f$ non-negative on $A_2$ is that, instead of the first line of \eqref{dualconstraint1}, we would have $\phi_\mu\ge\psi$ on $A_2$, which could be satisfied by setting $\mu(A_2)=2$ and $\mu(A_1)=0$. This would give a lower value of the objective than $\sum_iS(A_i)+\sum_jS(B_j)$.} So the value \eqref{objmax} is the solution for the dual program \eqref{satdual}, and therefore for the primal program \eqref{satthmprog} as well.
\end{proof}

\begin{proof}[Proof of theorem \ref{thm:restriction}]
The proof has a similar structure to that of theorem \ref{thm:saturatepos}. We consider the following convex program:
\be
\text{Minimize }\sum_{x\in{X}'}\left|f(x)-f'(x)\right|
\text{ over $f:X\to\R$ subject to: }
\forall\,A\subseteq X \,,\quad\left|\sum_{x\in A}f(x)\right|\le S(A)\,.
\ee
Our goal is to show that the solution to this convex program is 0. The dual program is:
\begin{multline}
\text{Maximize }-\sum_{A\subseteq X }\left(\mu(A)+\nu(A)\right)S(A)+\sum_{x\in{X}'}(\phi_\mu(x)-\phi_\nu(x))f'(x) \\
\shoveleft{\text{ over $\mu,\nu:2^{X}\to[0,\infty)$ subject to: }}\\ 
\forall\, x\in{X}'\,,\quad|\phi_\mu(x)-\phi_\nu(x)|\le1\,;\qquad
\forall \,x\in{X}\setminus{X}'\,,\quad\phi_\mu(x)=\phi_\nu(x)\,.
\end{multline}
By corollary \ref{thm:disentangling}, we can assume that any minimizing pair $\mu$, $\nu$ has disjoint supports. (We will not use the nesting of their supports in this proof, which therefore uses only the WM and not the SSA property of $S$.) The only way to satisfy the second constraint is if $\phi_\mu(x)=\phi_\nu(x)=0$ for all $x\in{X}\setminus{X}'$, in other words $\mu$, $\nu$ are non-zero only on subsets of ${{X}'}$. The objective then becomes
\be\label{dualobj}
-\sum_{A\subseteq {{X}'}}\left[\mu(A)\left(S(A)-\sum_{x\in A'}f'(x)\right)+\nu(A)\left(S(A)+\sum_{x\in A'}f'(x)\right)\right].
\ee
Since by assumption $f'$ is an EDF on ${X}'$, the coefficients in round parentheses are non-negative, so the objective is non-positive, and its maximum is 0, achieved by $\mu=\nu=0$.
\end{proof}

\begin{proof}[Proof of theorem \ref{thm:refine}]
This theorem is similar to that of theorem \ref{thm:restriction}, the only difference being that now we are constraining $\sum_{x\in{X}\setminus{X}'}f(x)=g(R)$, where $g(R)$ is given. We use the same strategy as in that proof, but need to add a new Lagrange multiplier $\lambda$ to enforce the new constraint, and therefore obtain a slightly more complicated dual program.

The primal program is:
\begin{multline}
\text{Minimize }\sum_{x\in{X}'}\left|f(x)-f'(x)\right|
\text{ over $f:X\to\R$ subject to:}\\
\forall\,A\subseteq X \,,\quad\left|\sum_{x\in A}f(x)\right|\le S(A)\,;\qquad
\sum_{x\in{X}\setminus{X}'}f(x)=g_B\,,
\end{multline}
and again our goal is to show that the optimal value is 0. The dual program is
\begin{multline}
\text{Maximize }-\sum_{A\subseteq X }\left(\mu(A)+\nu(A)\right)S(A)+\sum_{x\in{X}'}\left(\phi_\mu(x)-\phi_\nu(x)\right)f'(x) +\lambda g(R) \\
\shoveleft{\text{ over $\lambda\in\R$ and $\mu,\nu:2^{X}\to[0,\infty)$ subject to: }}\\ 
\forall\, x\in{X}'\,,\quad|\phi_\mu(x)-\phi_\nu(x)|\le1\,;\qquad
\forall \,x\in{X}\setminus{X}'\,,\quad\phi_\mu(x)-\phi_\nu(x)=\lambda\,.
\end{multline}
Now, appealing to corollary \ref{thm:disentangling}, we can assume that the optimal pair of functions $\mu,\nu$ is supported on nested and mutually disjoint subsets of ${X}$. Based on the second dual constraint, there are then three cases, depending on the sign of $\lambda$:
\begin{itemize}
\item $\lambda=0$: The supports of $\mu,\nu$ are contained in ${X}'$. It reduces to the situation in the proof of theorem \ref{thm:restriction}, and the maximal value is 0.
\item $\lambda>0$: The support of $\nu$ is contained in ${X}'$. The support of $\mu$ includes one region $A_m\supseteq({X}\setminus{X}')$ with $\mu(A_m)=\lambda$, with all other supoorted regions contained in ${X}'$. The dual objective now includes, in addition to \eqref{dualobj}, the following term:
\be
-\lambda\left(S(A_m)-\sum_{x\in A_m\cap{X}'}f'(x)-g(R)\right)=
-\lambda\left(T(\tilde A_m)-\sum_{x\in\tilde A_m}f'(x)\right),
\ee
where $\tilde A_m:=(A_m\cap{X}')\cup\{B\}$, and where we used \eqref{mergedEV}. Since $g$ is assumed to be an EDF for $T$, the coefficient in round parentheses is non-negative, so the term is maximized by taking $\lambda\to0$, giving 0. The terms in \eqref{dualobj} also have maximal value 0.
\item $\lambda<0$: The support of $\mu$ is contained in ${X}'$. The support of $\nu$ includes one region $B_N\supseteq({X}\setminus{X}')$ with $\nu(B_n)=-\lambda$, with all other supported regions contained in ${X}'$. The dual objective now includes, in addition to \eqref{dualobj}, the following term:
\be
\lambda\left(S(B_n)+\sum_{x\in B_n\cap{X}'}f'(x)+g(R)\right)=
\lambda\left(T(\tilde B_n)+\sum_{x\in\tilde B_n}f'(x)\right),
\ee
where $\tilde B_n:=(B_n\cap{X}')\cup\{B\}$, and where we used \eqref{mergedEV}. Since $g$ is assumed to be an EDF for $T$, the coefficient in round parentheses is non-negative, so the term is maximized by taking $\lambda\to0$, giving 0. The terms in \eqref{dualobj} also have maximal value 0.
\end{itemize}
We find that the optimizer is $\mu=\nu=\lambda=0$, and the optimal value is 0, proving the theorem.
\end{proof}

\acknowledgments

We would like to thank Juan Pedraza, Andy Svesko, and Brian Swingle for useful discussions. The work of MH is supported by the U.S.\ Department of Energy through award DE-SC0009986. He is also grateful to the Centro de Ciencias de Benasque Pedro Pascual, where part of this work was completed. The work of SRK is supported in part by ISF grant no. 2159/22, by Simons Foundation grant 994296 (Simons Collaboration on Confinement and QCD Strings), by the Minerva foundation with funding from the Federal German Ministry for Education and Research, and by the German Research Foundation through a German-Israeli Project Cooperation (DIP) grant “Holography and the Swampland”. SRK would like to thank Ofer Aharony, Micha Berkooz, Ashoke Sen, and also many people at ICTS Bangalore for useful discussions and for their hospitality during his stay at ICTS. The work of AR is supported by FWO-Vlaanderen project G012222N, the VUB Research Council through the Strategic Research Program High-Energy Physics, and FWO-Vlaanderen through a Senior Postdoctoral Fellowship 1223125N.

\appendix

\section{Comments on generalised entropy} \label{sec:gen_entropy}
In this appendix, we discuss generalised entropy and, in particular, its independence from the UV regulator $\epsilon$, and the generalised entropy of small regions, as this is needed in section~\ref{sec:cutin}. 

The generalised entropy of a region $r$ is~\cite{Bousso_2016} 
\bne \Sgen (r) = \lim_{\epsilon \to 0} \left[\frac{|\eth r|}{4 G_N^{(\epsilon)} } + S_b^{(\epsilon)}(r)  + \text{ higher curvature terms}\right] \label{eq:sgenq} \ene
The bare Newton constant has one-loop and higher counterterms:
\bne \frac{1}{G_N^{(\epsilon)}} = \frac{1}{G_N^{(ren.)}} + f_g \epsilon^{-(d-2)} + \text{ higher loops} \ene
$f_g$ is theory dependent, and has to be negative to cancel the positive $\epsilon$-dependent area-law divergence in $S_b^{(\epsilon)}$. Nonetheless, both $G_N^{(\epsilon)}$ and $G_{N}^{(ren.)}$ are positive, because we do not take $\epsilon < l_p^{(ren.)}$.

It is widely believe that the $\epsilon$-dependent terms in~\eqref{eq:sgenq} cancel so that $\Sgen $ is finite and independent of the UV cutoff:
\bne \Sgen (r) = \left[ \frac{|\eth r |}{4G_{N}^{(ren.)}} + S_b^{(ren.)}(r) + \text{ high curvature terms.} \right]\ene
That the $\epsilon$-dependence in $\Sgen $ cancels has only been shown in specific examples~\cite{Cooperman:2013iqr}; we will assume $\Sgen $ is always $\epsilon$-independent and finite. 

$\Sgen $ has bare area and entropy pieces that are $\epsilon$-dependent, and, as we vary $\epsilon$, any reduction in the area term is compensated by an increase in the entropy term, and vice versa.

The bare and renormalised Newton's constants are approximately the same: since $\epsilon^{d-2} \gg G_N^{(ren.)}$ - we do not take field theory UV cutoff $\epsilon$ smaller than the physical Planck length $(l_p^{(ren.)})^{d-2} = G_N^{(ren.)}$ - we have $G_N^{(\epsilon)} \approx G_N^{(ren.)}$ with corrections suppressed by $l_{p}^{(ren.)}/\epsilon$.

Now we consider an explicit example, to see how $\Sgen $ is $\epsilon$-independent and finite. Let us consider a vacuum state in 2+1d topological CFT and a one-parameter family of disk regions with radius $R$ for which the entanglement entropy is
\bne S^{(\epsilon)}_b (r) = \frac{2\pi R}{\epsilon} - \gamma \ene
where $\gamma > 0$ is the topological entropy, and $S^{(ren.)}  = -\gamma$. The generalised entropy of the disk is finite and $\epsilon$-independent:
\bne \Sgen  (r) = \frac{2\pi R}{G_N^{(ren.)}} -\gamma \ene

\paragraph{Small region limit.} Now we consider the generalised entropy of small regions because, in section~\ref{sec:cibtsc}, in deriving local bounds on $|v|$ in the strict cutoff-independent prescription, we considered $\Sgen  (r)$ with $|\eth r|$ in the regime $\frac{|v|}{|\nabla |v||} \gg |\eth r| \gg \epsilon^{d-2}$. In particular, we derived the bound 
\bne |v(x)| \leq \inf_{\eth r \ni x} \frac{\Sgen (\eth r)}{|\eth r|}. \label{eq:lcnbd}\ene

Assuming that the length scales of $|\eth r|$ are below the length scales of excitations in the bulk state, we can approximate the reduced state on $r$ as the vacuum state. 

For our 2+1d example, for small disks, we get
\bne \frac{\Sgen (r)}{|\eth r|} \approx \frac{1}{4G_{N}^{(\epsilon)}} +\frac{1}{\epsilon} - \frac{\gamma}{2\pi R} = \frac{1}{4G_{N}^{(ren.)}} - \frac{\gamma}{2\pi R}. \label{eq:2dtcf} \ene

The generalised entropy decreases as the disk shrinks, but is positive while $R$ is parametrically larger than $l_p^{(ren.)}$. The renormalised bulk entropy correction to $\Sgen $ is negative and small.

When the UV-finite part of $S_{b}^{(\epsilon)}(r)$ scales with the volume of $r$, $S^{(ren.)}(r) \propto \text{Vol}(r)$, a volume-law entropy, 
\bne \frac{\Sgen (r)}{|\eth r|} \approx \frac{1}{4G_{N}^{(ren.)}} + \# \frac{\text{Vol}(r)}{|\eth r|}. \label{eq:volle} \ene
So, for~\eqref{eq:volle}, the regions that give the tightest bound in~\eqref{eq:lcnbd} are small and spherical. As in~\eqref{eq:2dtcf}, in~\eqref{eq:volle} the ratio decreases as the region shrinks, so the tightest bound on $|v(x)|$ in~\eqref{eq:lcnbd} comes from the smallest regions in our allowed set with $|\eth r| \gg \epsilon^{d-2}$, and from spherical regions because those minimise Vol$(r)/|\eth r|$ for fixed $|\eth r|$.
\bibliographystyle{JHEP}
\bibliography{biblio.bib}

\end{document}